\documentclass[12pt,a4paper,fleqn]{scrartcl}
\RequirePackage{amsthm,amsmath,natbib}
\RequirePackage{amssymb,amstext}
\RequirePackage{hypernat}
\usepackage{graphicx}
\usepackage{enumerate}
\usepackage{bbm}
\usepackage{url}
\usepackage{verbatim}
\usepackage{ bm}
\usepackage{setspace}
\usepackage{tikz}
\usepackage{multicol}
\usepackage{titling}
\usepackage{xpatch}
\usepackage{marginnote}
\usepackage[top=1.3cm, bottom=1.3cm, outer=3cm, inner=2cm, heightrounded, marginparwidth=1.3cm, marginparsep=.5cm]{geometry}

\usepackage[font=footnotesize]{caption}
\DeclareOldFontCommand{\it}{\normalfont\itshape}{\mathit}

\usepackage[pdfpagemode=UseOutlines ,plainpages=false
,hypertexnames=false ,pdfpagelabels ,hyperindex=true,colorlinks=true]{hyperref}
\makeatletter%
	\Hy@breaklinkstrue%
	\makeatother%
\usepackage{color}
\definecolor{darkred}{rgb}{0.6,0.0,0.1}
\definecolor{darkgreen}{rgb}{0,0.5,0}
\definecolor{darkblue}{rgb}{0,0,0.5}
\hypersetup{colorlinks ,linkcolor=darkblue ,filecolor=darkgreen
,urlcolor=darkblue ,citecolor=darkblue}

\RequirePackage{natbib}
\numberwithin{equation}{section}
\newtheorem{theorem}{Theorem}[section]
\newtheorem{corollary}{Corollary}[section]

\newtheorem{example}{Example}[section]

\newtheorem{A}{Assumption}
\newtheorem{lemma}{Lemma}[section]
\newtheorem{remark}{Remark}[section]
\newtheorem{definition}{Definition}
\theoremstyle{definition}

\DeclareMathOperator{\Evtex}{\mathrm{E}}

\def\Var{\mathop{{\mathbb V}ar}\nolimits}%

\newcommand{\Jad}{J^\circ}
\newcommand{\Kad}{K^\circ}
\newcommand{\llVert}{\|}
\newcommand{\rrVert}{\|}
\newcommand{\llvert}{\|}
\newcommand{\rrvert}{\|}

\DeclareMathOperator*{\argmin}{arg\,min}
\newcommand{\1}{{\mathop{\mathbbm 1}}}

\renewcommand{\cite}{\citet}
\def\Cov{\mathop{ {\mathbb C}ov}\nolimits}%
\begin{document}
\title{Adaptive, Rate-Optimal Hypothesis Testing in Nonparametric IV Models\thanks{We thank three anonymous referees and Denis Chetverikov for very constructive comments. Earlier versions have been presented
at numerous workshops and conferences since February 2019. We thank Don Andrews, Tim Armstrong, Tim Christensen, Giovanni Compiani,
 Enno Mammen, Peter Mathe, and other participants at various meetings for helpful comments. The empirical result reported in Section~\protect\ref{subsec:scanner:data} is our own analyses using data provided through the Nielsen Datasets (from The Nielsen Company (US),
   LLC) at the Kilts Center for Marketing Data Center at The University of Chicago Booth School of
Business. Nielsen is not responsible for, had no role in, and was not involved in analyzing and preparing the empirical findings
 reported herein. Breunig gratefully acknowledges the support of the Deutsche Forschungsgemeinschaft (DFG, German Research Foundation) under Germany's Excellence Strategy -- EXC-2047/1 -- 390685813. Chen's research is supported by Cowles Foundation for Research in Economics.}}

 \author{{\textsc{  Christoph Breunig}}\thanks{Department of Economics,
University of Bonn, Adenauerallee 24-26, 53113 Bonn, Germany. Email: \url{	cbreunig@uni-bonn.de}}
  \and{\textsc{Xiaohong Chen}}\thanks{Cowles Foundation for Research in Economics, Yale University, Box 208281, New Haven, CT 06520, USA. Email: \url{xiaohong.chen@yale.edu}}}
\date{First version: August 2018; Latest revision: \today }
\maketitle
\begin{abstract}
\vskip -1cm
\singlespacing
{\small We propose a new adaptive hypothesis test for inequality
(e.g., monotonicity, convexity) and equality (e.g., parametric, semiparametric)
restrictions on a structural function in a nonparametric instrumental variables
(NPIV) model. Our test statistic is based on a modified leave-one-out sample
analog of a quadratic distance between the restricted and unrestricted
sieve two-stage least squares estimators. We provide computationally simple,
data-driven choices of sieve tuning parameters and Bonferroni adjusted
chi-squared critical values. Our test adapts to the unknown smoothness
of alternative functions in the presence of unknown degree of endogeneity
and unknown strength of the instruments. It attains the adaptive minimax
rate of testing in $L^{2}$. That is, the sum of the supremum of type I
error over the composite null and the supremum of type II error over nonparametric
alternative models cannot be minimized by any other tests for NPIV models
of unknown regularities. Confidence sets in $L^{2}$ are obtained by inverting
the adaptive test. Simulations confirm that, across different strength
of instruments and sample sizes, our adaptive test controls size and its
finite-sample power greatly exceeds existing non-adaptive tests for monotonicity
and parametric restrictions in NPIV models. Empirical applications to test
for shape restrictions of differentiated products demand and of Engel curves
are presented.}
\end{abstract}

\section{Introduction}
\label{sec1}

In this paper, we propose computationally simple, optimal hypothesis testing
in a nonparametric instrumental variables (NPIV) model. The maintained
assumption is that there is a nonparametric structural function $h$ satisfying
the NPIV model
%
\begin{align}
\label{npiv:model} \Evtex \bigl[Y-h(X)|W\bigr] = 0, 
\end{align}
where $X$ is a $d_{x}$-dimensional vector of possibly endogenous regressors,
$W$ is a $d_{w}$-dimensional vector of conditional (instrumental) variables
(with $d_{w}\ge d_{x}$), and the joint distribution of $(Y,X,W)$ is unspecified
beyond (\ref{npiv:model}). With the danger of abusing terminology, we call
a function $h$ satisfying model (\ref{npiv:model}) a NPIV function. We
are interested in testing a (composite) null hypothesis that a NPIV function
$h$ satisfies some simplifying economic restrictions, such as parametric
or semiparametric equality restrictions or inequality restrictions (e.g.,
nonnegativity, monotonicity, convexity, supermodularity, quasi-concavity).
Our new test builds on a simple data-driven choice of tuning parameter
that ensures asymptotic size control and non-trivial power uniformly against
a large class of nonparametric alternatives.

Let $L^{2}(X)$ denote the space of square integrable function of $X$. Our
new test is designed to test a composite null hypothesis
$\mathcal{H}_{0}$ that is a closed, convex strict subset of
$L^{2}(X)$ satisfying the NPIV model (\ref{npiv:model}). Before presenting
the theoretical properties of our new test, we derive the
\textit{minimax rate of testing} $r_{n}$ in $L^{2}$, which is the fastest
rate of separation in root-mean squared distance between the null hypothesis
$\mathcal{H}_{0}$ and the class of nonparametric alternative NPIV functions
$\mathcal{H}_{1} (\delta r_{n})$ that enables consistent testing uniformly
over the latter, with the rate $r_{n}$ shrinking to zero as the sample
size $n$ goes to infinity and $\delta >0$ being a finite constant independent
of $n$. We establish the minimax result in two steps: First, we derive,
uniformly over all possible tests, a lower bound for the sum of the supremum
of type I error over $\mathcal{H}_{0}$ and the supremum of type II error
over $\mathcal{H}_{1} (\delta r_{n})$ separated from the null hypothesis
by a rate $r_{n}$. Thus, there exists no other test that provides a better
performance with respect to the sum of those errors. Second, we propose
a test whose sum of the type I and the type II errors is bounded from above
(by the nominal level) at the same separation rate $r_{n}$. This test is
based on a modified leave-one-out sample analog of a quadratic distance
between the restricted and unrestricted sieve NPIV (i.e., sieve two-stage
least squares) estimators of $h$. The test is shown to attain the minimax
rate of testing $r_{n}$ when the sieve dimension is chosen optimally according
to the smoothness of the nonparametric alternative functions and the degree
of the ill-posedness of the NPIV model (that depends on the smoothness
of the conditional density of $X$ given $W$). This test is called minimax
rate-optimal (with known model regularities).

In practice, the smoothness of the nonparametric alternative functions
and the degree of the ill-posedness of the NPIV model are both unknown.
Our new test is a data-driven version of the minimax rate-optimal test
that adapts to the unknown smoothness of the nonparametric alternative
NPIV functions in the presence of the unknown degree of the ill-posedness.
Our test rejects the null hypothesis as soon as there is a sieve dimension
(say the smallest sieve dimension) in an estimated index set such that
the corresponding normalized leave-one-out quadratic distance estimator
exceeds 1; and fails to reject the null otherwise. The normalization builds
on Bonferroni corrected chi-squared critical values. The simple Bonferroni
correction is computed using the cardinality of the estimated index set,
which is in turn determined by a random exponential scan (RES) procedure
that automatically takes into account the unknown degree of ill-posedness.

We show that our new test attains the minimax rate of testing in
$L^{2}$ for severely ill-posed NPIV models, and is within a
$\sqrt{\log \log (n)}$ multiplicative factor of the minimax rate of testing
for mildly ill-posed NPIV models. This extra $\sqrt{\log \log (n)}$ term
is the necessary price to pay for adaptivity to unknown smoothness of nonparametric
alternative functions.\footnote{This is needed even for adaptive minimax
hypothesis testing in nonparametric regressions (without endogeneity);
see \citet{spokoiny1996}, \citet{horowitz2001}, and
\citet{guerre2005}.} A key technical part to establish our adaptive minimax
rate of testing in $L^{2}$ is to derive a sharp upper bound on the convergence
rate of a leave-one-out sieve estimator of a quadratic functional of a
NPIV function, which is proved using an exponential inequality for U-statistics
with increasing dimensions. We show that our adaptive test has asymptotic
size control under a composite null by deriving a tight, slowly divergent
lower bound for Bonferroni corrected chi-squared critical value. By inverting
our adaptive tests, we obtain $L^{2}$ confidence sets on restricted NPIV
functions. These confidence sets are free of additional choices of tuning
parameters. The adaptive minimax rate of testing determines the
$L^{2}$ radius of the confidence sets.

Our adaptive minimax $L^{2}$ rate of testing decreases to zero strictly
faster than the optimal $L^{2}$ rate of estimation (with known smoothness)
for mildly ill-posed NPIV models, and coincides with the optimal
$L^{2}$ rate of estimation for severely ill-posed NPIV models. In the existing
literature on testing for parametric, semiparametric, or shape NPIV restrictions
against nonparametric alternatives, all of the non-adaptive tests achieve
their asymptotic size controls by choosing some deterministic tuning parameters
such that the $L^{2}$ estimation bias for $h$ is of a smaller order than
the $L^{2}$ standard deviation (aka, under-smoothing), which leads to a
$L^{2}$ separation rate of testing shrinking to zero strictly slower than
the optimal $L^{2}$ rate of estimation, and hence strictly slower than
our adaptive minimax $L^{2}$ rate of testing for both mildly and severely
ill-posed NPIV models. In particular, among all of the existing NPIV tests
that have asymptotic size controls, our new adaptive test is asymptotically
more powerful, uniformly over a larger class of nonparametric alternatives.

In Monte Carlo simulations, we analyze the finite-sample properties of
our adaptive test for the null of monotonicity or a parametric hypothesis
using various simulation designs from others' work. The simulations reveal
the following patterns of our adaptive test in comparison to recent non-adaptive
tests: First, while the competing tests can be over-sized at the boundary
of the null hypothesis, our test delivers adequate size control under different
composite null hypotheses, across different sample sizes, and for varying
strengths of instruments. Second, our test is as powerful as the competing
tests when alternative functions are relatively simple, and is more powerful
when alternatives are more nonlinear/complex. The great power gains of
our adaptive test are present even for relatively weak strength of instruments
or small sample sizes. These findings highlight the importance of our data-driven
choice of the sieve dimension to simultaneously ensure size control and
powerful performance uniformly against a larger class of nonparametric
alternative NPIV functions. Finally, unlike many NPIV tests using bootstrapped
critical values, our powerful adaptive test uses simple Bonferroni corrected
chi-squared critical values and hence is fast to compute.

We present two empirical applications of our adaptive test. The first is
to test the connected substitutes shape restrictions in demand for differential
products using market level data (e.g., \citet{berry2014}). The second
is to test for monotonicity, convexity, and parametric forms in Engel curves
(e.g., \citet{BCK07econometrica}).

There is a growing number of papers on testing equality and inequality
(shape) restrictions in NPIV type models. See, for example,
\citet{Horowitz2006}, \citet{Santos12}, \citet{Breunig2015},
\citet{chen2013}, \citet{chernozhukov2015}, \citet{zhu2020},
\citet{fang2019} and references therein.\footnote{There are also papers
on NPIV estimation by directly imposing shape restrictions; see, for example,
\citet{horowitzlee2012}, \citet{blundell2017}, \citet{CW2017}, and
\citet{Freyberger2019}. See \citet{chetverikov2018} for a review on shape
restrictions and \citet{chetverikov2019testing} for adaptive kernel testing
for monotonicity of a regression without endogeneity.} Most of these papers
assume that some non-random sequences of key tuning (regularization) parameters
satisfy some theoretical rate conditions with known smoothness of NPIV
functions. None of the published work achieves the adaptive minimax
$L^{2}$ rate of testing for NPIV models. Our paper makes an important contribution
by providing the first data-driven choice of a key tuning parameter that
leads to a new minimax rate-adaptive and powerful test for equality and
inequality (shape) restrictions in NPIV models. Our paper also complements
a concurrent work by \citet{chen2021}, which constructs honest and near-adaptive
uniform confidence bands for a NPIV function and its partial derivatives
using a bootstrapped Lepski's procedure (in sup-norm).

The rest of the paper is as follows. Section~\ref{sec:test:def} describes
our new hypothesis test. Section~\ref{sec:minimax} establishes the oracle
minimax optimal rate of testing. Section~\ref{sec:adapt:test} shows that
this minimax optimal rate is attained (within a
$\sqrt{\log \log (n)}$ term) by our new test. Section~\ref{sec:mc} presents
simulation studies and Section~\ref{sec:empirical} provides empirical illustrations.
Appendices~\ref{sec:appendix} and \ref{appendix:adapt:ST} present proofs
of Theorems \ref{thm:minimax:test:lower}, \ref{thm:test:upper},
\ref{thm:adapt:test}, and \ref{thm:adapt:est:test}. The Supplemental Material includes Appendix~C for additional
 simulation results, Appendix~D for proofs of Corollaries
\ref{coro:CB} and \ref{thm:diam}, and Appendix~E for additional lemmas and their proofs.

\paragraph{Basic Notation}
For a random variable $X$, we let $L^{2}(X)$ denote the Hilbert space of
real-valued measurable functions $\phi $ of $X$ with finite second moment,
with the norm $\|\phi \|_{L^{2}(X)}:=\sqrt{\Evtex [\phi ^{2}(X)]}$ and the
inner product $\langle \cdot ,\cdot \rangle _{X}$. Let
$\|\phi \|_{\infty}:=\sup_{x} |\phi (x)|$ be the sup-norm and
$L^{\infty}=\{\phi :\|\phi \|_{\infty}<\infty \}$. For a matrix $M$, let
$M'$ be its transpose and $M^{-}$ be its generalized inverse. For a
$J\times J$ matrix $M=(M_{jl})_{1\leq j,l\leq J}$, we define its Frobenius
norm as $\|M\|_{F}=\sqrt{\sum_{j,l=1}^{J}M_{jl}^{2}}$. Let
$\|\cdot \|$ be the Euclidean norm when applied to a vector and the operator
norm induced by the Euclidean norm when applied to a matrix. For sequences
of positive real numbers $\{a_{n}\}$ and $\{b_{n}\}$, we use the notation
$a_{n} \lesssim b_{n}$ if
$\limsup_{n\to \infty}a_{n}/b_{n}<\infty $, and $a_{n}\sim b_{n}$ if
$a_{n}\lesssim b_{n}$ and $b_{n}\lesssim a_{n}$.

\section{Preview of the Adaptive Hypothesis Testing}
\label{sec:test:def}

We first introduce the null and the alternative hypotheses as well as the
concept of minimax rate of testing in Section~\ref{subsec:hyp}. We then
describe our new rate-adaptive test for NPIV type models in Section~\ref{subsec:sst}.

\subsection{Null Hypotheses and Nonparametric Alternatives}
\label{subsec:hyp}

Let $\mathcal H$ denote a closed subset of $L^{2}(X)$ that captures some
unknown degree of smoothness. Let
$\{(Y_{i},X_{i},W_{i})\}_{i=1}^{n}$ denote a random sample from the distribution
$\mathrm{P}_{h}$ of $(Y,X,W)$ satisfying the NPIV model (\ref{NPIV:maintained}):
%
\begin{align}
\label{NPIV:maintained} Y=h(X)+U, \quad \text{where } \Evtex _{h}[U|W]=0 \text{ and } h
\in \mathcal H. 
\end{align}
Here, $\Evtex _{h}$ denotes the (conditional) expectation under
$\mathrm{P}_{h}$. In this paper, we assume that the joint distribution
of $(X,W)$ does not depend on $h\in \mathcal H$ and that the conditional
density of $X$ given $W$ is continuous on its support. The conditional
expectation operator $T: L^{2}(X) \mapsto L^{2}(W)$ given by
$T h(w) := \Evtex [h(X)|W= w]$ is uniquely defined by the conditional density
of $X$ given $W$ and hence does not depend on $h$. We can then equivalently
express the NPIV model (\ref{NPIV:maintained}) as
$\Evtex _{h}[Y|W]=(T h)(W)$ for $h\in \mathcal H$. For ease of presentation,
we mainly consider a nonparametric class of functions as the maintained
hypothesis $\mathcal H$. Nevertheless, our theoretical results allow for
semiparametric structures $\mathcal H$ as well (see Section~\ref{subsec:unknown:null}).

Let $\mathcal H_{0}$ denote the null class of functions in
$\mathcal H$ that satisfies a conjectured restriction in (\ref{NPIV:maintained}).
We assume that $\mathcal H_{0}$ is a
\emph{nonempty, closed and convex, strict subset} of $\mathcal H$. For any
$h\in \mathcal H$, there is a unique element
$\Pi _{\mathcal H_0}h \in \mathcal H_{0}$ such that
$\inf_{\phi \in \mathcal H_{0}}\|h-\phi \|_{L^{2}(X)}=\|h-\Pi _{
\mathcal H_0}h\|_{L^{2}(X)}$ (by the Hilbert projection theorem). In addition
to a simple null $\mathcal H_{0}=\{h_{0}\}$ (with a known function
$h_{0}\in \mathcal H$), we allow for general parametric, semi/nonparametric
equality and inequality composite null restrictions. We present two examples
of composite null restrictions below (see Section~\ref{subsec:unknown:null} for additional examples).

\begin{example}[Nonparametric shape restrictions]%
\label{ex:shape}
$\mathcal H_{0}$ can be a closed convex subset of $\mathcal H$ determined
by inequality restrictions such that
$\mathcal H_{0}= \{h\in \mathcal H:  \partial ^{l} h\geq 0 \}$,
where $\partial ^{l} h$ denotes the $l$th partial derivative of $h$ with
respect to components of $x$. This allows for testing nonnegativity ($l=0$),
monotonicity ($l=1$), or convexity ($l=2$). We can also test for supermodularity
restrictions on NPIV functions corresponding to
$\mathcal H_{0}= \{h\in \mathcal H:  \partial ^{2} h/(\partial x_{1}
\partial x_{2})\geq 0 \}$. Our framework also allows for testing these
restricted function classes simultaneously since intersections of these
are again closed convex subsets of $\mathcal H$.
\end{example}

\begin{example}[Semiparametric restrictions]%
\label{ex:semi}
Let $F(\cdot ;\theta ,g)$ be a known function up to unknown
$(\theta ,g)$ and
$\mathcal H_{0}= \{h\in \mathcal H:  h(\cdot )=F(\cdot ;\theta ,g)
\text{ for some } \theta \in \Theta \text{ and } g\in \mathcal G \}$,
for a finite-dimensional, convex compact parameter space $\Theta $ and
a nonparametric closed and convex function class $\mathcal G$. The known
function $F(\cdot ;\theta ,g)$ could be nonlinear in $\theta $ but is assumed
to be linear (or affine) in $g$ and consequently, $\mathcal H_{0}$ is a
closed convex subset of $\mathcal H$. Examples include null hypotheses
of parametric form, or partially linear form, or partially parametric additive
form.
\end{example}

To analyze the power of any test of the null class $\mathcal H_{0}$ against
nonparametric alternatives, we require some separation in
$\|\cdot \|_{L^{2}(X)}$- distance between the null and the class of nonparametric
alternatives for all $h\in \mathcal H$. Below, we use the notation
$\|h-\mathcal H_{0}\|_{L^{2}(X)} := \inf_{\phi \in \mathcal H_{0}}\|h-
\phi \|_{L^{2}(X)}=\|h-\Pi _{\mathcal H_0}h\|_{L^{2}(X)}$. We consider
the class of nonparametric alternatives
\begin{align*}
\mathcal H_{1}(\delta r_{n}): = \bigl\{ h\in \mathcal H:
\llVert h- \mathcal H_{0} \rrVert _{L^{2}(X)}\geq \delta
r_{n} \bigr\}
\end{align*}
for some constant $\delta >0$ and a
\textit{separation rate of testing} $r_{n}>0$ that decreases to zero as
the sample size $n$ goes to infinity. We say that a test statistic
$\mathtt {T}_{n}$ with values in $\{0, 1\}$ is
\emph{consistent uniformly} over $\mathcal H_{1}(\delta r_{n})$ if
$\sup_{h\in \mathcal H_{1}(\delta r_{n})} \mathrm{P}_{h}(\mathtt {T}_{n}=0)=o(1)$.

\noindent\begin{minipage}{0.4\linewidth}
In Section~\ref{sec:minimax}, we establish the
\textit{minimax (separation) rate of testing} $r_{n}$ in the sense of
\citet{ingster1993}: We propose a test that minimizes the sum of the supremum
of the type I error over $\mathcal{H}_{0}$ and the supremum of the type
II error over $\mathcal{H}_{1}(\delta r_{n})$. Moreover, we show that the
sum of both errors cannot be improved by any other test.
\end{minipage}
  \begin{minipage}{0.6\linewidth}
  \vskip .3cm
  \begin{center}
\begin{tikzpicture}
  \shade[ball color = orange!40, opacity = 0.4] (0,0) circle (1cm);
  \draw (0,0) circle (1cm);
   \draw (0,0) circle (2cm);
  \draw (-1,0) arc (180:360:1 and 0.6);
  \draw[dashed] (1,0) arc (0:180:1 and 0.6);
  \fill[fill=black] (0,0) circle (0pt) node{{\small $\mathcal H_0$}};
  \draw[dashed] (1,0 ) -- node[above]{\footnotesize $\delta r_n$} (2,0);
   \shade [inner color=red, outer color=white, even odd rule] circle (3) circle (2);
   \fill[fill=black] (2.9,0) circle (0pt) node{{\footnotesize $\quad \mathcal H_1(\delta r_n)$}};
\end{tikzpicture}
 \end{center}
\end{minipage}
%

\begin{definition}
\label{defn1}
A separation rate of testing $r_{n}$ is called the minimax (separation)
rate of testing if the following two requirements are met for every level
$\alpha \in (0,1)$:%
\begin{enumerate}[(ii)]
\item[(i)] For some constant $\delta _{*}:=\delta _{*}(\alpha )>0$, it holds that
%
\begin{align}
\label{def:lower:rate} \liminf_{n\to \infty}\inf_{\mathtt {T}_{n}} \Bigl
\{\sup_{h\in
\mathcal H_{0}} \mathrm{P}_{h}(\mathtt
{T}_{n}=1)+ \sup_{h\in
\mathcal H_{1}(\delta _{*} r_{n})} \mathrm{P}_{h}(
\mathtt {T}_{n}=0) \Bigr\}\geq \alpha , 
\end{align}
where $\inf_{\mathtt {T}_{n}}$ is the infimum over all statistics with
values in $\{0, 1\}$.

\item[(ii)] There exists
a test statistic $\mathtt {T}_{n}:=\mathtt {T}_{n}(\alpha )$ with values
in $\{0, 1\}$ such that
%
\begin{align}
\label{def:upper:rate} \limsup_{n\to \infty} \Bigl\{\sup_{h\in \mathcal H_{0}}
\mathrm{P}_{h}( \mathtt {T}_{n}=1)+ \sup
_{h\in \mathcal H_{1}(\delta ^{*} r_{n})} \mathrm{P}_{h}(\mathtt {T}_{n}=0)
\Bigr\}\leq \alpha 
\end{align}
for some constant $\delta ^{*}>0$.
\end{enumerate}\end{definition}
We refer to Part (i) as the lower bound and Part (ii) as the upper bound,
and the test statistic $\mathtt {T}_{n}:=\mathtt {T}_{n}(\alpha )$ in Part
(ii) attaining the matching lower and upper bound as an optimal test. We
use $r_{n}^{*}$ to denote the minimax (separation) rate of testing as the
matching lower and upper bound.

In Section~\ref{sec:minimax}, we first establish a minimax rate of testing
$r_{n}^{*}$ assuming the knowledge of the smoothness of alternative NPIV
functions $h\in \mathcal H$ and the inversion property of the conditional
expectation operator $T: L^{2}(X) \mapsto L^{2}(W)$. Both are unknown in
practice. The minimax rate $r_{n}^{*}$ is attained by a sieve test statistic
using an optimal choice of sieve dimension (a tuning parameter) that depends
on these unknown objects, and hence is infeasible. In Section~\ref{sec:adapt:test}, we provide a data-driven modification of the optimal
sieve test, that is, a feasible testing procedure that adapts to the unknown
smoothness of the unrestricted NPIV function $h\in \mathcal H$ in the presence
of unknown smoothing properties of the inverse of the operator $T$. Precisely,
we propose a feasible test statistic $\widehat{\mathtt {T}}_{n}$ with data-driven
tuning parameters in Section~\ref{subsec:sst}. We show that
$\widehat{\mathtt {T}}_{n}$ attains the minimax rate of testing
$r_{n}^{*}$ within a $\sqrt{\log \log (n)}$ multiplicative factor, has
asymptotic size control over the composite null, and is consistent uniformly
over the class of nonparametric alternatives in Theorem~\ref{thm:adapt:est:test}. We call our test
$\widehat{\mathtt {T}}_{n}$ \textit{adaptive} and
\textit{rate-optimal} (or sometimes simply adaptive).

\subsection{Our Adaptive Test}
\label{subsec:sst}

Our test is based on a consistent estimate of the quadratic distance,
$\|h-\Pi _{\mathcal H_0}h\|^{2}_{L^{2}(X)}=\|h-\mathcal H_{0}\|^{2}_{L^{2}(X)}$,
between the NPIV function $h\in \mathcal H$ and its projection
$\Pi _{\mathcal H_0}h$ onto $\mathcal H_{0}$ under the
$\|\cdot \|_{L^{2}(X)}$. We first introduce some notation. Let
$\{\psi _{j}\}_{j=1}^{\infty}$ and $\{b_{k}\}_{k=1}^{\infty}$ be complete
basis functions for the Hilbert spaces $L^{2}(X)$ and $L^{2}(W)$, respectively.
Let $\psi ^{J}(\cdot )$ and $b^{K}(\cdot )$ be vectors of basis functions
of dimensions $J$ and $K=K(J)\geq J$, respectively. These can be cosine,
power series, spline, or wavelet basis functions. Let
$G=\Evtex [\psi ^{J}(X)\psi ^{J}(X)']$,
$G_{b}=\Evtex [b^{K(J)}(W)b^{K(J)}(W)']$, and
$S=\Evtex [b^{K(J)}(W)\psi ^{J}(X)']$. We assume that $G$, $G_{b}$, and
$S'G_{b}^{-1}S$ have full ranks. Then the $J\times K(J)$ matrix
$ A= G^{1/2} [S'G_{b}^{-1}S]^{-1}S'G_{b}^{-1}$ is well defined. Let
$\Psi _{J}$ denote the closed linear subspace of $L^{2}(X)$ spanned by
$\{\psi _{1},\dots ,\psi _{J}\}$. We define a population 2SLS projection
of $h\in L^{2}(X)$ onto the sieve space $\Psi _{J}$ as
\begin{align*}
Q_{J} h(\cdot ):=\psi ^{J}(\cdot )'G^{-1/2}A
\Evtex \bigl[b^{K}(W)h(X)\bigr] .
\end{align*}
For any NPIV function $h\in \mathcal H$ in (\ref{NPIV:maintained}), we
have $Q_{J} h(\cdot )=\psi ^{J}(\cdot )'G^{-1/2}A\Evtex _{h}[b^{K}(W)Y]$, and
%
\begin{align}
\label{QJ2} \bigl\llVert Q_{J} (h-\Pi _{\mathcal H_{0}}h) \bigr
\rrVert ^{2}_{L^{2}(X)}= \bigl\llVert A\Evtex _{h}
\bigl[b^{K}(W) \bigl(Y-\Pi _{\mathcal H_{0}} h(X)\bigr) \bigr] \bigr
\rrVert ^{2}, 
\end{align}
which approximates $\|h-\Pi _{\mathcal H_{0}}h\|^{2}_{L^{2}(X)}$ well as
$J$ grows large (see Lemma~\ref{lemma:bias}).

For each sieve dimension $J$, we construct a test based on an estimated
quadratic distance
$\|Q_{J} (h-\Pi _{\mathcal H_{0}}h)\|^{2}_{L^{2}(X)}$ between the unrestricted
and restricted NPIV estimators of a function $h$ satisfying (\ref{NPIV:maintained}).
Let $\Psi = (\psi ^{J}(X_{1}),\dots ,\psi ^{J}(X_{n}))'$,
$B = (b^{K}(W_{1}),\dots ,b^{K}(W_{n}))'$, $P_{B}= B(B'B)^{-} B'$, and
$\widehat A= \sqrt n(\Psi ' \Psi )^{1/2}[\Psi 'P_{B}\Psi ]^{-} \Psi ' B(B'B)^{-}$.
Let $\mathrm Y=(Y_{1},\ldots , Y_{n})'$. Our unrestricted sieve NPIV estimator
solves a sample 2SLS problem (\citet{BCK07econometrica}):
%
\begin{align}
\label{def:est:h} \widehat h_{J} &=\argmin _{\phi \in \Psi _{J}} \sum
_{1\leq i, i'
\leq n} \bigl(Y_{i}-\phi (X_{i})
\bigr)b^{K(J)}(W_{i})' \widehat A'
\widehat Ab^{K(J)}(W_{i'}) \bigl(Y_{i'}- \phi
(X_{i'}) \bigr) \notag
\\
&=\psi ^{J}(\cdot )'\bigl[\Psi 'P_{B}
\Psi \bigr]^{-} \Psi ' P_{B} \mathrm Y .
\end{align}
Let $\mathcal H_{0,J}$ denote a nonempty, closed and convex, finite-dimensional
subset of $\mathcal H_{0}$. A restricted NPIV estimator for
$\Pi _{\mathcal H_{0}}h \in \mathcal H_{0}$ is given by
%
\begin{align}
\label{def:est:h^R} \widehat h_{J}^{\text{\textsc r}} =\argmin _{\phi \in \mathcal H_{0,J}}
\sum_{1\leq i, i'\leq n} \bigl(Y_{i}-\phi
(X_{i}) \bigr)b^{K(J)}(W_{i})'
\widehat A'\widehat Ab^{K(J)}(W_{i'})
\bigl(Y_{i'}- \phi (X_{i'}) \bigr). 
\end{align}
The choice of $\mathcal H_{0,J}$ is allowed to depend on the structure
of the null class of NPIV functions $\mathcal H_{0}$. For a general nonparametric
or a semi-nonparametric composite null hypothesis,
$\mathcal H_{0,J}$ depends on sieve dimension $J$ and grows dense in
$\mathcal H_{0}$ as the sample size increases. For instance, we let
$\mathcal H_{0,J}=\Psi _{J}\cap \mathcal H_{0}$ under a nonparametric composite
null whenever $\Psi _{J}\cap \mathcal H_{0}\neq \emptyset $ (which holds
for the nonparametric inequality restrictions in Example~\ref{ex:shape}). We can also let $\mathcal H_{0,J}=\mathcal H_{0}$ under
a simple null ($\mathcal H_{0}=\{h_{0}\}$ for a known function
$h_{0}$), or under a parametric composite null ($\mathcal H_{0}=\{F(
\cdot ;\theta ),  \theta \in \Theta \}$ for some known mapping $F$).

For each sieve dimension $J$, we compute a $J$- dependent test statistic
$n \widehat{D}_{J} /\widehat{V}_{J}$, which is a standardized, centered
(or leave-one-out) version of the sample analog of (\ref{QJ2}):
%
\begin{eqnarray}
\label{whDJ} \widehat{D}_{J} &=& \frac{2}{n(n-1)}\sum
_{1\leq i< i'\leq n} \bigl(Y_{i}- \widehat h_{J}^{\text{\textsc r}}(X_{i})
\bigr)b^{K(J)}(W_{i})' \widehat A'
\widehat Ab^{K(J)}(W_{i'}) \bigl(Y_{i'}- \widehat
h_{J}^{\text{\textsc r}}(X_{i'}) \bigr), 
\\
\widehat{V}_{J}&=& \Biggl\llVert \widehat A \Biggl(\frac{1}{n}\sum
_{i=1}^{n} \bigl(Y_{i}- \widehat
h_{J}(X_{i})\bigr)^{2} b^{K(J)}(W_{i})b^{K(J)}(W_{i})'
\Biggr) \widehat A' \Biggr\rrVert _{F} , \label{eq2.8}
\end{eqnarray}
where $\widehat{ V}_{J}$ estimates the population normalization factor
%
\begin{equation}
\label{v2J} V_{J} = \bigl\llVert A \Evtex _{h}\bigl[
\bigl(Y-h(X)\bigr)^{2} b^{K(J)}(W)b^{K(J)}(W)'
\bigr] A' \bigr\rrVert _{F} , 
\end{equation}
which is the variance of
$\frac{2}{n(n-1)}\sum_{1\leq i< i'\leq n} (Y_{i}- h(X_{i}) )b^{K(J)}(W_{i})'
A' Ab^{K(J)}(W_{i'}) (Y_{i'}- h(X_{i'}) )$.

We compute our adaptive test for the null hypothesis
$\mathcal H_{0}$ against nonparametric alternatives in three simple steps.
%
\paragraph{Step 1.}%
Compute a \textit{random exponential scan} (RES) index set:
%
\begin{align}
\label{def:index_set} \widehat{\mathcal I}_{n}:= \bigl\{J\leq \widehat
J_{\max}: J= \underline J2^{j}\text{ where } j=0,1,\dots
,j_{\max} \bigr\}, 
\end{align}
where $\underline J:=\lfloor \sqrt{\log \log n}\rfloor $,
$j_{\max}:=\lceil \log _{2}(n^{1/3}/\underline J)\rceil $, and the empirical
upper bound
%
\begin{equation}
\label{def:J_max} \widehat J_{\max}:=\min \bigl\{J>\underline J: 1.5
\bigl[\zeta (J)\bigr]^{2} \sqrt{ (\log J)/n}\geq \widehat
s_{J} \bigr\}, 
\end{equation}
where $\widehat s_{J}$ is the minimal singular value of
$(B'B)^{-1/2}B'\Psi (\Psi ' \Psi )^{-1/2}$, and $\zeta (J)=\sqrt{J}$ for
spline, wavelet, or trigonometric sieve basis, and $\zeta (J)=J$ for power
series.

\paragraph{Step 2.}%
Let $\#(\widehat{\mathcal I}_{n})$ be the cardinality of the RES index
set. For a nominal level $\alpha \in (0,1)$, we compute a Bonferroni corrected
chi-squared critical value as
\begin{equation*}
\widehat \eta _{J} (\alpha ):= \bigl(q \bigl(\alpha /\#( \widehat{
\mathcal I}_{n}), J \bigr)-J \bigr)/\sqrt{J},
\end{equation*}
where $q (a, J)$ is the $100(1-a)\%$-quantile of the standard chi-squared
distribution with $J$ degrees of freedom.

\paragraph{Step 3.}%
Let
$\widehat {\mathcal{W}}_{J}(\alpha ) :=
n \widehat{D}_{J}/(\widehat \eta _{J}(\alpha )  \widehat{V}_{J})$
for all $J\in \widehat{\mathcal I}_{n}$. Compute the test
%
\begin{align}
\label{def:test:unknown} \widehat {\mathtt{T}}_{n}&:= \mathbbm{1} \bigl\{
\text{there exists } J\in \widehat{\mathcal I}_{n} \text{ such that }
\widehat{\mathcal W}_{J}(\alpha )>1 \bigr\}, 
\end{align}
where $\mathbbm{1}\{\cdot \}$ is the indicator function. Under
the nominal level $\alpha \in (0,1)$, $\widehat {\mathtt{T}}_{n}=1$ indicates
rejection of the null $\mathcal H_{0}$ and
$\widehat {\mathtt{T}}_{n}=0$ indicates a failure to reject the null.

\begin{remark}[Index set for $J$]
\label{rem2.1}
The RES index set $\widehat{\mathcal I}_{n}$ determines a collection of
candidate sieve dimensions $J$ for our test. The data-dependent upper bound
$\widehat J_{\max}$ ensures that the cardinality of the index set
$\widehat{\mathcal I}_{n}$ is not too large relative to the sampling variability
of unrestricted sieve NPIV estimation, but that $\widehat J_{\max}$ still diverges in probability at a rate much
faster than that of $\underline J$. Therefore, the index set is large enough to detect a large collection
of alternative NPIV functions.
In simulations and empirical applications where we have
used quadratic B-splines, we find that our adaptive test results are not
sensitive to the choice of the constant 1.5, and that the lower bound
$\underline J$ is not binding in most cases. For other sieve bases, one
might need to use a different constant to ensure a sufficiently large index
set.
\end{remark}

\begin{remark}[Choice of $K$]%
\label{remark:K}
Our adaptive testing procedure lets $K:=K(J)$ be any deterministic function
of $J$ satisfying
$\lim_{J\rightarrow \infty} \frac{K(J)}{J}=c\in [1,\infty )$, and simply
optimizes over $J \in \widehat{\mathcal I}_{n}$. Our theoretical results,
including the asymptotic size control, are valid for any finite constant
$c\geq 1$. In simulation studies and real data applications, we let
$K(J)=cJ$. Since a larger $c>1$ implies more over-identification restrictions
in a sieve NPIV (2SLS) regression, we expect that a larger $c>1$ would
lead to better power in finite samples. We have tried
$K(J)\in \{2J,  4J,  8J\}$ in simulation studies in various designs. The
simulation results show that (i) our adaptive test indeed has size control
regardless of sample sizes, strength of instruments, and even when
$K(J)=8J$; (ii) while our adaptive test with $K(J)=8J$ has better empirical
power for small sample sizes and weak instruments, the empirical power curves
are not sensitive to the choice of $K$ for moderate to large sample sizes
or strong instruments. These findings are consistent with our theory that
the choice of $J$ is the key tuning parameter in minimax rate-optimal hypothesis
testing in NPIV models using sieve methods.
\end{remark}

\begin{remark}[Critical values]
\label{rem2.3}
A remarkable feature of our adaptive test is that it provides asymptotic
size control for inequality restrictions without restricting the degree
of freedom of the Bonferroni corrected chi-squared critical values to the
number of binding constraints. This is established by the observation that
our Bonferroni corrected critical values
$\widehat \eta _{J} (\alpha )$ diverge slowly as $n\to \infty $ with probability
approaching 1; see Lemma~\ref{lemma:eta:bounds}. This, along with the cardinality
of $\widehat{\mathcal I}_{n}$ not becoming too large by construction, and
complexity restrictions on the composite null hypotheses, enables us to
establish asymptotic size control.
\end{remark}

\section{The Minimax Rate of Testing}
\label{sec:minimax}

This section derives the minimax separation rate of hypothesis testing
in NPIV models, when $\mathcal H$ is a relative compact subset of
$L^{2}(X)$. For simplicity, we assume in this paper that
$\mathcal H$ is a standard Sobolev ellipsoid of smoothness $p>0$, which
can be expressed as
\begin{equation*}
\mathcal H= \Biggl\{ h\in L^{2}(X): \sum_{j=1}^{\infty}
j^{2p/d_{x}} \langle h,\widetilde\psi _{j}\rangle
_{X}^{2}\leq C_{\mathcal H}^{2} \Biggr\},
\quad \text{for a finite constant } C_{\mathcal H}>0,
\end{equation*}
where $\{\widetilde \psi _{j}\}_{j=1}^{\infty}$ is the orthonormal basis
for $L^{2}(X)$ that is constructed from the basis
$\{\psi _{j}\}_{j=1}^{\infty}$ (using the Gram--Schmidt procedure). Assuming
the smoothness $p$ is known, we first establish a lower bound for the
$L^{2}$ rate of testing in Section~\ref{subsec:lower_bound}, and then
show that the lower bound can be achieved by a sieve test if the sieve
dimension $J$ can be chosen optimally in Section~\ref{subsec:upper_bound}.

\subsection{The Lower Bound}
\label{subsec:lower_bound}

Before we state the lower bound for the rate of testing, we introduce the
main assumptions.

\begin{A}%
\label{A:LB}
(i)
$\inf_{w\in \mathcal W}\inf_{h\in \mathcal H}\Var _{h}(Y-h(X)|W=w)
\geq \underline\sigma ^{2}>0$; (ii) for any $h\in \mathcal H$,
$Th=0$ implies that $\|h\|_{L^{2}(X)}^{2} =0$; (iii) the densities of
$X$ and $W$ are uniformly bounded below from zero and from above on their
supports, which are Cartesian product of bounded intervals; (iv) there
are a finite constant $C>0$ and a positive decreasing function
$\nu $ with $\nu _{j}:=\nu (j)$ such that
$\|T h\|_{L^{2}(W )}^{2} \leq C \sum_{j\geq 1} \nu _{j}^{2}\langle h,
\widetilde \psi _{j}\rangle _{X}^{2}$ for all $h\in \mathcal H$.
\end{A}

Assumptions \ref{A:LB}(i), (ii), (iii) are basic regularity conditions imposed
in the paper. Assumption~\ref{A:LB}(iv) specifies the smoothing property
of the conditional expectation operator $T$ relative to the basis
$\{\widetilde \psi _{j}\}$. The smoother $T$ is (i.e., the smoother the
conditional density of $X$ given $W$ is), the faster the sequence
${\nu _{j}}$ in Assumption~\ref{A:LB}(iv) decreases to zero, and the harder
it is to detect properties of the NPIV function in the $L^{2}(X)$ metric.

In this paper, we call a decreasing sequence $\{\nu _{j}\}$
\textit{regularly varying} if
$\nu _{J}^{-4}J \lesssim \sum_{j=1}^{J} \nu _{j}^{-4}$. The regularly
varying sequence $\{\nu _{j}\}$ allows for very broad decreasing patterns,
and includes two leading special cases: (1) \textit{mildly ill-posed} case
where $\nu _{j}=j^{-a/d_{x}}$ for some $a>0$; and (2)
\textit{severely ill-posed} case where
$\nu _{j}=\exp (-j^{a/d_{x}}/2)$ for some $a>0$.

\begin{theorem}%
\label{thm:minimax:test:lower}
Let Assumption~\ref{A:LB} hold. Consider testing a closed convex null
$\mathcal H_{0}$ versus
$\mathcal H_{1}(\delta r_{n})=  \{ h\in \mathcal H:\|h-\mathcal H_{0}
\|_{L^{2}(X)}\geq \delta r_{n}  \}$ for some constant
$\delta >0$ and a separation rate
%
\begin{equation}
\label{r_star_lower} r_{n}=n^{-1/2} \Biggl(\sum
_{j=1}^{J_{*}} \nu _{j}^{-4}
\Biggr)^{1/4}, \quad \text{with }J_{*}:=\max \Biggl\{J:
n^{-1/2} \Biggl(\sum_{j=1}^{J} \nu
_{j}^{-4} j^{4p/d_{x}} \Biggr)^{1/4} \leq
C_{\mathcal H} \Biggr\}. 
\end{equation}
Then: for any $\alpha \in (0,1)$, there exists a constant
$\delta _{*}:=\delta _{*}(\alpha )>0$ such that
\begin{align*}
\liminf_{n\to \infty}\inf_{\mathtt {T}_{n}} \Bigl\{\sup
_{h\in
\mathcal H_{0}} \mathrm{P}_{h}(\mathtt {T}_{n}=1)+
\sup_{h\in
\mathcal H_{1}(\delta _{*} r_{n})} \mathrm{P}_{h}(\mathtt
{T}_{n}=0) \Bigr\}\geq \alpha ,
\end{align*}
where $\sup_{h\in \mathcal H_{\ell}}\mathrm{P}_{h}(\cdot )$ denotes the
supremum over $h\in \mathcal H_{\ell}$ and distributions of
$(X,W,U)$ satisfying Assumption~\ref{A:LB} for $\ell =0,1$.%

Further, when $\{\nu _{j}\}$ is regularly varying, the separation rate
$r_{n}$ given in \eqref{r_star_lower} simplifies to
%
\begin{equation}
\label{r_star_lower1} r_{n} \sim J_{*}^{-p/d_{x}},\quad  \text{with
}J_{*}\sim \max \bigl\{J: n^{-1/2} J^{1/4} \nu
_{J}^{-1}\leq J^{-p/d_{x}} \bigr\}. 
\end{equation}
\begin{enumerate}[(2)]
\item[(1)] Mildly ill-posed ($\nu _{j}=j^{-a/d_{x}}$) case:
$r_{n}\sim n^{-2p/(4(p+a)+d_{x})}$.
\item[(2)] Severely ill-posed ($\nu _{j}=\exp (-j^{a/d_{x}}/2)$) case:
$r_{n}\sim (\log n)^{-p/a}$.
\end{enumerate}
\end{theorem}

According to Theorem~\ref{thm:minimax:test:lower}, the lower bound of the
$L^{2}$ rate of testing is $n^{-2p/(4(p+a)+d_{x})}$ in the mildly ill-posed
case, which goes to zero faster than the lower bound
$n^{-p/(2(p+a)+d_{x})}$ of the $L^{2}$ rate of estimation (\citet{HallHorowitz2005}
and \citet{ChenReiss2011}). For the severely ill-posed NPIV models, the
lower bound of the $L^{2}$ rate of testing is $(\log n)^{-p/a}$, which
coincides with the lower bound of estimation in both the $L^{2}$ norm (\citet{ChenReiss2011})
and the sup-norm (\citet{ChenChristensen2017}).

In the literature on linear ill-posed inverse problem with a compact operator
$T$, an ``\emph{exact link condition}'' is commonly used to describe the
smoothing (or compact embedding) property of $T$, which can be stated as
follows:
%
\begin{align}
\label{bounds:smooth:T} c\sum_{j\geq 1} \nu _{j}^{2}
\langle h, \widetilde \psi _{j}\rangle _{X}^{2}
\leq \llVert T h \rrVert _{L^{2}(W )}^{2} \leq C \sum
_{j\geq 1} \nu _{j}^{2} \langle h, \widetilde
\psi _{j}\rangle _{X}^{2}\quad  \text{for all } h \in
\mathcal H 
\end{align}
for some finite constants $C\geq c>0$ and a positive decreasing function
$\nu $ with $\nu _{j}:=\nu (j)$. The RHS inequality of
\eqref{bounds:smooth:T} (i.e., Assumption~\ref{A:LB}(iv)) is used for the
lower bound calculation, and the LHS inequality of
\eqref{bounds:smooth:T} is imposed for the upper bound calculation. However,
to have matching lower and upper bound, that is, to establish the rate
is minimax optimal, the exact link condition \eqref{bounds:smooth:T} or
something similar is typically imposed even with a known $T$; see, for
example, \citet{ChenReiss2011}. We note that any compact operator
$T$ has a unique singular value decomposition. If the basis
$\{\widetilde \psi _{j}\}$ is an eigenfunction basis associated with the
operator $T$, then \eqref{bounds:smooth:T} is automatically satisfied with
$C=c=1$ and $\{\nu _{j}\}_{j=1}^{\infty}$ being its singular values in
decreasing order. More generally, \eqref{bounds:smooth:T} is also satisfied
when $\{\widetilde \psi _{j}\}$ is a Riesz basis (see
\citet{BCK07econometrica}). Since the conditional expectation operator
$T$ is compact under very mild conditions (such as when the conditional
density of $X$ given $W$ is continuous), it typically satisfies
\eqref{bounds:smooth:T}, which is an alternative way to express the smoothing
property of the operator $T$.

In our proof of Theorem~\ref{thm:minimax:test:lower}, we reduce the lower
bound calculation for the NPIV model to that for a model with a known operator
$T$. Consequently, Assumption~\ref{A:LB}(iv) is sufficient to establish
the lower bound. However, for the upper bound calculation of the NPIV model,
we need to estimate the unknown operator $T$. Therefore, in addition to
the LHS inequality of \eqref{bounds:smooth:T}, some extra sufficient conditions
will be used to address the error of estimating $T$ nonparametrically.
See the next subsection for details.

\subsection{An Upper Bound Under a Simple Null Hypothesis}
\label{subsec:upper_bound}

For a simple null $\mathcal H_{0}=\{h_{0}\}$, we redefine
$\widehat{D}_{J}$ in \eqref{whDJ} with
$\widehat h_{J}^{\text{\textsc r}} = h_{0}$ as
%
\begin{equation}
\label{whDh0} \widehat{D}_{J}(h_{0})= \frac{2}{n(n-1)}\sum
_{1\leq i< i'\leq n} \bigl(Y_{i}- h_{0}(X_{i})
\bigr) \bigl(Y_{i'}- h_{0}(X_{i'})
\bigr)b^{K(J)}(W_{i})' \widehat A'
\widehat Ab^{K(J)}(W_{i'}). 
\end{equation}
We also redefine our test statistic $\widehat{\mathtt {T}}_{n}$ with a
singleton RES index set $\{J\}$ as
%
\begin{align}
\label{def:test:J} \mathtt {T}_{n,J} = \mathbbm{1} \biggl\{ \frac{n
\widehat{D}_{J}(h_{0})}{\widehat{V}_{J}}>\eta
_{J}(\alpha ) \biggr\} \quad \text{with } \eta _{J}(\alpha )=
\bigl(q(\alpha , J)-J\bigr)/\sqrt{J}. 
\end{align}
The test $ \mathtt{T}_{n,J}$ with optimally chosen $J$ serves as a benchmark
of our adaptive testing procedure (given in \eqref{def:struc_test}) for
the simple null hypothesis.

We define the projections
$\Pi _{J} h(\cdot )=\psi ^{J}(\cdot )'G^{-1} \langle \psi ^{J},h
\rangle _{L^{2}(X)}$ for $h\in L^{2}(X)$ and
$\Pi _{K} m(\cdot )=b^{K}(\cdot )'G_{b}^{-1} \Evtex [b^{K}(W)m(W)]$ for
$m\in L^{2}(W)$. Further, let
$s_{J}=\inf_{h\in \Psi _{J}}\|\Pi _{K} T h \|_{L^{2}(W)} /\|h\|_{L^{2}(X)}$,
that is, $s_{J}$ coincides with the minimal singular value of
$G_{b}^{-1/2}SG^{-1/2}$. Let
$\zeta _{J}=\max (\zeta _{\psi ,J}, \zeta _{b,K})$,
$\zeta _{\psi ,J}=\sup_{x}\|G^{-1/2}\psi ^{J}(x)\|$, and
$\zeta _{b,K}=\sup_{w}\|G_{b}^{-1/2}b^{K}(w)\|$. We assume throughout
the paper that $\zeta _{J}=O(\sqrt{J})$ (which holds for polynomial spline,
wavelet, and cosine bases), or $\zeta _{J}=O(J)$ (which holds for orthogonal
polynomial bases).

\begin{A}%
\label{A:basis}
(i)
$\sup_{w\in \mathcal W}\sup_{h\in \mathcal H}\Evtex _{h}[(Y-
\widetilde h(X))^{2}|W=w]\leq \overline\sigma ^{2}<\infty $, where
$\widetilde h\in \{h,\Pi _{\mathcal H_0}h\}$ and
$\sup_{h\in \mathcal H}\Evtex _{h}[(Y-h(X))^{4}]<\infty $; (ii)
$s^{-1}_{J} \zeta _{J}^{2}\sqrt{(\log J)/ n}=O(1)$; (iii)
$\zeta _{J}\sqrt{\log J}=O( J^{p/d_{x}})$; (iv)
$s_{J}^{-1}\|\Pi _{K} T(\Pi _{J} h-h)\|_{L^{2}(W)} \leq C_{T} \|\Pi _{J}
h-h\|_{L^{2}(X)}$ for a constant $C_{T}>0$, uniformly for
$h\in \mathcal H$.
\end{A}

Let $\Psi _{J,1}:=\{h\in \Psi _{J}: \|h\|_{L^{2}(X)}=1\}$. Then
$\tau _{J}:=  [\inf_{h\in \Psi _{J,1}}\|T h \|_{L^{2}(W)}  ]^{-1}$
is the sieve measure of ill-posedness that has been used in sieve estimation
of NPIV models (see, e.g., \citet{BCK07econometrica}). We have
$s_{J} \leq \tau _{J}^{-1}$ by definition.

\begin{A}%
\label{A:s_J}
(i)
$\sup_{h\in \Psi _{J,1}} \tau _{J} \|(\Pi _{K} T-T)h\|_{L^{2}(W)} = o(1)$;
(ii) the LHS inequality of \eqref{bounds:smooth:T} holds.
\end{A}

Assumption~\ref{A:basis}(i) is an extra condition on the data-generating
process (DGP) since it imposes upper bounds on conditional second moment
and finiteness of unconditional fourth moment. We note that the DGP displayed
in our proof of Theorem~\ref{thm:minimax:test:lower} already satisfies
this assumption; it has no effect on our lower bound result. Assumptions
\ref{A:basis}(ii), (iii), (iv) are imposed since our test statistic involves
linear sieve estimated operator $T$ to achieve the separation rate. Assumptions
\ref{A:basis}(ii), (iii) impose restrictions on the sieve dimension
$J$, which are satisfied by $J_{*}$ given in \eqref{r_star_lower1} of Theorem~\ref{thm:minimax:test:lower}. Assumption~\ref{A:basis}(iv) imposes an upper
bound on the smoothing properties of the conditional expectation operator
$T$. It is akin to the $L^{2}$ stability condition used in sieve NPIV estimation
and is satisfied by Riesz bases (see
\citet [Assumption~6]{BCK07econometrica}). Assumption~\ref{A:s_J}(i) is
a mild condition on the approximation properties of the basis used for
the instrument space (see
\citet [Assumption~4(i)]{ChenChristensen2017}). It implies that
$s_{J}$ and $\tau _{J}^{-1}$ are asymptotically equivalent:
\begin{equation*}
\begin{split}\tau _{J}^{-1}&\geq s_{J}=\inf_{h\in \Psi _{J,1}}
\llVert \Pi _{K} T h \rrVert _{L^{2}(W)}
\\
&\geq \inf
_{h\in \Psi _{J,1}} \llVert T h \rrVert _{L^{2}(W)}-\sup
_{h\in \Psi _{J,1}} \bigl\llVert (\Pi _{K} T -T)h \bigr\rrVert
_{L^{2}(W)}=\tau _{J}^{-1} \bigl(1-o(1)\bigr),
\end{split}\end{equation*}
while Assumption~\ref{A:s_J}(ii) implies
$\tau _{J}^{-1}=\inf_{h\in \Psi _{J,1}}\|T h \|_{L^{2}(W)} \geq
\sqrt c \nu _{J}$ for all $J$. Assumption~\ref{A:s_J} thus implies
\begin{equation*}
s_{J}^{-1} \sim \tau _{J} \leq (\sqrt
c)^{-1} \nu _{J}^{-1}.
\end{equation*}
Further,
$s_{J} \sim \tau _{J}^{-1} \leq \|T \widetilde\psi _{J} \|_{L^{2}(W)}
\leq \sqrt C \nu _{J}$ under Assumption~\ref{A:LB}(iv) and
$\{\widetilde\psi _{j}\}$ being an orthonormal basis in
$L^{2}(X)$, and Assumption~\ref{A:basis}(iv) is satisfied under Assumptions
\ref{A:LB}(iv) and \ref{A:s_J}. Therefore, Assumptions \ref{A:basis} and
\ref{A:s_J} have no effect on the lower bound calculation in Theorem~\ref{thm:minimax:test:lower}.

The next theorem provides an upper bound on the separation rate of testing
in $L^{2}$ under a simple null using the test statistic
$\mathtt {T}_{n,J}$.

\begin{theorem}%
\label{thm:test:upper}
Let Assumptions \ref{A:LB}(i)--(iii) and \ref{A:basis} hold. Consider testing
the simple hypothesis $\mathcal H_{0}=\{h_{0}\} $ (for a known function
$h_{0}$) versus
$\mathcal H_{1}(\delta ^{\circ }r_{n,J})=\{ h\in \mathcal H:\|h-h_{0}
\|_{L^{2}(X)}\geq \delta ^{\circ }r_{n,J}\}$ for a constant
$\delta ^{\circ }>0$ and a separation rate
%
\begin{align}
\label{thm:test:upper:rate} r_{n,J}=\max \bigl\{n^{-1/2}s_{J}^{-1}J^{1/4},J^{-p/d_{x}}
\bigr\}. 
\end{align}
Then, for any $\alpha \in (0,1)$, we have
%
\begin{align}
\label{thm:test:upper:bound} \limsup_{n\to \infty}\mathrm{P}_{h_{0}}(\mathtt
{T}_{n,J}=1)\leq \alpha \quad \text{and}\quad  \lim_{n\to \infty}\sup
_{h\in
\mathcal H_{1}(\delta ^{\circ }r_{n,J})} \mathrm{P}_{h}(\mathtt {T}_{n,J}=0)=0.
\end{align}
In addition, let Assumption~\ref{A:s_J} hold and
$J_{*0}:=\max  \{J: n^{-1/2}\nu _{J}^{-1} J^{1/4}\leq J^{-p/d_{x}}
 \}$. Then: the test statistic $\mathtt {T}_{n,J_{*0}}$ attains the
optimal separation rate of
%
\begin{equation}
\label{r_star_minimax} r_{n,J_{*0}}= (J_{*0} )^{-p/d_{x}} \sim
r_{n} , 
\end{equation}
which is the lower bound rate given in \eqref{r_star_lower1} when
$\{\nu _{j}\}$ is regularly varying.
\begin{enumerate}[(2)]
\item[(1)] Mildly ill-posed case:
$J_{*0}\sim n^{2d_{x}/(4(p+a)+d_{x})}$ and
$r_{n,J_{*0}}\sim n^{-2p/(4(p+a)+d_{x})}$.
\item[(2)] Severely ill-posed case:
$J_{*0}=  (c\log n )^{d_{x}/a}$ for some $c\in (0,1)$ and
$r_{n,J_{*0}}\sim (\log n)^{-p/a}$.
\end{enumerate}
\end{theorem}
Theorem~\ref{thm:test:upper} shows that, under Assumptions \ref{A:LB}(i)--(iii)
and \ref{A:basis}, the test statistic $\mathtt{T}_{n,J}$ given in
\eqref{def:test:J} attains the $L^{2}$ separation rate of testing
$r_{n,J}$ in \eqref{thm:test:upper:rate}. Given a sieve dimension
$J$, this rate consists of a standard deviation term ($n^{-1/2}s_{J}^{-1}J^{1/4}$)
and a bias term ($J^{-p/d_{x}}$). A central step to achieve this rate result
is to establish a rate of convergence of the quadratic distance estimator
$\widehat D_{J}(h_{0})$ (see Theorem~\ref{thm:rate:quad:fctl}), which we
show is sufficient for the consistency of $\mathtt {T}_{n,J}$ uniformly
over $\mathcal H_{1}(\delta ^{\circ }r_{n,J})$. In addition, under Assumption~\ref{A:s_J}, Theorem~\ref{thm:test:upper} implies that the sieve test
$\mathtt{T}_{n,J_{*0}}$ achieves the $L^{2}$
\textit{minimax rate of testing} for a simple null, with known smoothness
$p$ of the nonparametric alternatives and known degree of ill-posedness.

Given a sieve dimension $J$, the $L^{2}$ rate of sieve estimation for any
NPIV function $h\in \mathcal H$ is
$\max \{n^{-1/2}s_{J}^{-1}J^{1/2}, J^{-p/d_{x}}\}$ (see, e.g.,
\citet{ChenReiss2011}). Comparing the $L^{2}$ rate of estimation and of
testing via the sieve NPIV procedures, while both have the same bias term
$J^{-p/d_{x}}$, the $L^{2}$ rate of testing has a smaller ``standard deviation''
term $n^{-1/2}s_{J}^{-1}J^{1/4}$. Intuitively, we may obtain a higher precision
in testing as the $L^{2}$ rate of testing is determined by estimating a
quadratic norm of the unrestricted NPIV function $h \in \mathcal H$. Interestingly,
although this leads to a faster optimal $L^{2}$ rate of sieve testing
$r_{n,J_{*0}}\sim n^{-2p/(4(p+a)+d_{x})}$ than the optimal $L^{2}$ rate
of estimation $n^{-p/(2(p+a)+d_{x})}$ in the mildly ill-posed case, the
optimal $L^{2}$ rate of sieve testing
$r_{n,J_{*0}}\sim (\log n)^{-p/a}$ in the severely ill-posed case is the
same as the optimal rate of sieve estimation in both the $L^{2}$ norm (\citet{ChenReiss2011})
and the sup-norm (\citet{ChenChristensen2017}). This is because, in the
severely ill-posed case, the bias term dominates the standard deviation
term for the optimally chosen sieve dimension in both sieve testing and
estimation.

\section{Adaptive Inference}
\label{sec:adapt:test}

This section establishes theoretical properties of our test
$\widehat {\mathtt{T}}_{n}$ defined in \eqref{def:test:unknown}. We show
that it adapts to the unknown smoothness $p>0$ of the functions in
$\mathcal H$. Section~\ref{sec:adapt:hyp:test} establishes the rate
optimality of our test for simple null hypotheses. Section~\ref{subsec:unknown:null} extends this result to testing for composite
null problems. Section~\ref{subsec:CS} proposes $L^{2}$ confidence sets
by inverting the adaptive test under imposed restrictions on the NPIV function.

\subsection{Adaptive Testing Under a Simple Null Hypothesis}
\label{sec:adapt:hyp:test}

Under the simple null hypothesis $\mathcal H_{0}=\{h_{0}\}$ with a known
function $h_{0}$ satisfying \eqref{npiv:model}, our test
$\widehat{\mathtt{T}}_{n}$ given in \eqref{def:test:unknown} simplifies
to
%
\begin{align}
\label{def:struc_test} \widehat{\mathtt{T}}_{n} &= \mathbbm{1} \biggl\{
\text{there exists } J\in \widehat{\mathcal I}_{n}\text{ such that }
\frac{n\widehat{D}_{J}(h_{0})}{\widehat{V}_{J}}>
\widehat\eta _{J}( \alpha ) \biggr\}, 
\end{align}
where $\widehat{D}_{J}(h_{0})$ is defined in (\ref{whDh0}), and
$\widehat{\mathcal I}_{n}$, $\widehat{V}_{J}$,
$\widehat \eta _{J}(\alpha )$ are given in Section~\ref{subsec:sst}.

Recall that the RES index set $\widehat{\mathcal I}_{n}$, given in
\eqref{def:index_set}, depends on an upper bound $\widehat J_{\max}$ given
in \eqref{def:J_max}. To establish our asymptotic results below, we introduce
a non-random index set $\mathcal I_{n}$ with a deterministic upper bound
$\overline J$ as follows:
%
\begin{align}
\label{det:index:set} \mathcal I_{n}= \bigl\{J\leq \overline J: J=\underline
J2^{j} \text{ where } j=0,1,\dots ,j_{\max} \bigr\}\subset [
\underline J, \overline J], 
\end{align}
with
$\overline J=\sup \{J:  \zeta _{J}^{2} \sqrt{(\log J)/n}\leq
\overline c  s_{J}\}$ for some sufficiently large constant
$\overline c>0$. We show in Lemma~\ref{lemma:J_hat}(i) that
$\widehat J_{\max}\leq \overline J$ (and thus
$\widehat{\mathcal I}_{n}\subset \mathcal I_{n}$) holds with probability
approaching 1 uniformly over all functions $h\in \mathcal H$. Thus,
$\overline J$ serves as a deterministic upper bound for the RES index set
$\widehat{\mathcal I}_{n}$.

\begin{A}%
\label{A:adapt:test}
(i) Assumptions \ref{A:basis}(ii), (iv) hold uniformly for all
$J\in \mathcal I_{n}$; (ii)
$s_{J}^{-4}J \lesssim \sum_{j=1}^{J} s_{j}^{-4}$ uniformly for all
$J\in \mathcal I_{n}$; (iii) $p\geq 3d_{x}/4$ when using cosine, spline,
or wavelet basis functions and $p\geq 7d_{x}/4$ when using power series
basis functions.
\end{A}
Assumptions \ref{A:adapt:test}(i), (iii) strengthen Assumptions
\ref{A:basis}(ii), (iii), (iv) to hold uniformly over the deterministic index
set $\mathcal I_{n}$. They are used to establish Lemma~\ref{lemma:J_hat}. Assumption~\ref{A:adapt:test}(i) restricts the growth
of the deterministic upper bound $\overline J$ of the RES index set
$\widehat{\mathcal I}_{n}$. Assumption~\ref{A:adapt:test}(ii) is satisfied
if $\{s_{j}\}$ is regularly varying, which is implied by Assumptions
\ref{A:LB}(iv) and \ref{A:s_J} with $\{\nu _{j}\}$ regularly varying. We
note that Assumptions \ref{A:adapt:test}(ii), \ref{A:basis}(i), and
\ref{A:LB}(i) together imply that $ V_{J} \sim s_{J}^{-2}\sqrt{J}$ uniformly
for $h\in \mathcal H$ and $J\in \mathcal I_{n}$ (see Lemmas
\ref{upper:bound:v_n} and \ref{lower:bound:v_n}).

\begin{theorem}%
\label{thm:adapt:test}
Let Assumptions \ref{A:LB}(i)--(iii), \ref{A:basis}(i), \ref{A:s_J}, and
\ref{A:adapt:test} hold. Consider
testing the simple null $\mathcal H_{0}=\{h_{0}\}$ (for a known function
$h_{0}$) versus
$\mathcal H_{1}(\delta ^{\circ}\text{\textsf r}_{n})=\{ h\in \mathcal H:\|h-h_{0}
\|_{L^{2}(X)}\geq \delta ^{\circ}\text{\textsf r}_{n}\}$ for a constant
$\delta ^{\circ}>0$ and an adaptive separation rate
%
\begin{align}
\label{thm:adapt:test:rate} \text{\textsf r}_{n}=\bigl(J^\circ
\bigr)^{-p/d_{x}},\quad  \text{where } J^\circ :=\max \bigl\{J:
n^{-1/2} \nu _{J}^{-1}(J\log \log n)^{1/4}
\leq J^{-p/d_{x}} \bigr\}. 
\end{align}
Then, for any $\alpha \in (0,1)$, we have
%
\begin{align}
\label{thm:adapt:test:bound} \limsup_{n\to \infty}\mathrm{P}_{h_{0}}(
\widehat{\mathtt{T}}_{n}=1) \leq \alpha \quad \text{and}\quad  \lim
_{n\to \infty}\sup_{h\in
\mathcal H_{1}(\delta ^{\circ }\text{\textsf r}_{n})} \mathrm{P}_{h}(
\widehat{\mathtt{T}}_{n}=0)=0. 
\end{align}
\begin{enumerate}[(2)]
\item[(1)] Mildly ill-posed case:
$\text{\textsf r}_{n}\sim  (\sqrt{\log \log n}/n )^{2p/(4(p+a)+d_{x})}$.
\item[(2)] Severely ill-posed case:
$\text{\textsf r}_{n}\sim (\log n)^{-p/a}$.
\end{enumerate}
\end{theorem}

Theorem~\ref{thm:adapt:test} establishes an upper bound for the testing
rate of the adaptive test $\widehat {\mathtt{T}}_{n}$ under a simple null
hypothesis. The proof of Theorem~\ref{thm:adapt:test} relies on a novel
exponential bound for degenerate U-statistics based on sieve estimators
(see Lemma~\ref{Lemma:houdre}). In particular, we control the type I error
using tight lower bounds for adjusted chi-squared critical values (see
Lemma~\ref{lemma:eta:bounds}) and establish the consistency of
$\widehat{\mathtt{T}}_{n}$ uniformly over
$\mathcal H_{1}(\delta ^{\circ }\text{\textsf r}_{n})$.

From Theorem~\ref{thm:adapt:test}, we see that the adaptive test attains
the oracle minimax rate of testing within a $\sqrt{\log \log (n)}$ term
in the mildly ill-posed case. For the adaptive testing in regression models
without endogeneity (i.e., when $X=W$), it is well known that the extra
$\sqrt{\log \log (n)}$ term is required (see \citet{spokoiny1996}). In
the severely ill-posed case, our adaptive test attains the exact minimax
rate of testing and hence, there is no price to pay for adaptation. This
is because, in the severely ill-posed case, the bias term dominates the
standard deviation term when the sieve dimension coincides with
$J^{\circ}$, irrespective of the $\sqrt{\log \log (n)}$ term.

\begin{remark}%
\label{rm:no-assumption3}
As is clear from the proof, Result \eqref{thm:adapt:test:bound} of Theorem~\ref{thm:adapt:test}
remains valid with an adaptive rate $\text{\textsf r}_{n}=(J_s^\circ )^{-p/d_{x}}$ when $J_s^\circ := \max \bigl\{J:n^{-1/2} s_{J}^{-1}(J\log \log n)^{1/4}
\leq J^{-p/d_{x}} \bigr\}$, without imposing Assumption \ref{A:s_J}. The extra Assumption \ref{A:s_J}, or its consequence $s_J^{-1} \lesssim \nu_J^{-1}$, is used to establish the optimality of the adaptive rate $\text{\textsf r}_{n}=(J_s^\circ )^{-p/d_{x}}$ only. The same remark
also applies to Result \eqref{thm:adapt:test:est:upper} of Theorem~\ref{thm:adapt:est:test}, Corollary~\ref{coro:CB} and Corollary~\ref{thm:diam} below.
\end{remark}%

\subsection{Adaptive Testing Under Composite Null Hypotheses}
\label{subsec:unknown:null}

We extend the results from Section~\ref{sec:adapt:hyp:test} to adaptive
testing for a general composite null hypothesis $\mathcal H_{0}$, which
is a nonempty, closed and convex strict subset of $\mathcal H$. Without
loss of generality, we assume $0 \in \mathcal H_{0}$. This is satisfied
for the inequality restrictions in Example~\ref{ex:shape} and the semiparametric
equality restrictions considered in Example~\ref{ex:semi} if, for instance,
$F(\cdot ;\theta ,g)= 0$ for some $\theta \in \Theta $ and
$g\in \mathcal G$.

Below, we impose some conditions on the complexity of the closed and convex
null class of functions $\mathcal H_{0}$. Let
$\mathcal S^{K}=\{\mathsf e\in \mathbb R^{K}: \mathsf e_{1}^{2}+
\cdots +\mathsf e_{K}^{2}=1\}$ denote the $(K-1)$-dimensional unit sphere.
Let $K^\circ =K(J^\circ )$,
$\widetilde b^{K}(\cdot )=G_{b}^{-1/2}b^{K}(\cdot )$, and
$Z:=(X',W')'$. For any
$h\in \mathcal H_{1}(\delta ^{\circ}\text{\textsf r}_{n})$, we consider the following
class of functions:
\begin{align*}
\mathcal F_{h,\mathsf e}:= \bigl\{(\phi - \Pi _{\mathcal H_{0}} h) (X)
\widetilde b^{K^\circ }(W)'\mathsf e: \phi \in \mathcal
H_{0,J^\circ } \bigr\}, \quad \mathsf e\in \mathcal S^{K^\circ },
\end{align*}
with its envelope function denoted by $F_{h,\mathsf e}$. Let
$N_{[]}(\epsilon , \mathcal F, L^{2}(Z))$ be the $L^{2}(Z)$ covering number
with bracketing for $\mathcal F$, which is the minimal number of
$\epsilon $-brackets, in $L^{2}(Z)$ sense, needed to cover
$\mathcal F$. We let
$\mathcal C_{h} :=\max_{\mathsf e\in \mathcal S^{K^\circ }}\int _{0}^{1}
  (1+\log N_{[]}  (\epsilon \|F_{h,\mathsf e}\|_{L^{2}(Z)},
\mathcal F_{h,\mathsf e}, L^{2}(Z)  )  )^{1/2}\,d\epsilon $.
%
\begin{A}%
\label{h0:est}
(i) For any $\varepsilon >0$, it holds that
$\sup_{h\in \mathcal H_{0}}\mathrm{P}_{h} (\max_{J\in \mathcal I_{n}}(
\zeta _{J}\|\widehat h_{J}^{\text{\textsc r}}-h\|_{L^{2}(X)}/c_{J} ) >
\varepsilon  )\to 0$ with $c_{J} =\max \{1,(\log \log J)^{1/4}\}$; (ii)
for some constant $C>0$, it holds that
$\sup_{h\in \mathcal H_{1}(\delta ^{\circ}\text{\textsf r}_{n})}\mathrm{P}_{h}
 (\zeta _{J^\circ }\mathcal C_{h} \|\widehat h_{J^\circ }^{
\text{\textsc r}}-\Pi _{\mathcal H_0}h\|_{L^{2}(X)}> C )\to 0$ and
$\sup_{h\in \mathcal H_{1}(\delta ^{\circ}\text{\textsf r}_{n})}\mathcal C_{h}
\lesssim (J^\circ )^{1/4}$.
\end{A}

Assumption~\ref{h0:est} restricts the complexity of the composite null
hypothesis $\mathcal H_{0}$. Assumption~\ref{h0:est}(i) implies that
$\widehat {\mathtt{T}}_{n}$ has size control uniformly over the composite
null $\mathcal H_{0}$. Assumption~\ref{h0:est}(ii) ensures the consistency
of $\widehat {\mathtt{T}}_{n}$ uniformly over
$\mathcal H_{1}(\delta ^{\circ}\text{\textsf r}_{n})$. Note that Assumption~\ref{h0:est} imposes estimation rate conditions on
$\widehat h_{J}^{\text{\textsc r}}$ under the composite null and the nonparametric
alternatives, which can be viewed as NPIV extensions of the parametric
estimation rate conditions imposed in
\citet [Assumption~2]{horowitz2001} for testing for a parametric regression
against nonparametric regressions.
%
\begin{remark}[Sufficient conditions for Assumption~\ref{h0:est}(i)]%
\label{rem:A1}
Assumption~\ref{h0:est}(i) is a very mild condition on the estimation rate
(in $L^{2}$) of the restricted sieve NPIV estimator under
$\mathcal H_{0}$.%

(1) In the case of parametric restrictions, where
$\|\widehat h_{J}^{\text{\textsc r}}-h\|_{L^{2}(X)}\leq \text{const}. \times n^{-1/2}$
with probability approaching 1 uniformly over $h\in \mathcal H_{0}$, Assumption~\ref{h0:est}(i) is automatically satisfied by Assumption~\ref{A:adapt:test}(i).%

(2) Under nonparametric restrictions, we note that
$\|\widehat h_{J}^{\text{\textsc r}}-h\|_{L^{2}(X)}\leq \|\widehat h_{J}-h\|_{L^{2}(X)}$
for all $h\in \mathcal H_{0}$, and that
%
\begin{align}\label{rate:a1}
\max_{J\in\mathcal I_n}\frac{\zeta_J\|\widehat h_J-h\|_{L^2(X)}}{c_J}\leq const.\times \max_{J\in\mathcal I_n}\Bigg\{\frac{\zeta_J \sqrt{J}}{\sqrt ns_Jc_J}
+\frac{\zeta_J\|\Pi^{\mathcal I_n}_Jh-h\|_{L^2(X)}}{c_J}\Bigg\}
\end{align}
with probability approaching 1 uniformly for $h\in \mathcal H_{0}$, where
$\Pi ^{\mathcal I_{n}}_{J}$ denotes the projection onto the closed linear
subspace of $L^{2}(X)$ spanned by
$\{\psi _{J}:J\in \mathcal I_{n}\}$. The first summand on the right-hand
side of \eqref{rate:a1} converges to zero by the definition of
$\overline J=\overline J(n)$. For the bias part, we assume that the index
set has sufficient information to approximate the NPIV function
$h\in \mathcal H_{0}$. Let $p_{0}$ denote the smoothness and $d_{0}$ the
dimension of the nonparametric component under $\mathcal H_{0}$. If
$\|\Pi ^{\mathcal I_{n}}_{J}h-h\|_{L^{2}(X)} =O(J^{-p_{0}/d_{0}})$ and
$\zeta _{J}=O(\sqrt J)$, the second summand of the right-hand side of
\eqref{rate:a1} uniformly converges to zero if
$p_{0}/d_{0} \geq 1/2$. Since the class $\mathcal H_{0}$ is a less complex
subset of $\mathcal H$, it is reasonable to assume that
$p_{0}/d_{0} \geq p/d_{x}$ and thus $p_{0}/d_{0} \geq 1/2$ is automatically
satisfied given Assumption~\ref{A:adapt:test}(iii).
\end{remark}

\begin{remark}[Sufficient conditions for Assumption~\ref{h0:est}(ii)]%
\label{rem:smooth}
Assumption~\ref{h0:est}(ii) restricts the complexity of
$\mathcal H_{0}$ to have no effect on the adaptive minimax rate of testing
asymptotically. Note that for any $\epsilon >0$ and
$\mathsf e\in \mathcal S^{K^\circ }$, we have
%
\begin{align*}
\Evtex\Big[\sup_{\phi_1,\phi_2\in\mathcal H_{0,\Jad}:\, \|\phi_1-\phi_2\|_\infty\leq\epsilon}\big|(\phi_1-\phi_2)(X)\widetilde b^{\Kad}(W)'\mathsf e\big|^2\Big]\leq \epsilon^2,
\end{align*}
using that $\Evtex (\widetilde b^{K^\circ }(W)'\mathsf e)^{2}=1$. Thus,
$\log N_{[]}  (\epsilon , \mathcal F_{h,\mathsf e}, L^{2}(Z)
  )\leq \log N_{[]}  (\epsilon ,\mathcal H_{0,J^\circ }, L^{
\infty}  )\lesssim \epsilon ^{-d_{x}/p}$ if the functions in
$\mathcal H_{0}$ have uniformly bounded partial derivatives with highest
order derivatives being Lipschitz; see
\citet [Theorem~2.7.1]{Vaart2000}. We obtain
$\mathcal C_{h}\lesssim 1$ under the condition $2p\geq d_{x}$, which is
satisfied given Assumption~\ref{A:adapt:test}(iii). In this case, a sufficient
condition for Assumption~\ref{h0:est}(ii) is given by
$\mathrm{P}_{h} (\zeta _{J^\circ }\|\widehat h_{J^\circ }^{
\text{\textsc r}}-\Pi _{\mathcal H_0}h\|_{L^{2}(X)}> C )\to 0$ uniformly for
$h\in \mathcal H_{1}(\delta ^{\circ}\text{\textsf r}_{n})$, which is less restrictive
than Assumption~\ref{h0:est}(i) since the sieve dimension is fixed at
$J^\circ $. When the basis functions in $\widetilde b^{K^\circ }$ are uniformly
bounded, such as for trigonometric bases, we immediately obtain
$\mathcal C_{h}\lesssim 1$. If $\mathcal H_{0}$ consists of convex functions
that are Lipschitz and map a compact and convex set in $\mathbb R$ to
$[0,1]$, then $\mathcal C_{h}\lesssim 1$ by
\citet [Corollary~2.7.10]{Vaart2000}.
\end{remark}

The next result establishes an upper bound for the rate of testing under
a composite null hypothesis using the test statistic
$\widehat {\mathtt{T}}_{n}$ given in \eqref{def:test:unknown}.

\begin{theorem}%
\label{thm:adapt:est:test}
Let Assumptions \ref{A:LB}(i)--(iii), \ref{A:basis}(i), \ref{A:s_J},
\ref{A:adapt:test}, and \ref{h0:est} hold. Consider testing the composite null $\mathcal H_{0}$ versus
$\mathcal H_{1}(\delta ^{\circ}\text{\textsf r}_{n})=\{ h\in \mathcal H: \|h-
\mathcal H_{0}\|_{L^{2}(X)}\geq \delta ^{\circ }\text{\textsf r}_{n}\}$ for a
constant $\delta ^{\circ}>0$ and the adaptive (separation) rate
$\text{\textsf r}_{n} =(J^\circ )^{-p/d_{x}}$ given in Theorem~\ref{thm:adapt:test}. Then, for any $\alpha \in (0,1)$, we have
%
\begin{align}
\label{thm:adapt:test:est:upper} \limsup_{n\to \infty}\sup_{h\in \mathcal H_{0}}
\mathrm{P}_{h}( \widehat{\mathtt{T}}_{n}=1)\leq \alpha
\quad \text{and} \quad \lim_{n
\to \infty}\sup_{h\in \mathcal H_{1}(\delta ^{\circ }\text{\textsf r}_{n})}
\mathrm{P}_{h}(\widehat{\mathtt{T}}_{n}=0)=0. 
\end{align}
\begin{enumerate}[(2)]
\item[(1)] Mildly ill-posed case:
$\text{\textsf r}_{n}\sim  (\sqrt{\log \log n}/n )^{2p/(4(p+a)+d_{x})}$.
\item[(2)] Severely ill-posed case:
$\text{\textsf r}_{n}\sim (\log n)^{-p/a}$.
\end{enumerate}
\end{theorem}

Theorem~\ref{thm:adapt:est:test} states that
$\widehat {\mathtt{T}}_{n}$ attains the same adaptive rate of testing
$\text{\textsf r}_{n}$ for a composite null as that for a simple null. Moreover,
\eqref{thm:adapt:test:est:upper} shows that
$\widehat{\mathtt {T}}_{n}$ simultaneously has asymptotic size control
over the composite null, and is consistent uniformly over the largest class
of nonparametric alternatives
$\mathcal H_{1}(\delta ^{\circ }\text{\textsf r}_{n})$. The asymptotic size control
is established by controlling the sieve approximation error uniformly over
the index set $\widehat{\mathcal I}_{n}$ under the null, due to a projection
property built in the construction of our test
$\widehat {\mathtt{T}}_{n}$. See Lemma~\ref{Lemma:adapt:est:step1}, in which we utilize the convergence of von Neumann's alternating projection algorithm.
Theorem~\ref{thm:adapt:est:test} is applicable to any composite null hypothesis
$\mathcal H_{0}$ that is a closed convex strict subset of
$\mathcal H$, including closed convex cone null restrictions as special
cases.

Theorem~\ref{thm:adapt:est:test} shows that our adaptive test has asymptotic
size control and non-trivial power against a large class of nonparametric
NPIV alternatives without using under-smoothed choice of sieve dimensions
in testing. This is different from the existing non-adaptive tests for
semiparametric or shape NPIV restrictions, which achieve asymptotic size
controls via under-smoothed choice of tuning parameters in $L^{2}$ estimation.
For instance, in their bootstrap test for convex cone restrictions of a
NPIV function, \citet{fang2019} estimated the unrestricted NPIV function
by a sieve 2SLS estimator assuming known smoothness, and chose the sieve
dimension $J$ deterministically such that the estimation bias
$J^{-p/d_{x}}$ is of a smaller order than the standard deviation
$n^{-1/2}s_{J}^{-1}J^{1/2}$ in $L^{2}$ estimation, which leads to a non-adaptive
rate of testing $n^{-1/2}s_{J}^{-1}J^{1/2}$ that is suboptimal for
$L^{2}$ testing of NPIV models.

\begin{remark}%
\label{rm:other-test}
Our adaptive minimax $L^{2}$ rate of testing
$ (\sqrt{\log \log n}/n )^{2p/(4(p+a)+d_{x})}$ decreases to zero
strictly faster than the optimal $L^{2}$ rate of estimation
$n^{-p/(2(p+a)+d_{x})}$ (even assuming known smoothness) for mildly ill-posed
NPIV models, and coincides with the optimal $L^{2}$ rate of estimation
$(\log n)^{-p/a}$ for severely ill-posed NPIV models. Therefore, any test
statistic based on a tuning parameter chosen for the under-smoothed
$L^{2}$ rate of NPIV estimation will not be as powerful as our new test
uniformly over a large class of nonparametric alternatives.
\end{remark}
\subsubsection*{Adaptive Testing in Semiparametric Models}
Partially parametric models are often used in empirical work and can be
easily incorporated in our framework either as restricted models or as
the maintained models. Let
$\Theta \oplus \mathcal G =\{h(x_{1},x_{2})= x_{1}'\theta +g(x_{2}):
\theta \in \Theta , g\in \mathcal G\}$, where $\Theta $ denotes a finite-dimensional
parameter space, and $\mathcal G$ denotes a class of nonparametric functions.

Let the NPIV model \eqref{NPIV:maintained} be the maintained hypothesis.
We can test inequality restrictions as in Example~\ref{ex:shape} and a
semiparametric structure simultaneously. For example, we can test for a
partial linear structure with a nondecreasing function $g$ by setting
$\mathcal H_{0}=\{h\in \Theta \oplus \mathcal G: \partial _{x_{2}} g
\geq 0\}$. The class of alternative functions can then be written as
$\mathcal H_{1}(r_{n}): = \{ g\in \mathcal G:  \|g-\mathcal G_{0}
\|_{L^{2}(X_{2})}\geq r_{n} \}$, where
$\mathcal G_{0}=\{g\in \mathcal G: \partial _{x_{2}} g\geq 0\}$ and the
rate of testing $r_{n}$ does not depend on the dimensionality of
$X_{1}$. We can also test for the nonnegativity of the coefficient
$\theta $ and a partial linear restriction by setting
$\mathcal H_{0}=\{h\in \Theta \oplus \mathcal G: \partial _{x_{1}} h
\geq 0\}$. As in Example~\ref{ex:semi}, we can test semiparametric equality
restriction by taking $\mathcal H_{0}=\Theta \oplus \mathcal G$.

Let the partial linear IV model be the maintained hypothesis in model
\eqref{NPIV:maintained} with $\mathcal H=\Theta \oplus \mathcal G$. The
maintained partial linear structure can be easily enforced in the sieve
space used to estimate the unconstrained NPIV function. For instance, we
impose a partial linear structure $\mathcal H$ in our empirical illustration
on demand for differential products in Section~\ref{subsec:scanner:data}. Monotonicity in all arguments of $h$ can be
imposed by
$\mathcal H_{0}=\{h\in \Theta \oplus \mathcal G: \theta \geq 0,
\partial _{x_{2}} g\geq 0\}$. We also allow for second or higher order
derivatives in the hypotheses considered above.

\subsection{Confidence Sets in $L^{2}$}
\label{subsec:CS}

One can construct $L^{2}$ confidence sets for a NPIV function by inverting
our adaptive test. For any small $\alpha >0$, the $(1-\alpha )$ confidence
set for a NPIV function $h$ belonging to a restricted nonparametric class
$\mathcal{H}_{0}$ is given by
%
\begin{align}
\label{def:cs} \mathcal C_{n}(\alpha ) = \biggl\{h\in \mathcal
H_{0}: \frac{n\widehat{D}_{J}(h)}{ \widehat{V}_{J}}
\leq \widehat \eta _{J}( \alpha ) \text{ for all } J\in \widehat{\mathcal
I}_{n} \biggr\}. 
\end{align}
This confidence set does not depend on additional tuning parameters. The
following corollary exploits our previous results to characterize the asymptotic
size and power properties of our procedure.

\begin{corollary}%
\label{coro:CB}
Let Assumptions \ref{A:LB}(i)--(iii), \ref{A:basis}(i), \ref{A:s_J}, and
\ref{A:adapt:test} hold. Let
$\text{\textsf r}_{n} =(J^\circ )^{-p/d_{x}}$ be the adaptive rate of testing
given in Theorem~\ref{thm:adapt:test}. Then, for any
$\alpha \in (0,1)$, it holds that
%
\begin{align}
\label{coro:CB:size} \limsup_{n\to \infty}\sup_{h\in \mathcal H_{0}}
\mathrm{P}_{h} \bigl( h\notin \mathcal C_{n}(\alpha ) \bigr)
\leq \alpha 
\end{align}
and there exists a constant $\delta ^{\circ}>0$ such that
%
\begin{align}
\label{coro:CB:power} \lim_{n\to \infty}\inf_{h\in \mathcal H_{1}(\delta ^{\circ}
\text{\textsf r}_{n})}
\mathrm{P}_{h} \bigl( h\notin \mathcal C_{n}(\alpha )
\bigr)=1. 
\end{align}
\end{corollary}

Corollary~\ref{coro:CB} result (\ref{coro:CB:size}) shows that the
$L^{2}$ confidence set $\mathcal C_{n}(\alpha )$ controls size uniformly
over the class of functions $\mathcal H_{0}$. Moreover, result (\ref{coro:CB:power})
establishes power uniformly over the class
$\mathcal H_{1}(\delta ^{\circ }\text{\textsf r}_{n})$. We immediately see from
Corollary~\ref{coro:CB} that the diameter of the $L^{2}$ confidence ball,
$\textsf{diam}(\mathcal C_{n}(\alpha ))=\sup  \{\|h_{1}-h_{2}\|_{L^{2}(X)}:
 h_{1},h_{2}\in \mathcal C_{n}(\alpha ) \}$, depends on the degree
of ill-posedness and the unknown smoothness $p$ of $\mathcal H$.

\begin{corollary}%
\label{thm:diam}
Let Assumptions \ref{A:LB}(i)--(iii), \ref{A:basis}(i), \ref{A:s_J}, and
\ref{A:adapt:test} hold. Then, for any
$\alpha \in (0,1)$, we have
$\sup_{h\in \mathcal H_{0}}\mathrm{P}_{h}  (\textsf{diam}(
\mathcal C_{n}(\alpha ))\geq C \text{\textsf r}_{n}   )=o(1)$, for some constant
$C>0$ and the adaptive rate $\text{\textsf r}_{n}=(J^\circ )^{-p/d_{x}}$ given
in Theorem~\ref{thm:adapt:test}.
\end{corollary}

Corollary~\ref{thm:diam} yields a confidence set whose diameter shrinks
to zero at the adaptive optimal testing rate (of the order
$ (J^\circ )^{-p/d_{x}}$) and whose implementation does not require specifying
the values of any unknown regularity parameters. Our confidence set
$\mathcal C_{n}(\alpha )$ thus adapts to the unknown smoothness $p$ of
$\mathcal H$ (the class of unrestricted NPIV functions).
\section{Monte Carlo Studies}
\label{sec:mc}

This section presents Monte Carlo performance of our adaptive test for
monotonicity and parametric form of an NPIV function using simulation designs
based on \citet{chernozhukov2015}. See Supplemental Appendix~C for additional simulation results using other
designs. All the simulation results reported here are based on
$5000$ Monte Carlo replications for each experimental design and at
$\alpha =0.05$ nominal level. The simulation results clearly indicate that
our simple adaptive test has size-control and finite-sample non-trivial
power uniformly against a large class of NPIV alternatives, even for models
with relatively weak instruments. In addition, simulation and real data
application results reported in \citet{BC2020}, but not here due to the
lack of space, have demonstrated that our adaptive test and its bootstrapped
version perform similarly well in both finite-sample size and power.

For all the designs in this section, $Y$ is generated according to the
NPIV model \eqref{NPIV:maintained} for scalar-valued random variables
$X$ and $W$. We let $X_{i} =\Phi (X_{i}^{*})$ and
$ W_{i} =\Phi (W_{i}^{*})$, where $\Phi $ denotes the standard normal distribution
function, and generate the random vector
$(X_{i}^{*},W_{i}^{*}, U_{i})$ according to
%
\begin{align}\label{design:I}
\begin{pmatrix}
X_i^*\\
W_i^*\\
U_i
\end{pmatrix}
\sim
\mathcal N\left(\begin{pmatrix}
0\\
0\\
0
\end{pmatrix},
\begin{pmatrix}
1 &\xi &0.3\\
\xi &1 &0\\
0.3 &0&1
\end{pmatrix}\right)~.
\end{align}
The parameter $\xi $ captures the strength of instruments and varies in
the experiments below. As $\xi $ increases, the instrument becomes stronger
(or the ill-posedness gets weaker). While \citet{chernozhukov2015} fixed
$\xi = 0.5$ in their design, we let $\xi \in \{0.3, 0.5, 0.7\}$ in our
simulation studies. The functional form of $h$ varies in different Monte
Carlo designs below.

\subsection{Adaptive Testing for Monotonicity}
\label{sim:monoton}

We generate $Y$ using \eqref{NPIV:maintained} and \eqref{design:I} with
$h$ from the \citet{chernozhukov2015} design:
%
\begin{align}
\label{design:h:null} h(x)=c_{0} \biggl[1-2\Phi \biggl(\frac{x-1/2}{c_{0}}
\biggr) \biggr]\quad  \text{for some constant } c_{0} \in [0, 1]. 
\end{align}
This function $h(x)$ is decreasing in $x$, where $c_{0}$ captures the degree
of monotonicity. We note that $c_{0}=0$ corresponds to
$h(x)\equiv 0$ (the boundary case); $h(x) \approx 0$ for $c_{0}$ close
to zero and $h(x)\approx \phi (0)(1-2x)$ for $c_{0}$ close to 1, where
$\phi $ denotes the standard normal probability density function. The null
hypothesis is that the NPIV function $h$ is weakly decreasing on the support
of $X$.

We implement our adaptive test statistic $\widehat {\mathtt{T}}_{n}$ given
in \eqref{def:test:unknown} using quadratic B-spline basis functions with
varying number of knots for $h$. Due to piecewise linear derivatives, monotonicity
constraints are easily imposed on the restricted function at the derivative
at $J-1$ points. For the instrument sieve $b^{K(J)}(W)$, we also use quadratic
B-spline functions with a larger number of knots with
$K(J)\in \{2 J,  4J,  8J\}$. Implementation of the restricted sieve NPIV
estimator $\widehat h_{J}^{\text{\textsc r}}$ is straightforward using the R package
\texttt{coneproj}. We compare our adaptive test to the nonadaptive test
of \citet{fang2019}, which involves approximately computing
$[n^{-1/2}s_{J}^{-1}J^{1/2}]^{-1}\min_{h\in \mathcal H_{0}}\|
\widehat h_{J}- h\|_{L^{2}(X)}$ for a deterministic choice of sieve dimensions
$J$ and $K\geq J$ in their B-spline 2SLS estimate $\widehat h_{J}$. Their
2019 arXiv preprint presents a simulation study with $J=3$,
$K\in \{3, 4, 5\}$, and other tuning parameter choices
$c_{n} = (\log n)^{-1}$ and $\gamma _{n} = 0.01 / \log n$, such that their
test achieves approximately empirical size control with a sample size
$n=500$. Below, we use FS to denote their test with $J=3$ and $K=5$ (as
$K=5$ yields the best empirical power in their simulation), which is computed
using R language translation of their Matlab program code. To study the
sensitivity to the choice of $K$, we also implement their test with
$K = 12, 24$. In our simulations, we implement their test using 200 bootstrap
iterations.

\begin{table}[h!]
 \begin{center}
 \renewcommand{\arraystretch}{1}
{\footnotesize  \begin{tabular}{c|cc||cc|cc|cc|c|c|c}
\hline
\textit{$n$}&$c_0$& $\xi$&{\it $\widehat {\mathtt{T}}_n$}&{\it $\widehat J$}&{\it $\widehat {\mathtt{T}}_n$}&{\it $\widehat J$} &{\it $\widehat {\mathtt{T}}_n$}&{\it $\widehat J$} &FS&FS&FS\\
 &  & &  \multicolumn{2}{c|}{$K(J)=2J$}&\multicolumn{2}{c|}{$K(J)=4J$}&\multicolumn{2}{c|}{$K(J)=8J$}&$K=5$&$K=12$&$K=24$\\
    \hline
 		$500$	&boundary &$0.3$ 	& 0.007& 3.00		 	& 0.023  &  3.03 & 0.040  	& 3.25	& 0.009 &0.045 &0.113\\
 &	& $0.5$   															& 0.020& 3.29			&  0.025 &  3.35 & 0.039  	& 3.41&0.041  &0.059  &0.095\\
 &	 &   $0.7$	    												& 0.030 & 3.56			& 0.035  	& 3.56 & 0.040  	& 3.73&0.057  &0.066&0.093\\
  \hline
 		&$0.01$&$0.3$ 									& 0.006& 3.00		 	& 0.021  &  3.03 & 0.038  	& 3.25	& 0.008 &0.040 &0.103\\
 &	& $0.5$   															& 0.019& 3.30			&  0.023 &  3.36 & 0.036  	& 3.41&0.039  &0.055  &0.086\\
 &	 &   $0.7$	    												& 0.029 & 3.57			& 0.033  	& 3.58 & 0.037	  	& 3.75&0.046  &0.057&0.080\\
  \hline
 &$0.1$&$0.3$  										& 0.005 &  3.00			& 0.016 	& 3.03 & 0.022 	& 3.25&0.004  &0.023& 0.050 \\
 &	& $0.5$   														& 0.013 &  3.33			& 0.018  	& 3.38 & 0.025 	& 3.43&0.019 	&0.026&0.038 \\
 &	 &   $0.7$	    											& 0.019 &  3.65			& 0.023		& 3.65 & 0.026  	& 3.82 &0.014 & 0.017 &0.022\\
 \hline
$1000$ &boundary &$0.3$  	& 0.009&  3.01				& 0.019 	& 3.06		& 0.032 	& 3.30&0.013 & 0.037 & 0.079\\
 &	& $0.5$   															& 0.017 &  3.47		  & 0.023		& 3.44		& 0.031  	& 3.44&0.040 &0.049 & 0.066\\
 &	 &   $0.7$	    												& 0.029 &  3.84		 	& 0.034	 	& 3.93		& 0.040  	& 3.95&0.052 & 0.058	&0.075 \\
  \hline
 &$0.01$&$0.3$  										& 0.009&  3.01				& 0.019  	& 3.06		& 0.029 	& 3.30&0.014 & 0.033 & 0.075\\
 &	& $0.5$   													 & 0.017 &  3.48		  & 0.023	& 3.45		& 0.030  	& 3.44&0.038 &0.045 & 0.060\\
 &	 &   $0.7$	    											& 0.026 &  3.88		 	& 0.030	 	& 3.96		& 0.036  	& 3.98&0.041 & 0.050	&0.061 \\
  \hline
  &$0.1$&$0.3$  										& 0.006 &  3.02		 	& 0.013 	& 3.06	& 0.019 & 3.30&0.008	&	0.019 & 0.038 \\
 &	& $0.5$   														& 0.012  & 3.54		 	& 0.016		& 3.49	& 0.022 	& 3.48 &0.016 	& 0.018&0.022 \\
 &	 &   $0.7$	    											& 0.017  & 4.02			& 0.019 	& 4.09		& 0.024	& 4.10&0.008 	&	0.008 &0.010  \\
  \hline
$5000$ &boundary &$0.3$  		&0.021&  3.36				& 0.025  	& 3.38		& 0.029 	& 3.38&0.038 & 0.046 & 0.056\\
 &	& $0.5$   													 			& 0.033 &  3.54		  & 0.034		& 3.60		& 0.041  	& 3.79&0.051 &0.055 & 0.057\\
 &	 &   $0.7$	    													& 0.041 &  4.11		 	& 0.044	 	& 4.10		& 0.044  	& 4.07&0.052 & 0.055	&0.057 \\
  \hline
&$0.01$&$0.3$  									&0.020 &  3.36		 &0.024		&  3.39		& 0.028  	& 3.39			&0.037 &0.043		&0.052 \\
 &	& $0.5$   												&0.031 &  3.56		 	&0.033 	&   3.62	& 0.038  	& 3.80			&0.038 	&0.044 &0.046 \\
 &	 &   $0.7$	    									&0.038  &  4.18		 &0.039 	&  4.17		& 0.039 		& 4.14			&0.034 		&0.036 &0.039 \\
  \hline
  &$0.1$&$0.3$  								&0.016 &  3.39		 &0.018 	&   3.41			& 0.020 	&3.39		&0.019 &0.021 &0.025 \\
 &	& $0.5$   											&0.022 &  3.68			  &0.022	&   3.73			& 0.027  	& 3.91			&0.008 &0.009 	&0.008  \\
 &	 &   $0.7$	    								&0.023 &  4.46			  &0.025 	&   4.44		& 0.026  	& 4.40			&0.001		&0.001 &0.001 \\
  \hline
   \end{tabular}}
 \end{center}
 \caption{{\small Testing monotonicity---empirical size of our adaptive
test $\widehat {\mathtt{T}}_{n}$ and of the FS test (with $J=3$). Monte Carlo average value
$\widehat J$. Nominal level $\alpha =0.05$. True DGP from Section~\protect\ref{sim:monoton} using NPIV function \protect\eqref{design:h:null}. Instrument
strength increases in $\xi $.}}\label{table:size}
 \end{table}

\paragraph{Size.} Table~\ref{table:size} presents the average data-driven choice
of tuning parameter $J$, denoted by $\widehat J$. Specifically,
$\widehat J$ is the average choice of $J$ that maximizes
$\widehat{\mathcal{W}}_{J} (\alpha )$ over the RES index set
$\widehat{\mathcal I}_{n}$ when the null is not rejected; and is the smallest
$J\in \widehat{\mathcal I}_{n}$ such that
$\widehat{\mathcal{W}}_{J} (\alpha )>1$ when the null is rejected. This
data-driven choice of $J$ corresponds to \textit{early stopping} when the
null is rejected. Table~\ref{table:size} shows that, for the same sample
size $n$, the average data-driven choice $\widehat J$ increases as the
instrument strength (captured by the parameter $\xi $) increases; while
for the same instrument strength $\xi $, $\widehat J$ weakly increases
as the sample size $n$ increases. Table~\ref{table:size} also reports empirical
rejection probabilities under the null hypothesis using our adaptive test
$\widehat {\mathtt{T}}_{n}$ and the FS test. Our adaptive test is slightly
under-sized across different sample sizes $n\in \{500, 1000, 5000\}$, different
instrument strength $\xi \in \{0.3,0.5,0.7\}$, different degrees of monotonicity
$c_{0}\in \{0,0.01,0.1\}$ with $c_{0}=0$ being the ``boundary'' case. Table~\ref{table:size} shows that our adaptive test has empirical size control
for all $K(J) \in \{2J, 4J, 8J\}$, which is in line with our theoretical
results establishing asymptotic size control for any deterministic relation
of $K\sim cJ$ for some fixed constant $c\geq 1$. The difference between
the empirical size (of our adaptive test) for different choice of
$K(J)$ is small when $n$ is large or $\xi = 0.7$. While the FS test with
$K=5$ has empirical size control, the FS test with $K=12$ can be slightly
over-sized, and with $K=24$ can be heavily over-sized for all $\xi $ and
$n=500$, especially so for functions at or close to the boundary.\footnote{In
our previous version (arxiv:2006.09587v4), we implemented what we called
a \emph{nonadaptive bootstrap test} $T^{B}_{n,3}$ of \citet{fang2019}, which
is essentially their test, but uses empirical root-mean squared metric
instead of their trapezoid rule approximated
$\|\widehat h_{J}- h\|_{L^{2}(X)}$, and a cone projection onto a $J$-dimensional
sieve space instead of their optimization over grid points (for $x$). The
``nonadaptive bootstrap test'' $T^{B}_{n,3}$ has an empirical size closer
to that of our adaptive test.}

\begin{figure}[h!]
\begin{center}
    \includegraphics[scale=.7]{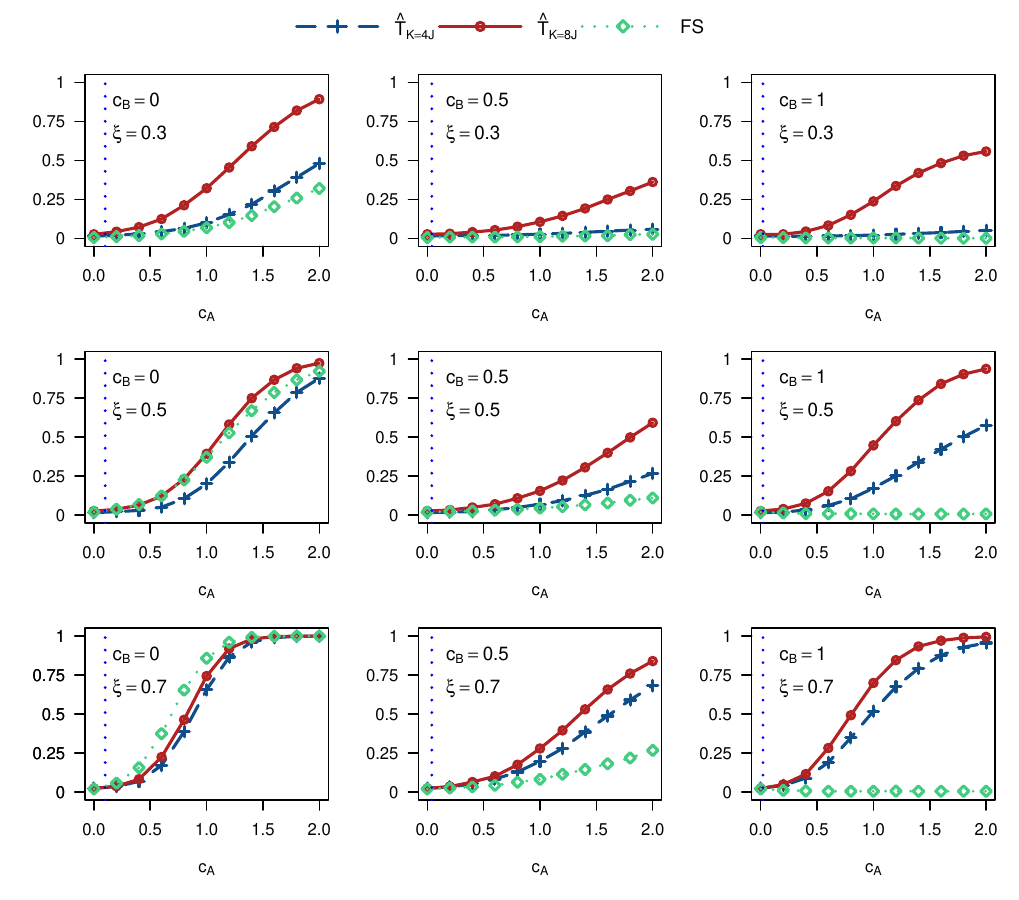}
         \vspace*{-5mm}
\caption{Testing monotonicity---empirical power of our adaptive
test $\widehat {\mathtt{T}}_{n}$ with $K(J)=4J$ (dashed plus lines) and
$K(J)=8J$ (solid circle lines) and the FS test (with $J=3$, $K=5$, dotted
square lines). True DGP from Section~\protect\ref{sim:monoton} using NPIV function
\protect\eqref{sim:h:sinus} with $n=500$. The vertical dotted line indicates when
the null hypothesis is violated. Alternatives are quadratic when
$c_{B}=0$ and become more complex as $c_{B}>0$ increases.}%
\label{fig:power-new}
\end{center}
\end{figure}

\paragraph{Power.} We next examine the rejection probabilities of our adaptive
test when the data are generated according to \eqref{NPIV:maintained} and
\eqref{design:I} using the NPIV function
%
\begin{align}
\label{sim:h:sinus} h(x)= -x/5+c_{A} \bigl(x^{2}+c_{B}
\sin (2\pi x) \bigr), 
\end{align}
where $c_{A} \in [0, 2]$ and $c_{B}\in \{0, 0.5, 1\}$. The null hypothesis
is that the NPIV function $h(\cdot )$ is weakly decreasing over the support
of $X$. When $c_{B}=0$, the null is satisfied only if
$c_{A}\leq 0.1$. When $c_{B}=0.5$, the null hypothesis is satisfied only
if $c_{A}\leq 0.1/(1+\pi /2)\approx 0.04$. When $c_{B}=1$, the null is
satisfied only if $c_{A}\leq 0.1/(1+\pi )\approx 0.02$.

Figure~\ref{fig:power-new} depicts the empirical power function of our
adaptive test $\widehat {\mathtt{T}}_{n}$ (dashed plus lines for
$K(J)=4J$ and solid circle lines for $K(J)=8J$), and of the FS test (dotted
square lines, $J=3$, $K=5$), under the 5\% nominal level for different instrument
strengths $\xi \in \{0.3, 0.5,0.7\}$, and sample size $n=500$.\footnote{The
finite-sample power of our adaptive test with $K(J)=2J$ is slightly smaller
than that with $K(J)=4J$ when $n=500$, but the power difference disappears
when $n$ becomes larger.} Figure~\ref{fig:power-new-5000} shows these power
curves for a larger sample size $n=5000$. From both figures, we see that
our adaptive test becomes more powerful for $c_{A}>0.1$ as the instrument
strength $\xi $ and the sample size $n$ increase. For weak instrument strength
$\xi =0.3$ and a small sample size (i.e., $n=500$), our adaptive test with
a larger $K(J)=8J$ is more powerful.

\begin{figure}[h!]
\begin{center}
    \includegraphics[scale=.7]{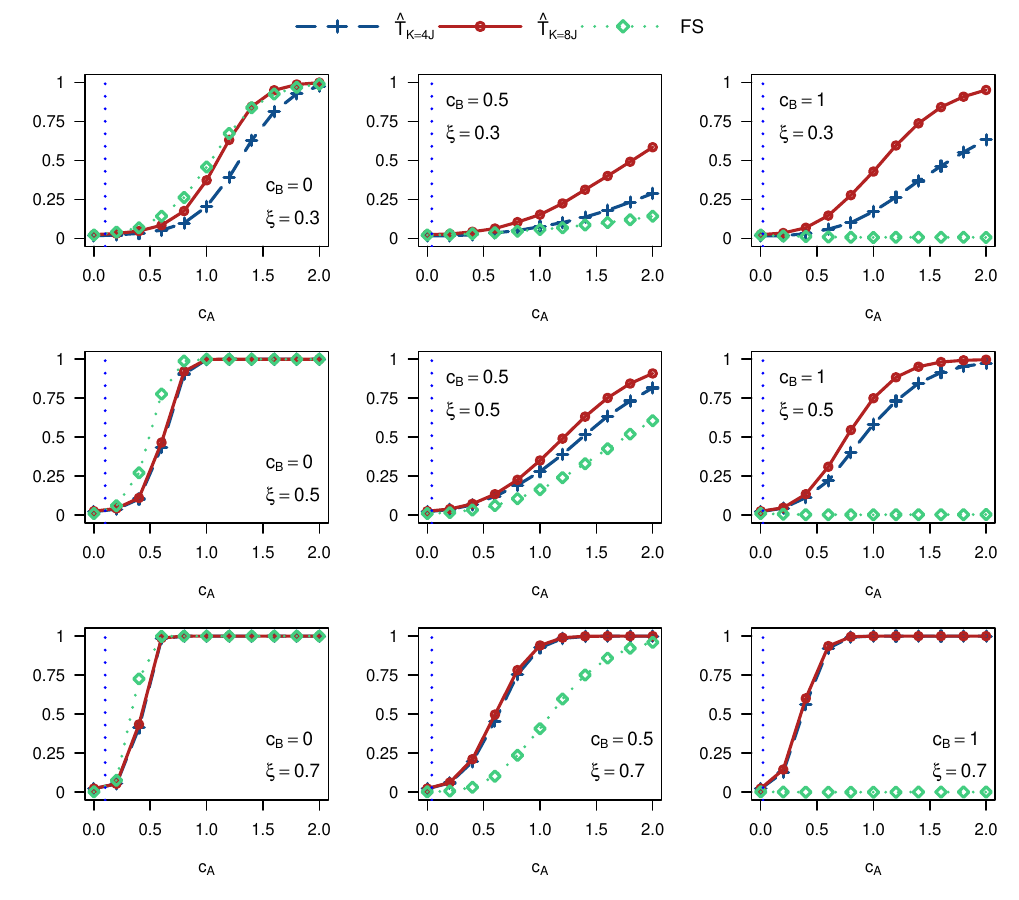}
         \vspace*{-5mm}
\caption{Testing monotonicity---replication of Figure~\protect\ref{fig:power-new} with $n=5000$.}%
\label{fig:power-new-5000}
\end{center}
\end{figure}

Figures~\ref{fig:power-new} and \ref{fig:power-new-5000} highlight the
importance of adaptation for the power of nonparametric monotonicity tests.
When the alternative is of a simple quadratic form (i.e., $c_{B} =0$),
there is little difference between our adaptive test
$\widehat {\mathtt{T}}_{n}$ and the FS test. But, as the alternative becomes
more nonlinear when $c_{B}>0$ increases, the FS test becomes much less
powerful than our adaptive test. This shows that a test with a tuning parameter
$J$ that is a deterministic nondecreasing function of $n$ can be powerful
in a certain direction but not for other nonlinear deviations.

In Supplemental Appendix~C, we present another
simulation design, which is based on an NPIV monotonicity design of
\citet{CW2017}. Simulation results using that design reveal that the empirical
size and power of our adaptive test have patterns very similar to the ones
reported in this subsection.

\subsection{Testing for Parametric Restrictions}
\label{sim:par}

We now test for a parametric specification. We assume that the data are
generated according to the design \eqref{NPIV:maintained} and
\eqref{design:I} with the NPIV function $h$ given by
\eqref{sim:h:sinus} with $c_{A} \in [0,4]$ and
$c_{B}\in \{0, 0.5\}$. The null hypothesis is $h$ being linear (i.e.,
$c_{A}=c_{B}=0$).

We implement our adaptive test $\widehat {\mathtt{T}}_{n}$ given in
\eqref{def:test:unknown} using quadratic B-spline basis functions with
varying number of knots and where the constrained function coincides with
the parametric 2SLS estimator. The number of knots varies within the RES
index set $\widehat{\mathcal I}_{n}$ as implemented in the last subsection,
with $K(J)\in \{2J, 4J, 8J\}$. We compare our adaptive test to the asymptotic
$t$-test and the test by \citet{Horowitz2006} (denoted by JH).\footnote{\citet{Horowitz2006}
already demonstrated in his simulation studies, with a sample size
$n=500$ and $1000$ Monte Carlo replications, that his test is more powerful
than several existing tests including \citet{bierens1990}'s.} To compute
the JH test that involves kernel density estimation, we follow
\citet{Horowitz2006} to estimate the joint density $f_{XW}$ using the kernel
$K(v)=(15/16)(1 - v^{2})^{2}\mathbbm{1}\{|v|\leq 1\}$, with the
kernel bandwidth chosen via cross-validation minimizing mean squared error
of estimating $f_{XW}$.


\begin{table}[h!]
 \begin{center}
 \renewcommand{\arraystretch}{1.05}
{\footnotesize  \begin{tabular}{c|c||cc|cc|cc|c|c}
\hline
 \textit{$n$}  & $\xi$ & $\widehat {\mathtt{T}}_n$, $K(J)=2J$& $\widehat J$ &$\widehat {\mathtt{T}}_n$, $K(J)=4J$&$\widehat J$&$\widehat {\mathtt{T}}_n$, $K(J)=8J$&$\widehat J$&$t$-test & JH test  \\
    \hline
     $500$
      	& $0.3$   										& 0.008& 	 3.00 		& 0.021		& 	 3.03&  0.040		& 	 3.29& 	0.001&0.049\\
 	& $0.5$   													& 0.022& 	 3.32 		& 0.024		& 	 3.40	& 0.037		& 	 3.46& 0.027&0.054\\
  &   $0.7$	   												& 0.036	&  3.61 		& 0.037		& 	 3.63	& 0.035		& 	 3.81&		0.045&0.058\\
 \hline
 $1000$& $0.3$   							& 0.014& 	3.01 		& 0.024			& 	 3.08	& 0.032		& 	 3.33& 0.006 &0.060\\
 	& $0.5$   													& 0.025& 	3.52 		& 0.033			& 	 3.49	& 0.033		& 	 3.48& 0.043&0.060\\
  &   $0.7$	   												& 0.036	&  3.91  	& 0.039			& 	 4.03	& 0.042		& 	 4.06& 0.046&0.053\\
 \hline
 $5000$  	& $0.3$   					& 0.022 & 3.38  	& 0.029			& 	 3.41	& 0.037		& 	 3.43& 0.032 & 0.057 \\
  	& $0.5$   												& 0.043 & 3.58 		& 0.048			& 	 3.65	& 0.045		& 	 3.85& 0.050 & 0.061\\
  &   $0.7$	   												& 0.050	&  4.17 		& 0.051				& 	 4.15	& 0.050		& 	4.14& 0.049 &0.055\\
  \hline
 \end{tabular}}
 \end{center}
  \vspace*{-5mm}
 \caption{{\small Testing parametric form---empirical size of our adaptive
test $\widehat {\mathtt{T}}_{n}$, the $t$-test, and JH test. Monte Carlo
average value $\widehat J$. Nominal level $\alpha =0.05$. True DGP from
Section~\protect\ref{sim:par} using NPIV function \protect\eqref{sim:h:sinus} with
$c_{A}=c_{B}=0$. Instrument strength increases in $\xi $.}}\label{table:size:par}
 \end{table}

\paragraph{Size.} Table~\ref{table:size:par} reports empirical rejection probabilities
of several tests under the null hypothesis of linearity of $h$. Results
are presented under different sample sizes
$n\in \{500, 1000, 5000\}$ and instrument strength
$\xi \in \{0.3, 0.5,0.7\}$. It also reports our adaptive test with different
$K(J)$ and $\widehat J$ (which is defined the same way as that in Table~\ref{table:size}). We note that $\widehat J$ is again weakly increasing
with sample size and with instrument strength. While the JH test can be
slightly over-sized, our adaptive test $\widehat {\mathtt{T}}_{n}$ provides
adequate size control across different sample size $n$, different instrument
strength $\xi $, and different $K(J)$. The difference in empirical size
of our adaptive test with different $K(J)$ is again small for large
$n$, which is consistent with our theory.

\begin{figure}[h!]
    \includegraphics[scale=.7]{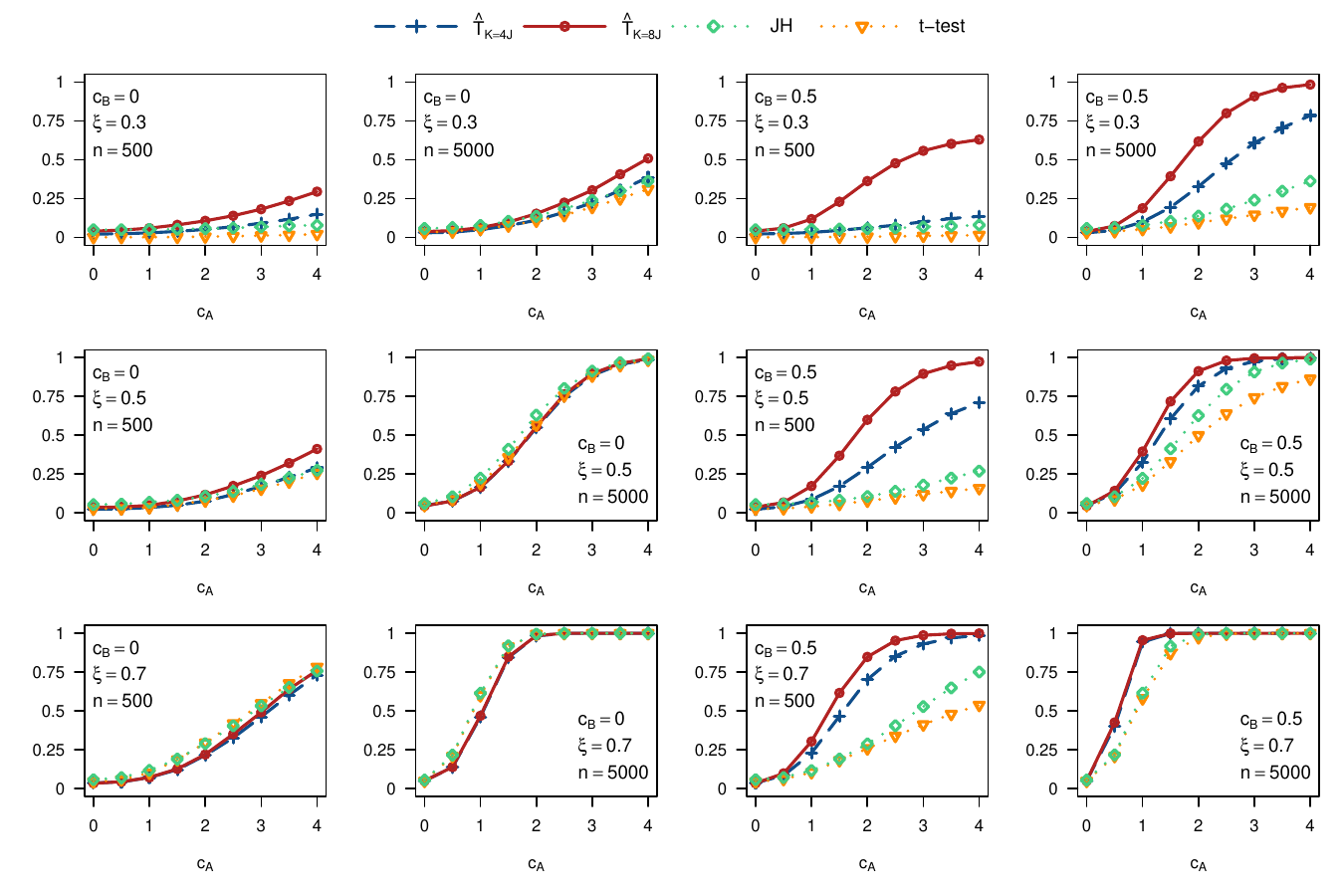}
         \vspace*{-5mm}
\caption{Testing parametric form---empirical power of our adaptive
test $\widehat {\mathtt{T}}_{n}$ with $K(J)=4J$ (dashed plus lines) and
$K(J)=8J$ (solid circle lines), of JH test (dotted square lines), and of
$t$-test (dotted triangle lines). True DGP from Section~\protect\ref{sim:par} using
NPIV function \protect\eqref{sim:h:sinus}. Alternatives are quadratic when
$c_{B}=0$ and more complex for $c_{B}=0.5$.}%
\label{fig:power:par}
\end{figure}

\paragraph{Power.}
 Figure~\ref{fig:power:par} provides empirical power curves
for the 5\% level tests with sample sizes $n\in \{500,5000\}$. From this
figure, we see that our adaptive test $\widehat {\mathtt{T}}_{n}$ (dashed
plus lines with $K(J)=4J$ and solid circle lines with $K(J)=8J$) has power
similar to the asymptotic $t$-test (dotted triangle lines) and the JH test
(dotted square lines) for a simple quadratic alternative with
$c_{B}=0$. When the alternative function in \eqref{sim:h:sinus} becomes
more nonlinear/complex with $c_{B}= 0.5$, our adaptive test becomes more
powerful than the JH test. This is theoretically sensible since the
\citet{Horowitz2006} test is designed to have power against
$n^{-1/2}$ smooth alternative only. Since our adaptive test is slightly
under-sized for small sample sizes or for weak instrument strength, the
size-adjusted empirical power of our test is even better (see our arxiv:2006.09587v3
version). To sum up, our adaptive minimax test not only controls size,
but also has very good finite-sample power uniformly against a large class
of nonparametric alternatives.

Finally, in Supplemental Appendix~C, we present
additional simulation comparisons of our adaptive test against our adaptive
version of \citeauthor{bierens1990}'s (\citeyear{bierens1990}) type test when the dimension of conditional
instrument $W$ is larger than the dimension of the endogenous variables
$X$. We observe that our adaptive test $\widehat {\mathtt{T}}_{n}$ again
has size control and even better finite-sample power when
$d_{w} > d_{x}$.

\section{Empirical Applications}
\label{sec:empirical}

We present two empirical applications of our adaptive test for NPIV models.
The first one tests for connected substitutes restrictions in differentiated
products demand using market level data. The second one tests for monotonicity,
convexity, or parametric specification of Engel curves for non-durable
good consumption using household level data. The applications demonstrate
that our simple adaptive test is powerful to detect economic shape restrictions.

In both empirical applications, we implement our adaptive test
$\widehat {\mathtt{T}}_{n}$ given in \eqref{def:test:unknown} with
$K(J)=4J$. The null hypothesis is rejected at the nominal level
$\alpha =0.05$ whenever $\widehat {\mathcal{W}}_{J} (\alpha )>1$ for some
$J\in \widehat{\mathcal I}_{n}$ (the RES index set). Let
$\widehat{\mathcal J}$ be
$\{J\in \widehat{\mathcal I}_{n}: \widehat {\mathcal{W}}_{J}(\alpha )>1
\}$ when our test rejects the null, and be
$\arg \max_{J\in \widehat{\mathcal I}_{n}} \widehat {\mathcal{W}}_{J}(
\alpha )$ when our test fails to reject the null. Let $\widehat J$ be the
minimal integer of
$\widehat{\mathcal J} \subset \widehat{\mathcal I}_{n}$. Tables in this
section report $\widehat{\mathcal J} $ and
$\widehat {\mathcal{W}}_{\widehat J}$. We also report the corresponding
$p$ value, which should, by Bonferroni correction, be compared to the nominal
level $\alpha =0.05$ divided by the cardinality of
$\widehat{\mathcal I}_{n}$. Finally, since our test is based on a leave-one-out
version, the value of $\widehat {\mathcal{W}}_{\widehat J}$ could be negative.

\subsection{Adaptive Testing for Connected Substitutes in Demand for Differential Products}
\label{subsec:scanner:data}

\citet{berry2014} provided conditions under which a nonparametric demand
system for differentiated products can be inverted to NPIV equations using
market level data. A key restriction is what they called ``connected substitutes.''
\citet{compiani2018} applied their nonparametric identification results
and estimated the system of inverse demand by directly imposing the connected
substitutes restrictions in his implementation of sieve NPIV estimator,
and obtained informative results as an alternative to BLP demand in simulation
studies and a real data application.

We revisit \citet{compiani2018}'s empirical application using the 2014
Nielsen scanner data set that contains market (store/week) level data of
consumers in California choosing from organic strawberries, non-organic
strawberries, and an outside option. While \citet{compiani2018} directly
imposed ``connected substitutes'' restriction in his sieve NPIV estimation
of inverse demand, we want to test this restriction. Following
\citet{compiani2018}, we consider
\begin{align*}
X_{o}+U=h(\textbf{P}, S_{o}, S_{no},\mathit{In}),\qquad
\Evtex [U|\textbf{W}_{p}, X_{o}, X_{no}, \mathit{In}]=0,
\end{align*}
where $h$ denotes the inverse of the demand for organic strawberries,
$X_{o}$ denotes a measure of taste for organic products, $X_{no}$ denotes
the availability of other fruit, $S_{o}$ and $S_{no}$ denote the endogenous
shares of the organic and non-organic strawberries, respectively.
$(X_{o}, X_{no})$ are the two included instruments for the two endogenous
shares $(S_{o}, S_{no})$. $In$ denotes store level (zip code) income and
$U$ unobserved shocks for organic produce. The vector
$\textbf{P}=(P_{o},P_{no},P_{\mathrm{out}})$ denotes the endogenous prices of organic
strawberries, non-organic strawberries, and non-strawberry fresh fruit,
respectively. We follow \citet{compiani2018} and let
$\textbf{W}_{p}=(W_{o},W_{no},W_{\mathrm{out}},W_{s1},W_{s2})$ be a five-dimensional
vector of conditional instruments for the price vector $\textbf{P}$, including
three Hausman-type instrumental variables $(W_{o},W_{no},W_{\mathrm{out}})$ and
two shipping-point spot prices $(W_{s1},W_{s2})$ (as proxies for the wholesale
prices faced by retailers).

As shown by \citet [Lemma~1]{compiani2018}, the connected substitutes assumption
of \citet{berry2014} implies the following shape restrictions on the function
$h$: First, $h$ is weakly increasing in the organic product price
$P_{o}$. Second, $h$ is weakly increasing in the organic product share
$S_{o}$. Third, $h$ is weakly increasing in the non-organic product share
$S_{no}$. Fourth,
$\partial h/\partial s_{o}\geq \partial h/\partial s_{no}$ (the so-called
diagonal dominance). Below, we test for these inequality restrictions.

We use the data set of \citet{compiani2018},\footnote{For details on the
construction of the data and descriptive statistics, see
\citet [Appendix F]{compiani2018}.} where income ranges from the first
and to the third quartile of its distribution and prices for organic produces
are restricted to be above its 1st and below its 99th percentile. The resulting
sample has size $n= 11910$. We implement our adaptive test
$\widehat {\mathtt{T}}_{n}$ by making use of a semiparametric specification
of the function $h$: we consider the tensor product of quadratic B-splines
$\psi ^{J_{1}}(P_{o})$ and the vector
$(1,In, P_{no}, \psi ^{3}(S_{o}))$, where we use a cubic B-spline transformation
of $S_{o}$ without knots and without intercept, hence $J=6J_{1}$. The variables
$( P_{\mathrm{out}}, S_{no}, S_{no}P_{no}, S_{no}S_{o})$ are included additively
and we set $K(J)=4J$. We obtain the RES index set
$\widehat{\mathcal I}_{n}=\{22,28,34\}$.
%
\begin{table}[h]
\renewcommand{\arraystretch}{1.2}
\centering
{\footnotesize \begin{tabular}{c|cccc}
\hline
 $H_0$& $\widehat {\mathcal{W}}_{\widehat J}$ & $p$ val. & reject $H_0$?&$\widehat{\mathcal J}$
\\ \hline\hline
$\partial h/\partial p_o\geq 0$ &0.854 &  0.031 & no &$\{34\}$ 	\\
$\partial h/\partial p_o\leq 0$ &3.154 &  0.000 & yes &$\{28,34\}$ 	\\
\hline
$\partial h/\partial s_o\geq 0$ &0.661 &  0.057 & no &$\{34\}$ 	\\
$\partial h/\partial s_o\leq 0$ &2.022 &  0.001 & yes &$\{22, 28,34\}$ 	\\
\hline
$\partial h/\partial s_{no}\geq 0$ &-0.115 & 0.471 & no &$\{22\}$ 	\\
$\partial h/\partial s_{no}\leq 0$ &-0.238 &  0.734 & no &$\{22\}$ 	\\
\hline
$\partial h/\partial s_o\geq \partial h/\partial s_{no}$ &0.663 &  0.057 & no &$\{34\}$ 	\\
$\partial h/\partial s_o\leq \partial h/\partial s_{no}$ &2.022 &  0.001 & yes &$\{22, 28,34\}$ 	\\
\hline
\end{tabular}}%
  \vspace*{-2mm}
\caption{Adaptive testing for the shape of $h$ (the inverse demand for organic produce).}
\label{tab:scanner:data}
\end{table}

According to Table~\ref{tab:scanner:data}, at the nominal level
$\alpha =0.05$, our adaptive test fails to reject that $h$ is weakly increasing
in the own price (but rejects $\partial h/\partial p_{o}\leq 0$), and fails
to reject that $h$ is weakly increasing in the own share (but rejects
$\partial h/\partial s_{o}\leq 0$). Our test fails to reject that
$h$ is weakly increasing or decreasing in the non-organic share (i.e.,
fails to reject a constant partial effect of $h$ with respect to the non-organic
share). Our test also fails to reject the diagonal dominance (but rejects
$\partial h/\partial s_{o}\leq \partial h/\partial s_{no}$). In summary,
our adaptive test provides strong empirical evidence for the connected
substitutes restriction.

\subsection{Adaptive Testing for Engel Curves}
\label{sec:emp}

The system of Engel curves plays a central role in the analysis of consumer
demand for non-durable goods. It describes the $i$th household's budget
share $Y_{\ell ,i}$ for non-durable goods $\ell $ as a function of its
log-total expenditure $X_{i}$ and other exogenous characteristics such
as family size and age of the head of the $i$th household. The most popular
class of parametric demand systems is the almost ideal class, pioneered
by \citet{deaton1980}, where budget shares are assumed to be linear in
log-total expenditure. \citet{banks1997} proposed a popular extension of
this system of linear Engel curves to include a squared term in log-total
expenditure, and their parametric Student $t$-test rejects linear form
in favor of quadratic Engel curves.

\citet{BCK07econometrica} estimated a system of nonparametric Engel curves
as functions of endogenous log-total expenditure and family size, using
log-gross earnings of the head of household as a conditional instrument
$W$. We use a subset of their data from the 1995 British Family Expenditure
Survey, with the head of household aged between 20 and 55 and in work,
and household with one or two children. This leaves a sample of size
$n = 1027$. As an illustration we consider Engel curves
$h_{\ell} (X)$ for four non-durable goods $\ell $: ``food in,'' ``fuel,''
``travel,'' and ``leisure'': $\Evtex [Y_{\ell} - h_{\ell} (X) | W]=0$. We use
the same quadratic B-spline basis with up to three knots to approximate
all the Engel curves and set $K(J)=4J$. Hence, the RES index set
$\widehat{\mathcal I}_{n}=\{3,4,5\}$ is the same for the different Engel
curves.

\begin{table}[h]
\renewcommand{\arraystretch}{1.2}
\centering
{\footnotesize \begin{tabular}{c|cccc|cccc}
\hline
& \multicolumn{4}{c|}{$H_0$: $h$ is increasing}& \multicolumn{4}{c}{$H_0$: $h$ is decreasing}\\
\hline
Goods&  $\widehat {\mathcal{W}}_{\widehat J}$ & $p$ value &reject $H_0$? & $\widehat{\mathcal J}$& $\widehat {\mathcal{W}}_{\widehat J}$ & $p$ value&reject $H_0$? & $\widehat{\mathcal J}$
\\ \hline\hline
``food in''& 2.871 &0.000	 & yes& $\{3\}$			&-0.324&0.852&no 	&$\{4\}$	\\
``fuel''& 8.192 & 0.000 & yes& 		$\{3,4,5\}$			& 0.547&0.072 &no		&$\{3\}$	\\
``travel''&2.527 & 0.000 & yes &		$\{3,4\}$			&0.381&0.124&no		&$\{3\}$	\\
``leisure''& 0.299 & 0.165 &no & 		$\{4\}$			&4.552 &0.000&yes&$\{3,4\}$\\
\hline
\end{tabular}}%
  \vspace*{-2mm}
\caption{Adaptive testing for monotonicity of Engel curves.}
\label{tab:test:data}
\end{table}
\begin{table}[h]
\renewcommand{\arraystretch}{1.2}
\centering
{\footnotesize \begin{tabular}{c|cccc|cccc}
\hline
& \multicolumn{4}{c|}{$H_0$: $h$ is convex}& \multicolumn{4}{c}{$H_0$: $h$ is concave}\\
\hline
Goods&  $\widehat {\mathcal{W}}_{\widehat J}$ & $p$ value &reject $H_0$? & $\widehat{\mathcal J}$& $\widehat {\mathcal{W}}_{\widehat J}$ & $p$ value&reject $H_0$? & $\widehat{\mathcal J}$
\\ \hline\hline
``food in''& -0.287 &0.791	 & no& $\{4\}$			&-0.324&0.853&no 	&$\{3\}$	\\
``fuel''&-0.325 & 0.844 & no& 		$\{3\}$			&1.621&0.001&yes		&$\{3\}$	\\
``travel''& 1.188 & 0.007 & yes &		$\{3\}$			&-0.322&0.837&no		&$\{5\}$	\\
``leisure''&-0.197 & 0.656 &no & 		$\{5\}$			&0.691 &0.047&no&$\{4\}$\\
\hline
\end{tabular}}%
  \vspace*{-2mm}
\caption{Adaptive testing for convexity/concavity of Engel curves.}
\label{tab:test:data:convex}
\end{table}
\begin{table}[h]
\renewcommand{\arraystretch}{1.2}
\centering
{\footnotesize \begin{tabular}{c|cccc|cccc}
\hline
& \multicolumn{4}{c|}{$H_0$: $h$ is linear}& \multicolumn{4}{c}{$H_0$: $h$ is quadratic}\\
\hline
Goods&  $\widehat {\mathcal{W}}_{\widehat J}$ & $p$ value&reject $H_0$?&$\widehat{\mathcal J}$ &  $\widehat {\mathcal{W}}_{\widehat J}$ & $p$ value&reject $H_0$?&$\widehat{\mathcal J}$
\\ \hline\hline
``food in''& -0.273 & 0.781&no &		$\{3\}$				&  0.125&0.272&no&$\{3\}$	\\
``fuel''&1.623 & 0.001 & yes&					$\{3\}$					&-0.120&0.540&no&$\{5\}$	\\
``travel''&1.210 & 0.006 &yes&				$\{3\}$					& -0.014&0.407&no&$\{4\}$	\\
``leisure''&0.691 & 0.047 &no &			$\{4\}$				&  0.513 &0.086&no&$\{4\}$	\\
\hline
\end{tabular}}%
  \vspace*{-2mm}
\caption{Adaptive testing for linear/quadratic specification of Engel curves.}
\label{tab:test:data:parametric}
\end{table}
Table~\ref{tab:test:data} reports our adaptive test for weak monotonicity
of Engel curves. It shows that our test rejects increasing Engel curves
for ``food in,'' ``fuel,'' and ``travel'' categories, and also rejects
decreasing Engel curve for ``leisure'' at the $0.05$ nominal level. Previously,
to decide whether the Engel curves are strictly monotonic, estimated derivatives
of these functions together with their non-adaptive 95\% uniform confidence
bands were also provided in \citet [Figure~4]{ChenChristensen2017}. Those
uniform confidence bands are constructed using sieve score bootstrapped
critical values with non-data-driven choice of sieve dimension $J$, and
contain zero almost over the whole support of household expenditure. It
is interesting to see that our adaptive test is more informative about
monotonicity in certain directions that are not obvious from their 95\%
uniform confidence bands. Table~\ref{tab:test:data:convex} reports our
adaptive test for convexity and concavity of these Engel curves. At the
5\% nominal level, we reject convexity of travel goods and reject concavity
of Engel curves for fuel consumption. These are in line with
\citet [Figure~4]{ChenChristensen2017}, but again, statistically significant
statements about the convexity/concavity of Engel curves are only possible
using our adaptive testing procedure. Finally, Table~\ref{tab:test:data:parametric} presents our adaptive tests for linear or
quadratic specifications (against nonparametric alternatives) of the Engel
curves for the four goods. At the nominal level $\alpha =0.05$, this table
shows that our adaptive test fails to reject a quadratic form for all the
goods, while it rejects a linear Engel curve for fuel and travel goods.
Our results are consistent with the conclusions obtained by
\citet{banks1997} using Student $t$-test for linear against quadratic forms
of Engel curves.


\begin{appendix}
\section{Proofs of Theorems \protect\ref{thm:minimax:test:lower} and \protect\ref{thm:test:upper} in Section~\protect\ref{sec:minimax}}
\label{sec:appendix}

\begin{proof}[{Proof of Theorem~\ref{thm:minimax:test:lower}.}]
We first derive the lower bound for testing a simple null hypothesis
$\mathcal H_{0} = \{h_{0}\}$. Let $\mathrm{P}_{\theta}$ denote the joint
distribution of $(Y,X,W)$ satisfying $Y= T h_{\theta}+V$ with known operator
$T$ and $V|W\sim \mathcal N(0,\sigma ^{2})$, the so-called reduced-form
nonparametric indirection regression (NPIR) model as in
\citet{ChenReiss2011} with fixed variance $\sigma ^{2}>0$. We may assume that $\{\lambda_j, \widetilde\psi_j, \widetilde b_j\}$ forms a singular value decomposition of the compact operator $T$. To establish
the lower bound, a consideration of the NPIR model is sufficient, as we
show in the first inequality of \eqref{ineq:lower} below.

By \citet{reiss2008}, the reduced-form NPIR is asymptotic equivalent to
the Gaussian white noise model
$dY(w)= Th_{\theta}(w)\,dw + \frac{\sigma}{\sqrt n}\,dB(w)$ where $dB$ is
a Gaussian white noise in
$L_{\mathcal W}^{2}:=\{\phi :\int _{\mathcal W} [\phi (w)]^{2}\,dw<
\infty \} $ and, in particular, to the Gaussian sequence model
$y_{k}= \int Th_{\theta}(w) \widetilde b_{k}(w)\,dw+
\frac{\sigma}{\sqrt n} \xi _{k}$,
$y_{k}:= \int \widetilde b_{k}(w)\,dY(w)$ and
$\xi _{k}\sim \mathcal N(0,1)$. Without loss of generality, we let
$h_{0}=0$ and $\mathcal H_{0} = \{0\}$. We introduce
$\theta =(\theta _{j})_{j\geq 1}$ with $\theta _{j}\in \{-1,1\}$ and introduce
the test function
\begin{align}
\label{test:fct} h_{\theta}(\cdot )=\frac{\delta _{*}}{ \sqrt{n}}
\sum_{j=1}^{J_{*}} \nu _{j}^{-2}
\theta _{j}\widetilde \psi _{j}(\cdot ) \Biggl(\sum
_{j=1}^{J_{*}} \nu _{j}^{-4}
\Biggr)^{-1/4}, 
\end{align}
for some sufficiently small $\delta _{*}:=\delta _{*}(\alpha )>0$. Here,
$\{\widetilde \psi _{j}\}_{j\geq 1}$ forms an orthonormal basis in
$L^{2}(X)$ and the dimension parameter $J_{*}$ satisfies the inequality
restriction
\begin{align}
\label{lower:bound:J} \frac{1}{ n} \Biggl(\sum_{j=1}^{J_{*}}
\nu _{j}^{-4} j^{4p/d_{x}} \Biggr)^{1/2}\leq
C_{\mathcal H}^{2}. 
\end{align}
Therefore, orthonormality of the basis functions
$\{\widetilde \psi _{j}\}_{j\geq 1}$ in $L^{2}(X)$ together with the Cauchy--Schwarz
inequality implies for any $\theta \in \{\pm 1\}^{J}$ with any
$J\geq J_{*}$:
\begin{equation*}
\sum_{j=1}^{\infty }\langle h_{\theta},
\widetilde\psi _{j}\rangle _{X}^{2}j^{2p/d_{x}}=
\frac{\delta _{*}^{2}}{n}\sum_{j=1}^{J_{*}}
\nu _{j}^{-4} j^{2p/d_{x}} \Biggl(\sum
_{l=1}^{J_{*}} \nu _{l}^{-4}
\Biggr)^{-1/2} \leq \frac{\delta _{*}^{2}}{n} \Biggl(
\sum_{j=1}^{J_{*}} \nu _{j}^{-4}
j^{4p/d_{x}} \Biggr)^{1/2}\leq C_{\mathcal H}^{2}
\end{equation*}
for all $\delta _{*}\in (0,1]$, and thus, we conclude that
$h_{\theta}\in \mathcal H$ by the definition of the Sobolev ellipsoid
$\mathcal H$. For any $\theta \in \{\pm 1\}^{J_{*}}$, we have
\begin{equation}
\label{test:in:alt} \llVert h_{\theta}-\mathcal H_{0} \rrVert
_{L^{2}(X)}= \llVert h_{\theta} \rrVert _{L^{2}(X)} = \frac{\delta
_{*} }{ \sqrt n} \Biggl(\sum_{j=1}^{J_{*}}
\nu _{j}^{-4} \Biggr)^{1/4} =\delta _{*}
r_{n}, 
\end{equation}
and hence, $h_{\theta}\in \mathcal H_{1}(\delta _{*} r_{n})$.

Let $\mathrm{P}^{*}$ denote the probability distribution obtained of the
NPIR model by assigning the uniform distribution on
$\{\pm 1\}^{J_{*}}$ and $\mathrm{P}_{0}$ the probability distribution when
$h_{\theta}=0$.
From the proof of \citet [Lemma~3]{collier2017}, we infer
the following reduction to testing between two probability measures under
a simple null hypothesis. Using that
$h_{\theta}\in \mathcal H_{1}(\delta _{*} r_{n})$ for all
$\theta \in \{\pm 1\}^{J_{*}}$, we thus evaluate
\begin{align}
&\inf_{\mathtt{T}_{n}} \Bigl\{\sup_{h\in \mathcal H_{0}}
\mathrm{P}_{h}( \mathtt{T}_{n}=1)+ \sup_{h\in \mathcal H_{1}(\delta _{*} r_{n})}
\mathrm{P}_{h}(\mathtt{T}_{n}=0) \Bigr\}\nonumber
\\
&\quad \geq \inf
_{\mathtt{T}_{n}} \Bigl\{\mathrm{P}_{0}(\mathtt{T}_{n}=1)+
\sup_{\theta \in \{\pm 1\}^{J_{*}}} \mathrm{P}_{\theta} (\mathtt{T}_{n}=0)
\Bigr\}
\nonumber
\\
&\quad \geq \inf_{\mathtt{T}_{n}} \bigl\{\mathrm{P}_{0}(
\mathtt{T}_{n}=1)+ \mathrm{P}^{*} (\mathtt{T}_{n}=0)
\bigr\}\nonumber
\\
&\quad \geq 1-\mathcal V\bigl(\mathrm{P}^{*}, \mathrm{P}_{0}
\bigr)\geq 1-\sqrt{\chi ^{2}\bigl(\mathrm{P}^{*},
\mathrm{P}_{0}\bigr)} \label{ineq:lower} , 
\end{align}
where $\mathcal V (\cdot , \cdot )$ denotes the total variation distance
and $\chi ^{2}(\cdot , \cdot )$ denotes the $\chi ^{2}$ divergence.

Since $T\widetilde \psi _{k}=\lambda_k \widetilde b _{k}$, we have
$y_{k}= \gamma _{k} \theta _{k}+\frac{\sigma}{\sqrt n} \xi _{k}$, where
$\gamma _{k}:=\delta _{*} n ^{-1/2}\lambda_k \nu _{k}^{-2} (\sum_{j=1}^{J_{*}} \nu _{j}^{-4} )^{-1/4}$. Consequently,
by the derivation of equation (2.106) in
\citet{tsybakov2009introduction}, the $\chi ^{2}$ divergence between
$\mathrm{P}^{*}$ and $\mathrm{P}_{0}$ satisfies
\begin{align*}
\chi ^{2}\bigl(\mathrm{P}^{*}, \mathrm{P}_{0}
\bigr)&=
\int \biggl( \frac{d\mathrm{P}^{*}}{ d\mathrm{P}_{0}}
\biggr)^{2}\,d\mathrm{P}_{0}-1= \prod
_{k=1}^{J_{*}} \frac{\exp \bigl(-n\gamma
_{k}^{2}/\sigma ^{2}\bigr)+\exp \bigl(n\gamma
_{k}^{2}/\sigma ^{2}\bigr)}{2}-1.
\end{align*}
By \citet [Section~2.7.5]{tsybakov2009introduction}, there exists a constant
$c_{1}>0$ such that
$\exp (-n\gamma _{k}^{2}/\sigma ^{2})+\exp (n\gamma _{k}^{2}/\sigma ^{2})
\leq 2\exp   (c_{1}n^{2}\gamma _{k}^{4}  )$.
Assumptions~\ref{A:LB}(iii), (iv) imply for a finite constant $c>0$ that $\lambda_j^2\leq c \nu _{j}^{2}$ for all $j$.
Consequently,
$\sum_{k= 1}^{J_{*}} \gamma_k^{4} \leq c^2\delta _{*}^{4} n^{-2}$, and we obtain:
\begin{align*}
\chi ^{2}\bigl(\mathrm{P}^{*}, \mathrm{P}_{0}
\bigr)\leq \exp \Biggl(c_{1} n^{2} \sum
_{k=1}^{J_{*}}\gamma _{k}^{4}
\Biggr) -1\leq \exp \bigl(\delta _{*}^{4} c_{1}
c^2 \bigr) -1\leq 1-\alpha ,
\end{align*}
for $\delta _{*}=\delta _{*}(\alpha )>0$ sufficiently small. Consequently,
the result follows by making use of inequality \eqref{ineq:lower}.

In the regularly varying case ($ \nu _{J_{*}}^{-4}J_{*} \lesssim
\sum_{j=1}^{J_{*}} \nu _{j}^{-4}$) for
$J_{*}\sim \max   \{J: n^{-1/2} J^{1/4} \nu _{J}^{-1}\leq J^{-p/d_{x}}
  \}$, we note that inequality \eqref{lower:bound:J} holds within a
constant and we have that
$r_{n}= n^{-1/2} (\sum_{j=1}^{J_{*}} \nu _{j}^{-4} )^{1/4}
\sim n^{-1/2}J_{*}^{1/4}\nu _{J_{*}}^{-1}\sim J_{*}^{-p/d_{x}}$. Consider
the mildly ill-posed case ($\nu _{j}=j^{-a/d_{x}}$). The choice of
$J_{*}\sim n^{2d_{x}/(4(p+a)+d_{x})}$ ensures constraint
\eqref{lower:bound:J} within a constant and implies
$r_{n} \sim n^{-2p/(4(p+a)+d_{x})}$. Consider the severely ill-posed case
($\nu _{j}=\exp (-j^{a/d_{x}}/2)$). The choice of
$J_{*}=  (c\log n  )^{d_{x}/a}$ satisfies
\eqref{lower:bound:J} within a constant and implies
$r_{n} \sim (\log n)^{-p/a}$, which completes the proof for the simple
null $\mathcal H_{0} =\{0\}$ case.

We now turn to the lower bound for testing a closed convex composite null
hypothesis. Consider the test function given in equation
\eqref{test:fct}. Since $\mathcal H_{0}$ is a nonempty, closed and convex,
strict subset of $\mathcal H$, there exists a unique element
$\Pi _{\mathcal H_0}h \in \mathcal H_{0}$ (by the Hilbert projection theorem)
such that
\begin{align}
\label{ineq:LB:CT} \llVert h_{\theta}-\mathcal H_{0} \rrVert
_{L^{2}(X)}= \llVert h_{\theta}-\Pi _{
\mathcal H_{0}} h_{\theta}
\rrVert _{L^{2}(X)} \geq \llVert h_{\theta _{*}}-\Pi _{
\mathcal H_{0}}
h_{\theta _{*}} \rrVert _{L^{2}(X)} 
\end{align}
for some $\theta _{*}\in \{\pm 1\}^{J_{*}}$. As above, we may assume
$\Pi _{\mathcal H_{0}} h_{\theta _{*}}=0$ without loss of generality (otherwise,
consider $\widetilde Y=Y-T\Pi _{\mathcal H_{0}} h_{\theta _{*}}$ in the
reduced-form NPIR model). Given the inequality \eqref{ineq:LB:CT}, we thus
conclude
$\|h_{\theta}-\mathcal H_{0}\|_{L^{2}(X)} \geq \|h_{\theta _{*}}\|_{L^{2}(X)}
\geq \delta _{*} r_{n}$, by following inequality \eqref{test:in:alt}. Therefore,
we may proceed with the proof of the lower bound as for the simple null
case.
\end{proof}

\begin{lemma}%
\label{lemma:normal}
Let Assumptions \ref{A:LB}(i)--(iii) and \ref{A:basis} hold. Then, under
the simple hypothesis $\mathcal H_{0}=\{h_{0}\}$ for a known function
$h_{0}$, we have
$\mathrm{P}_{h_{0}}  (n\widehat{D}_{J}(h_{0})/\widehat{V}_{J}>
\eta _{J}(\alpha )  )=\alpha +o(1)$.
\end{lemma}
A proof of Lemma~\ref{lemma:normal} is given in Supplemental Appendix~E.

\begin{proof}[{Proof of Theorem~\ref{thm:test:upper}.}]
First, by Lemma~\ref{lemma:normal}, we control the type I error of the
test $\mathtt{T}_{n,J}$ given in \eqref{def:test:J}:
$\limsup_{n\to \infty}\mathrm{P}_{h_{0}}(\mathtt{T}_{n,J} =1)=
\limsup_{n\to \infty}\mathrm{P}_{h_{0}}  (n\widehat{D}_{J}(h_{0})>
\eta _{J}(\alpha )\widehat{V}_{J}  )\leq \alpha $. To control the
type II error, we have uniformly for
$h\in \mathcal H_{1}(\delta ^{\circ }r_{n,J})$,
\begin{align*}
\mathrm{P}_{h} (\mathtt{T}_{n,J}=0 )&\leq
\mathrm{P}_{h} \bigl(n\widehat{D}_{J}(h_{0}) \leq
\eta _{J}(\alpha )\widehat{V}_{J}, \widehat{V}_{J}
\leq (1+c_{0})V_{J} \bigr)+\mathrm{P}_{h} \bigl(
\widehat{V}_{J}> (1+c_{0}) V_{J} \bigr)
\\
&\leq \mathrm{P}_{h} \bigl(n\widehat{D}_{J}(h_{0})
\leq (1+c_{0}) \eta _{J}(\alpha )V_{J}
\bigr)+o(1)= o(1),
\end{align*}
where the second equation is due to Lemma~\ref{lem:est:var}(i) and the
last equation is due to Lemma~\ref{lemma:type2}(i) in Appendix~\ref{appendix:adapt:ST}. We thus obtain Result \eqref{thm:test:upper:bound}. Note that $\nu _{J}^{-2}\geq c s_{J}^{-2}$ by
Assumption~\ref{A:s_J}, with the definition of $J_{*0}$, the final rate
results for the mildly ill-posed case ($\nu _{j}= j^{-a/d_{x}}$) and for
the severely ill-posed case ($\nu _{j}= \exp (-j^{a/d_{x}}/2)$) follow
from $r_{n,J_{*0}}= (J_{*0})^{-p/d_{x}}$ directly.
\end{proof}

\section{Proofs of Theorems \protect\ref{thm:adapt:test} and \protect\ref{thm:adapt:est:test} in Section~\protect\ref{sec:adapt:test}}
\label{appendix:adapt:ST}

We first introduce additional notation. For a $r\times c$ matrix $M$ with
$r \leq c$ and full row rank $r$, we let $M_{l}^{-}$ denote its left pseudoinverse,
namely $(M'M)^{-}M'$. The $J\times K$ matrices $\widehat A$ and $A$ defined
in Section~\ref{subsec:sst} can be written as
$\widehat A=(\widehat G_{b}^{-1/2}\widehat S \widehat G^{-1/2})^{-}_{l}
\widehat G_{b}^{-1/2}$ and
$A=(G_{b}^{-1/2}S G^{-1/2})^{-}_{l}G_{b}^{-1/2}$. Then
$ \| AG_{b}^{1/2} \|= \|  (G_{b}^{-1/2}SG^{-1/2} )^{-}_{l}
 \|=s_{J}^{-1}$ with $s_{J}=s_{\min}(G_{b}^{-1/2}S G^{-1/2})>0$. Let
$\widetilde b^{K}(\cdot )=G_{b}^{-1/2}b^{K}(\cdot )$ and
$\widetilde \psi ^{J}(\cdot )=G^{-1/2}\psi ^{J}(\cdot )$. For any
$h\in L^{2}(X)$, its population 2SLS projection onto the sieve space
$\Psi _{J}$ is
\begin{equation}
\label{def:QJ} Q_{J} h(\cdot )=\widetilde \psi ^{J}(\cdot
)'A\Evtex \bigl[b^{K}(W)h(X)\bigr] = \widetilde \psi
^{J}(\cdot )'\bigl(G_{b}^{-1/2}S
G^{-1/2}\bigr)_{l}^{-}\Evtex \bigl[ \widetilde
b^{K}(W) h(X)\bigr] . 
\end{equation}
We next present Theorem~\ref{thm:rate:quad:fctl} and eight lemmas (Lemma~\ref{lemma:bias}--Lemma~\ref{lemma:J_hat}) that are used to establish our
adaptive testing upper bounds. The proofs of these results are postponed
to Supplemental Appendix~E. Below, we shorten ``with
probability $\mathrm{P}_{h}$ approaching 1 uniformly for
$h\in \mathcal H$'' to ``wpa1 uniformly for $h\in \mathcal H$.''
\begin{theorem}%
\label{thm:rate:quad:fctl}
Let Assumptions \ref{A:LB}(ii)--(iii) and \ref{A:basis} hold. Then, wpa1
uniformly for $h\in \mathcal H$:
\begin{equation*}
\widehat{D}_{J}(\Pi _{\mathcal H_0}h)- \bigl\llVert Q_{J}(h-
\Pi _{\mathcal H_0}h) \bigr\rrVert _{L^{2}(X)}^{2} \lesssim
n^{-1}s_{J}^{-2} \sqrt{J}+n^{-1/2}s_{J}^{-1}
\bigl( \llVert h-\Pi _{\mathcal H_0}h \rrVert _{L^{2}(X)}+J^{-p/d_{x}}
\bigr).
\end{equation*}
\end{theorem}
Theorem~\ref{thm:rate:quad:fctl} provides an upper bound for quadratic
distance estimation, which is essential for our upper bound on the minimax
rate of testing in $L^{2}$.

\begin{lemma}%
\label{lemma:bias}
Let Assumption~\ref{A:basis}(iv) hold. Then we have uniformly for
$h\in \mathcal H$: (i)
$\|Q_{J}(h-\Pi _{\mathcal H_0}h)\|_{L^{2}(X)}=\|h-\Pi _{\mathcal H_0}h
\|_{L^{2}(X)}+O(J^{-p/d_{x}})$ and (ii)
$\|Q_{J}h-h\|_{L^{2}(X)}=O(J^{-p/d_{x}})$.
\end{lemma}

\begin{lemma}%
\label{upper:bound:v_n}
Let Assumption~\ref{A:basis}(i) hold. Then:
$V_{J}\leq \overline\sigma ^{2}s_{J}^{-2}\sqrt{J}$ uniformly for
$h\in \mathcal H$ and $J\in \mathcal I_{n}$.
\end{lemma}

\begin{lemma}%
\label{lower:bound:v_n}
Let Assumption~\ref{A:LB}(i) hold. Then:
$J\leq \sum_{j=1}^{J}s_{j}^{-4}\leq \underline\sigma ^{-4} V_{J}^{2}$
uniformly for $h\in \mathcal H$ and $J\in \mathcal I_{n}$.
\end{lemma}

\begin{lemma}%
\label{lem:est:var}
Let Assumption~\ref{A:LB}(i)--(iii) be satisfied.
\begin{enumerate}[(ii)]

\item[(i)] If, in addition, Assumption~\ref{A:basis} holds, then for any
$c>0$, we have
\begin{align*}
\sup_{h\in \mathcal H}\mathrm{P}_{h} \bigl( \llvert 1-
\widehat{V}_{J}/V_{J} \rrvert >c \bigr)=o(1).
\end{align*}

\item[(ii)] If, in addition, Assumptions \ref{A:basis}(i) and
\ref{A:adapt:test}(i) hold, then for any $c>0$, we have
\begin{align*}
\sup_{h\in \mathcal H}\mathrm{P}_{h} \Bigl(\max
_{J\in \mathcal I_{n}} \llvert 1-\widehat{V}_{J}/V_{J}
\rrvert >c \Bigr)=o(1).
\end{align*}
\end{enumerate}%
\end{lemma}

\begin{lemma}%
\label{lemma:eta:bounds}
For all $\alpha \in (0,1)$ and $J\in \widehat{\mathcal I}_{n}$, we have
for $n$ sufficiently large and almost surely that
\begin{equation*}
\frac{\sqrt{\log \log (J)-\log (\alpha )}}{4}\leq \widehat \eta _{J}(
\alpha )\leq 4\sqrt{\log \log (n)-\log (\alpha )}.
\end{equation*}
\end{lemma}

For any $h\in \mathcal{H}$, let
$U_{i}^{J}:=Ab^{K}(W_{i})(Y_{i}-\Pi _{\mathcal H_0}h(X_{i}))$ with
$U_{ij}$ as its $j$th entry, $1\leq j\leq J$. Then
$Q_{J}(h-\Pi _{\mathcal H_0}h)=\Evtex _{h}[U^{J}]'\widetilde \psi ^{J}$ and
$\|\Evtex _{h}[U^{J}]\|^{2} =\|Q_{J}(h-\Pi _{\mathcal H_0}h)\|_{L^{2}(X)}^{2}$
for any NPIV function $h\in \mathcal{H}$. Let
$Z_{i}=(Y_{i},X_{i}',W_{i}')'$. For any set $D_{i}$, we define
\begin{align*}
R(Z_{i}, Z_{i'},D_{i}):=\bigl(U_{i}^{J}{
\mathbbm{1}}_{D_{i}}\bigr)'\bigl(U_{i'}^{J}{
\mathbbm{1}}_{D_{i'}}\bigr) - \Evtex _{h}\bigl(U_{i}^{J}{
\mathbbm{1}}_{D_{i}}\bigr)'\Evtex _{h}
\bigl(U_{i}^{J}{\mathbbm{1}}_{D_{i}}\bigr) ,
\end{align*}
$R_{1}(Z_{i}, Z_{i'}):=R(Z_{i}, Z_{i'},M_{i})$ and
$R_{2}(Z_{i}, Z_{i'}):=R(Z_{i}, Z_{i'},M_{i}^{c})$, where
$M_{i}=\{|Y_{i}-\Pi _{\mathcal H_0}h(X_{i})|\leq M_{n}\}$ and
$M_{n}=\sqrt{n} \zeta _{\overline J}^{-1} (\log \log \overline J)^{-3/4}$.
Let
\begin{align*}
\Lambda _{1}&:= \biggl(\frac{n(n-1)}{2}\Evtex \bigl[R_{1}^{2}(Z_{1},Z_{2})
\bigr] \biggr)^{1/2},
\\
\Lambda _{2}&:=n\sup_{ \llVert \nu  \rrVert _{L^{2}(Z)}\leq 1,
 \llVert \kappa  \rrVert _{L^{2}(Z)}
\leq 1}
\Evtex \bigl[R_{1}(Z_{1},Z_{2})\nu
(Z_{1})\kappa (Z_{2})\bigr],
\\
\Lambda _{3}&:= \Bigl(n\sup_{z} \bigl\llvert
\Evtex \bigl[R_{1}^{2}(Z_{1},z)\bigr] \bigr\rrvert
\Bigr)^{1/2}, \quad \text{and}\quad  \Lambda _{4}:=\sup
_{z_{1},z_{2}} \bigl\llvert R_{1}(z_{1},z_{2})
\bigr\rrvert .
\end{align*}
\begin{lemma}%
\label{Lemma:houdre}
\begin{enumerate}[(ii)]
\item[(i)] There exists a generic constant $C_{R_{1}}>0$, such that for all
$u>0$ and $n\in \mathbb N$, we have
\begin{align*}
\mathrm{P}_{h} \biggl( \biggl\llvert \sum
_{1\leq i< i'\leq n} R_{1}(Z_{i}, Z_{i'})
\biggr\rrvert \geq C_{R_{1}} \bigl(\Lambda _{1}\sqrt{u} +
\Lambda _{2} u+\Lambda _{3} u^{3/2} +\Lambda
_{4} u^{2} \bigr) \biggr)\leq 6\exp (-u).
\end{align*}
\item[(ii)] Let Assumption~\ref{A:basis}(i) hold. Then, for the kernel
$R_{1}$, the following hold under $\mathcal H_{0}$:
\begin{eqnarray*}
\Lambda _{1} &\leq& \sqrt{n(n-1)/2} V_{J},\qquad  \Lambda
_{2} \leq \overline\sigma ^{2} n s_{J}^{-2}
,
\\
\Lambda _{3} &\leq& \overline\sigma ^{2}\sqrt{n}
M_{n} \zeta _{b,K} s_{J}^{-2}, \qquad \Lambda
_{4} \leq M_{n}^{2} \zeta _{b,K}^{2}
s_{J}^{-2} .
\end{eqnarray*}
\end{enumerate}%
\end{lemma}

\begin{lemma}%
\label{lemma:type2}\text{}
\begin{enumerate}[(ii)]
\item[(i)] Under the conditions of Theorem~\ref{thm:test:upper}, we have for some
constant $c_{0}>0$ that
$\mathrm{P}_{h} ( n\widehat{D}_{J}(h_{0})\leq (1+c_{0})\eta _{J}(
\alpha )V_{J} )=o(1)$ uniformly for
$h\in \mathcal H_{1}(\delta ^{\circ }r_{n,J})$.

\item[(ii)] Under the conditions of Theorem~\ref{thm:adapt:test}, we have
$\mathrm{P}_{h} (n\widehat{D}_{J^{*}}(h_{0})\leq 2c_{1}\sqrt{\log
\log n}  V_{J^{*}} )=o(1)$ uniformly for
$h\in \mathcal H_{1}(\delta ^{\circ}\text{\textsf r}_{n})$, where $J^{*}$ and
$c_{1}$ are given in the proof of Theorem~\ref{thm:adapt:test}.
\end{enumerate}\end{lemma}

\begin{lemma}%
\label{lemma:J_hat}
Let Assumption~\ref{A:adapt:test}(i)(iii) be satisfied. Then
$\widehat J_{\max}$ given in \eqref{def:J_max} satisfies%
\begin{enumerate}[(ii)]
\item[(i)]
$  \sup_{h\in \mathcal H}\mathrm{P}_{h}  (\widehat J_{\max}>
\overline J  )=o(1)$; and

\item[(ii)]
$  \sup_{h\in \mathcal H}\mathrm{P}_{h}  (2J^\circ >
\widehat J_{\max}  )=o(1)$ under Assumption \ref{A:s_J}.
\end{enumerate}\end{lemma}

\begin{proof}[{Proof of Theorem~\ref{thm:adapt:test}.}]
We prove this result in three steps. First, we bound the type I error of
the test statistic
$\widetilde{\mathtt{T}}_{n} = \mathbbm{1} \{ \max_{J
\in \mathcal I_{n}} (n\widehat{D}_{J}(h_{0})/(\eta _{J}'(\alpha )V_{J})
 )>1 \}$,
$\eta _{J}'(\alpha ):=(1-c_{0})\sqrt{\log \log J-\log \alpha}/4$ for some
constant $0<c_{0}<1$. Second, we bound the type II error of
$\widetilde {\mathtt{T}}_{n}$ where $\eta _{J}'(\alpha )$ is replaced by
$\eta ''(\alpha ):=4(1+c_{0})\sqrt{\log \log n-\log \alpha}$. Third, we
show that the derived bounds in Steps 1 and 2 are sufficient to control
the type I and type II errors of our adaptive test
$\widehat{\mathtt{T}}_{n}$ for a simple null hypothesis
$\mathcal H_{0}=\{h_{0}\}$.

 \textbf{Step~1:} To control the type I error of
$\widetilde {\mathtt{T}}_{n}$, we use a decomposition under
$\mathcal H_{0}=\{h_{0}\}$ via the U-statistic
$\mathcal U_{J,l}=\frac{2}{n(n-1)}\sum_{1\leq i< i'\leq n}R_{l}(Z_{i},
Z_{i'})$ for $l=1,2$ and $U_{i}=Y_{i}-h_{0}(X_{i})$:
\begin{align*}
\mathrm{P}_{h_{0}} (\widetilde {\mathtt{T}}_{n}=1 )\leq{}&
\mathrm{P}_{h_{0}} \Biggl(\max_{J\in \mathcal I_{n}} \Biggl\llvert
\frac{1}{\eta _{J}'(\alpha )V_{J}(n-1)}\sum
_{j=1}^{J}\sum_{i\neq i'}
U_{ij}U_{i'j} \Biggr\rrvert
\\
&{} +\max_{J\in \mathcal I_{n}} \biggl\llvert \frac{1}{\eta
_{J}'(\alpha )V_{J}(n-1)}\sum
_{i\neq i'}U_{i}U_{i'}b^{K}(W_{i})'
\bigl(A'A-\widehat A'\widehat A \bigr)b^{K}(W_{i'})
\biggr\rrvert >1 \Biggr)
\\
\leq{}& I + \mathit{II} + \mathit{III},
\end{align*}%
with $I:=\mathrm{P}_{h_{0}} (\max_{J\in \mathcal I_{n}} |n
\mathcal U_{J,1}/(\eta _{J}'(\alpha )V_{J}) |>\frac{1}{4} )$,
$\mathit{II}:=\mathrm{P}_{h_{0}} (\max_{J\in \mathcal I_{n}} |n
\mathcal U_{J,2}/(\eta _{J}'(\alpha )V_{J}) |> \frac{1}{4} )$,
\begin{align*}
\mathit{III}:=\mathrm{P}_{h_{0}} \biggl(\max_{J\in \mathcal I_{n}}
\biggl\llvert \frac{1}{\eta _{J}'(\alpha
)V_{J}(n-1)}\sum_{i\neq i'}
U_{i}U_{i'}b^{K}(W_{i})'
\bigl(A'A-\widehat A' \widehat A \bigr)b^{K}(W_{i'})
\biggr\rrvert >\frac{1}{2} \biggr) .
\end{align*}
First, we consider term $\mathit{III}$. Using the definition of
$\eta _{J}'(\alpha )$ and the fact that\break
$\sqrt{\log \log J-\log \alpha}>\sqrt{\log \log J}$ for any
$\alpha \in (0,1)$, we obtain $\mathit{III}=o(1)$ by applying Lemma~E.6.%

Next, we consider term $I$. Define
$\Lambda (u,J):=\Lambda _{1}\sqrt{u} +\Lambda _{2} u+\Lambda _{3} u^{3/2}
+\Lambda _{4} u^{2}$. By Lemma~\ref{Lemma:houdre}(ii) with
$M_{n}=\sqrt{n} \zeta _{\overline J}^{-1} (\log \log \overline J)^{-3/4}$,
we have for all $J\in \mathcal I_{n}$:
\begin{align*}
\Lambda (u,J) &\leq nV_{J}\sqrt{u/2} +\overline\sigma
^{2}ns_{J}^{-2} u+ \overline\sigma ^{2}n
s_{J}^{-2} (\log \log \overline J)^{-3/4}u^{3/2}+n
s_{J}^{-2} (\log \log \overline J)^{-3/2}u^{2}
\end{align*}
for $n$ sufficiently large. Replacing in the previous inequality $u$ by
$u_{J}=2\log \log J^{c_{\alpha}}$ where
$c_{\alpha}=\sqrt{1+(\pi /\log 2)^{2}}/\sqrt \alpha $, we obtain for
$n$ sufficiently large:
\begin{align*}
\Lambda (u_{J},J)&\leq nV_{J}\sqrt{\log \log
J^{c_{\alpha}}} + \frac{2\overline\sigma ^{2}n}{s_{J}^{2}}
\log \log J^{c_{\alpha}}+ \frac{\overline\sigma ^{2}n}{
s_{J}^{2}} \bigl(2\log \log J^{c_{\alpha}}
\bigr)^{3/4}+ \frac{4n}{ s_{J}^{2}} \sqrt{\log \log
J^{c_{\alpha}}}
\\
&\leq \frac{5}{4}nV_{J}\sqrt{\log \log J-\log \alpha} +3 \overline
\sigma ^{2}ns_{J}^{-2}(\log \log J-\log \alpha )
\\
&\leq \frac{5}{1-c_{0}}nV_{J} \eta _{J}'(
\alpha ) + \frac{12\overline\sigma ^{2}}{1-c_{0}}ns_{J}^{-2}
\eta _{J}'(\alpha ) \sqrt{\log \log J},
\end{align*}
by the definition of $\eta _{J}'(\alpha )$. Since
$s_{J}^{-2}\sqrt{J}\sim V_{J}$ uniformly in
$h\in \mathcal H$ and $J\in \mathcal I_{n}$ (by Assumption~\ref{A:adapt:test}(ii), Lemmas
\ref{upper:bound:v_n} and \ref{lower:bound:v_n}), we have
$V_{L}/V_{J}\lesssim s_{L}^{-2}s_{J}^{2} \sqrt{ L/J }=o(1)$ for all
$L=o(J)$ uniformly in
$h\in \mathcal H$ and $J\in \mathcal I_{n}$. Thus, for all $J\in \mathcal I_{n}$ and for $n$ sufficiently
large:
$\Lambda (u_{J},L(J))\leq C_{R_{1}}\frac{n-1}{8} V_{J} \eta _{J}'(
\alpha )$ with $L(J)=\exp (1/6) J\underline J^{-1/2}$. By Lemma~\ref{Lemma:houdre}(i) with $u=2\log \log J^{c_{\alpha}}$ and the fact that
$J=\underline J2^{j}$ for all $J\in \mathcal I_{n}$, we obtain for
$n$ sufficiently large:
\begin{align*}
I &\leq \sum_{J\in \mathcal I_{n}}\mathrm{P}_{h_{0}}
\biggl( \llvert n \mathcal U_{J,1} \rrvert >\frac{\eta
_{J}'(\alpha )}{4} V_{J} \biggr)
\\
&= \sum
_{J\in \mathcal I_{n}}\mathrm{P}_{h_{0}} \biggl( \biggl\llvert \sum
_{ i< i'} R_{1}(Z_{i},
Z_{i'}) \biggr\rrvert \geq \frac{\eta _{J}'(
\alpha )}{4} \frac{n-1}{2}V_{J} \biggr)
\\
&\leq \sum_{J\in \mathcal I_{n}}\mathrm{P}_{h_{0}} \biggl(
\biggl\llvert \sum_{
i< i'} R_{1}(Z_{i},
Z_{i'}) \biggr\rrvert \geq C_{R_{1}}\Lambda
\bigl(u_{J},L(J)\bigr) \biggr)
\\
& \leq 6\sum_{J\in \mathcal I_{n}}
\exp \bigl(-2\log \log \bigl(L(J)^{c_{
\alpha}}\bigr) \bigr).
\end{align*}
Using the fact that $\sum_{j\geq 1}j^{-2}=\pi ^{2}/6$, we obtain
\begin{align*}
I &\leq 6 c_{\alpha}^{-2}\sum_{J\in \mathcal I_{n}}
\bigl(\log L(J) \bigr)^{-2}
\\
&\leq \alpha \frac{6}{1+(\pi /\log
2)^{2}}\sum_{j\geq 0}(1/6+j \log
2)^{-2}
\\
&\leq \alpha \frac{6}{1+(\pi /\log 2)^{2}} \biggl(1/6+(\log
2)^{-2} \sum_{j\geq 1}j^{-2}
\biggr)= \alpha .
\end{align*}
Consider term $\mathit{II}$. Since
$\Evtex _{h_{0}}|U{\mathbbm{1}}_{\{|U|> M_{n}\}}|\leq M_{n}^{-3}
\Evtex _{h_{0}}[U^{4}{\mathbbm{1}}_{\{|U|> M_{n}\}}]\leq M_{n}^{-3}
\Evtex _{h_{0}}[U^{4}]$, Markov's inequality yields
\begin{align*}
\mathit{II}&\leq \Evtex _{h_{0}}\max_{J\in \mathcal I_{n}} \biggl
\llvert \frac{4}{\eta _{J}'(\alpha )V_{J} (n-1)}
\sum_{i< i'}U_{i}{\mathbbm{1}}_{M_{i}^{c}}
U_{i'}{\mathbbm{1}}_{M_{i'}^{c}} b^{K}(W_{i})'A'Ab^{K}(W_{i'})
\biggr\rrvert
\\
&\leq 4n\Evtex _{h_{0}} \llvert U{\mathbbm{1}}_{\{ \llvert U \rrvert > M_{n}\}} \rrvert
\Evtex _{h_{0}} \llvert U{\mathbbm{1}}_{\{ \llvert U \rrvert > M_{n}\}} \rrvert \max
_{J\in \mathcal I_{n}} \frac{\zeta _{J}^{2} \bigl\llVert
\bigl(G_{b}^{-1/2}SG^{-1/2}\bigr)^{-}_{l}
\bigr\rrVert ^{2}}{\eta _{J}'(\alpha
)V_{J}}
\\
&\leq 4nM_{n}^{-6} \bigl(\Evtex _{h_{0}}
\bigl[U^{4}\bigr] \bigr)^{2} \zeta _{
\overline J}^{2}
\max_{J\in \mathcal I_{n}} \frac{s_{J}^{-2}}{\eta
_{J}'(\alpha )V_{J}},
\end{align*}
where the fourth moment of $U=Y-h_{0}(X)$ is bounded under Assumption~\ref{A:basis}(i). Lemma~\ref{lower:bound:v_n} implies
$s_{J}^{-2}\leq \underline\sigma ^{-2}V_{J}$. By the definition of
$M_{n}=\sqrt{n}  \zeta _{\overline J}^{-1} (\log \log \overline J)^{-3/4}$
and Assumption~\ref{A:adapt:test}(i), we obtain
$\mathit{II}=o  (n^{-2} (\log \log \overline J)^{9/2}   \zeta _{
\overline J}^{8}  )=o(1)$.

 \textbf{Step~2:} We control the type II error of the test statistic
$\widetilde {\mathtt{T}}_{n}$ where $\eta _{J}'(\alpha )$ is replaced by
$\eta ''(\alpha )>0$. From the definition
$\overline J=\sup \{J: s_{J}^{-1}\zeta _{J}^{2}\sqrt{(\log J)/n}
\leq \overline c\}$, we infer that the dimension parameter
$J^\circ $ given in \eqref{thm:adapt:test:rate} satisfies
$\underline J\leq J^\circ \leq \overline J/2$ for $\overline c$ sufficiently
large by Assumptions \ref{A:s_J} and \ref{A:adapt:test}(iii). Thus, by
the construction of the set $\mathcal I_{n}$, there exists
$J^{*}\in \mathcal I_{n}$ such that $J^\circ \leq J^{*}< 2J^\circ $. Let
$K^{*}=K(J^{*})$. We note that for all
$h\in \mathcal H_{1}(\delta ^{\circ}\text{\textsf r}_{n})$:
\begin{align*}
\mathrm{P}_{h} (\widetilde {\mathtt{T}}_{n}=0 )&=
\mathrm{P}_{h} \bigl(n \widehat{D}_{J}(h_{0})\leq
\eta ''(\alpha ) V_{J}\text{ for all } J\in
\mathcal I_{n} \bigr)
\\
&\leq \mathrm{P}_{h} \bigl(n \widehat{D}_{J^{*}}(h_{0})
\leq c_{1} \sqrt{\log \log n-\log \alpha} V_{J^{*}} \bigr)
\end{align*}
with $c_{1}=4(1+c_{0}) $, by the definition of $\eta ''(\alpha )$. Note
that
$\log \log n-\log \alpha =(\log \log n)[1-(\log \alpha )/(\log \log n)]
\leq 2\log \log n$ for all $n$ sufficiently large. Consequently, we may
apply Lemma~\ref{lemma:type2}(ii) which implies
$\mathrm{P}_{h}  (\widetilde {\mathtt{T}}_{n}=0  )=o(1)$ uniformly
for $h\in \mathcal H_{1}(\delta ^{\circ}\text{\textsf r}_{n})$.

 \textbf{Step~3:} Finally, we account for estimation of the normalization
factor $V_{J}$ and for estimation of upper bound of the RES index
$\widehat{\mathcal I}_{n}$. We control the type I error of the test
$\widehat{\mathtt{T}}_{n}$ under simple null hypotheses as follows. The
lower bound in Lemma~\ref{lemma:eta:bounds} implies
\begin{align*}
\mathrm{P}_{h_{0}} (\widehat{\mathtt{T}}_{n}=1 ) \leq{}&
\mathrm{P}_{h_{0}} \Bigl(\max_{J\in \widehat{\mathcal I}_{n}} \bigl\{n
\widehat{D}_{J}(h_{0})/\bigl(\eta _{J}'(
\alpha )\widehat{V}_{J}\bigr) \bigr\}>(1-c_{0})^{-1}
\Bigr)
\\
\leq{}& \mathrm{P}_{h_{0}} \Bigl(\max_{J\in \mathcal I_{n}} \bigl\{n
\widehat{D}_{J}(h_{0})/\bigl(\eta _{J}'(
\alpha )\widehat{V}_{J}\bigr) \bigr\}>(1-c_{0})^{-1},
\\
&{}
\widehat{V}_{J}\geq (1-c_{0}) V_{J}\text{ for all
}J\in \mathcal I_{n} \Bigr)&
\\
&{} +\mathrm{P}_{h_{0}} \bigl(\widehat{V}_{J}<
(1-c_{0}) V_{J} \text{ for all }J\in \mathcal
I_{n} \bigr)+ \mathrm{P}_{h_{0}} ( \widehat J_{\max}>
\overline J )
\\
\leq{}& \mathrm{P}_{h_{0}} \Bigl(\max_{J\in \mathcal I_{n}} \bigl\{n
\widehat{D}_{J}(h_{0})/\bigl(\eta _{J}'(
\alpha )V_{J}\bigr) \bigr\}> 1 \Bigr)
\\
&{}+ \mathrm{P}_{h_{0}} \Bigl(
\max_{J\in \mathcal I_{n}} \llvert \widehat{V}_{J}/V_{J}-1
\rrvert > c_{0} \Bigr)+o(1)\leq \alpha +o(1),
\end{align*}
where the third inequality is due to Lemmas \ref{lemma:J_hat}(i) and
\ref{lem:est:var}(ii), and the last inequality is due to Step~1 of this
proof. To bound the type II error of the test
$\widehat{\mathtt{T}}_{n}$, recall the definition of
$J^{*}\in \mathcal I_{n}$ given in Step~2 of this proof. Using the upper
bound of Lemma~\ref{lemma:eta:bounds} together with Lemmas
\ref{lemma:J_hat}(ii) and \ref{lem:est:var}, we evaluate uniformly for
$h\in \mathcal H_{1}(\delta ^{\circ}\text{\textsf r}_{n})$:
\begin{align*}
\mathrm{P}_{h} (\widehat{\mathtt{T}}_{n}=0 ) \leq{}&
\mathrm{P}_{h} \bigl(n\widehat{D}_{J^{*}}(h_{0})\leq
(1+c_{0})^{-1}\eta ''(\alpha )
\widehat{V}_{J^{*}} \bigr)+ \mathrm{P}_{h}
\bigl(J^{*}>\widehat J_{
\max} \bigr)&
\\
\leq{}& \mathrm{P}_{h} \bigl(n\widehat{D}_{J^{*}}(h_{0})
\leq (1+c_{0})^{-1} \eta ''(\alpha
) \widehat{V}_{J^{*}}, \widehat{V}_{J^{*}}\leq
(1+c_{0})V_{J^{*}} \bigr)
\\
&{}+\mathrm{P}_{h} \bigl(
\widehat{V}_{J^{*}}> (1+c_{0})V_{J^{*}} \bigr)+o(1)
\\
\leq{}& \mathrm{P}_{h} \bigl(n\widehat{D}_{J^{*}}(h_{0})
\leq \eta ''( \alpha )V_{J^{*}} \bigr)+o(1)=
o(1),
\end{align*}
where the last equation is due to Step~2 of this proof.

Since both the mildly ill-posed and severely ill-posed are special cases
of regularly varying, the rest of the results follow. In the mildly ill-posed
case, we obtain
$J^\circ \sim  (n/\sqrt{\log \log n} )^{2d_{x}/(4(p+a)+d_{x})}$ which
implies
$\text{\textsf r}_{n}\sim  (\sqrt{\log \log n}/n )^{2p/(4(p+a)+d_{x})}$.
In the severely ill-posed case, note that if
$J^\circ \sim  (c\log n )^{d_{x}/a}$ for some constant
$c\in (0,1)$, then we obtain
$n^{-1/2}(J^\circ \log \log n)^{1/4} s_{J^\circ }^{-1}\lesssim (J^
\circ )^{-p/d_{x}}\sim  (\log n )^{-p/a}$.
\end{proof}

\begin{proof}[{Proof of Theorem~\ref{thm:adapt:est:test}.}]
We prove this result in three steps. First, we bound the type I error of
the test statistic
$\widetilde{\mathtt{T}}_{n} = \mathbbm{1} \{ \max_{J
\in \mathcal I_{n}} \{n\widehat{D}_{J}/(\eta _{J}'(\alpha )V_{J})
 \}>1 \}$, where $\eta _{J}'(\alpha )$ is given in the proof of Theorem~\ref{thm:adapt:test}. Second, we bound the type II error of
$\widetilde {\mathtt{T}}_{n}$, where $\eta _{J}'(\alpha )$ is replaced
by $\eta ''(\alpha )$ given in the proof of Theorem~\ref{thm:adapt:test}. Third, we show that Steps 1 and 2 are sufficient
to control the type I and type II errors of our adaptive test
$\widehat{\mathtt{T}}_{n}$ for the composite null.

 \textbf{Step~1:} We control the type I error of the test statistic
$\widetilde {\mathtt{T}}_{n}$ using the decomposition
\begin{align*}
n(n-1) \widehat{D}_{J}&=\sum_{i\neq i'}
\bigl(Y_{i}-\widehat h^{
\text{\textsc r}}_{J}(X_{i})
\bigr) \bigl(Y_{i'}-\widehat h^{\text{\textsc r}}_{J}(X_{i'})
\bigr)b^{K}(W_{i})'\widehat A'
\widehat Ab^{K}(W_{i'})
\\
&= \biggl\llVert \sum_{i} \bigl(Y_{i}-
\widehat h^{\text{\textsc r}}_{J}(X_{i}) \bigr) \widehat
Ab^{K}(W_{i}) \biggr\rrVert ^{2}-\sum
_{i} \bigl\llVert \bigl(Y_{i}-\widehat
h^{
\text{\textsc r}}_{J}(X_{i}) \bigr)\widehat
Ab^{K}(W_{i}) \bigr\rrVert ^{2}.
\end{align*}
For any $h\in \mathcal H_{0}$, we define
$h_{J}^{*}:=\argmin _{\phi \in \mathcal H_{0,J}} \|\sum_{i}(\phi -h)(X_{i})
\widehat Ab^{K}(W_{i})\|$. The definition of the restricted NPIV estimator
$\widehat h^{\text{\textsc r}}_{J} \in \mathcal H_{0,J}$ in
\eqref{def:est:h^R} yields for all $h\in \mathcal H_{0}$:
\begin{align*}
&\biggl\llVert \sum_{i} \bigl(Y_{i}-
\widehat h^{\text{\textsc r}}_{J}(X_{i}) \bigr) \widehat
Ab^{K}(W_{i}) \biggr\rrVert
\\
 &\quad \leq \biggl\llVert \sum
_{i} \bigl(Y_{i}- h_{J}^{*}(X_{i})
\bigr)\widehat Ab^{K}(W_{i}) \biggr\rrVert
\\
&\quad \leq \biggl\llVert \sum_{i} \bigl(Y_{i}-
h(X_{i}) \bigr)\widehat Ab^{K}(W_{i}) \biggr
\rrVert + \biggl\llVert \sum_{i}\bigl(h-
h_{J}^{*}\bigr) (X_{i})\widehat
Ab^{K}(W_{i}) \biggr\rrVert .
\end{align*}
By Lemma~\ref{Lemma:adapt:est:step1} (see below), uniformly for
$J\in \mathcal I_{n}$, we have
\begin{align*}
&\frac{n\widehat{D}_{J}}{ \eta _{J}'(\alpha
)V_{J}}- \frac{n\widehat{D}_{J}(h)}{ \eta
_{J}'(\alpha )V_{J}}
\\
&\quad \lesssim
\bigl(V_{J} \sqrt{(\log \log J)/J}\bigr)^{-1/2}n^{-1}
\sum_{i} \bigl(Y_{i}- h(X_{i})
\bigr)b^{K}(W_{i})'\widehat A'
\widehat A b^{K}(W_{i}) \bigl(\widehat h_{J}^{
\text{\textsc r}}-
h \bigr) (X_{i})&
\\
&\qquad {} + \bigl(V_{J}\sqrt{(\log \log J)/J}\bigr)^{-1/2} \biggl
\llVert \frac{1}{\sqrt n} \sum_{i}
\bigl(Y_{i}- h(X_{i}) \bigr)\widehat Ab^{K}(W_{i})
\biggr\rrVert
\\
&\quad =: \bigl(V_{J}\sqrt{(\log \log J)/J}\bigr)^{-1/2}
(T_{1,J}+2T_{2,J} )
\end{align*}%
wpa1 uniformly for $h\in \mathcal H_{0}$, where $\widehat{D}_{J}(h)$ is
given in \eqref{whDh0} (with $h_{0}$ replaced by
$h=\Pi _{\mathcal H_0}h$ under $\mathcal H_{0}$). Now we may follow Step~1 of the proof of Theorem~\ref{thm:adapt:test} and obtain
\begin{align*}
\limsup_{n\to \infty} \sup_{h\in \mathcal H_{0}}
\mathrm{P}_{h} \Bigl(\max_{J\in \mathcal I_{n}} \bigl\{n
\widehat{D}_{J}(h)/\bigl(\eta _{J}'( \alpha
)V_{J}\bigr) \bigr\}>1/4 \Bigr)\leq \alpha .
\end{align*}
It remains to control $T_{1,J}$ and $T_{2,J}$. Consider $T_{1,J}$. For
all $J\in \mathcal I_{n}$, we evaluate
\begin{align*}
T_{1,J}={}& \frac{1}{n}\sum_{i}
\bigl(Y_{i}- h(X_{i}) \bigr)b^{K}(W_{i})'A'A
b^{K}(W_{i}) \bigl(\widehat h^{\text{\textsc r}}_{J}-h
\bigr) (X_{i})
\\
&{} +\frac{1}{n}\sum_{i} \bigl(Y_{i}-
h(X_{i}) \bigr)b^{K}(W_{i})' \bigl(
\widehat A'\widehat A - A'A \bigr) b^{K}(W_{i})
\bigl(\widehat h^{
\text{\textsc r}}_{J}-h \bigr) (X_{i}):=
T_{11,J}+T_{12,J}.
\end{align*}
Consider $T_{11,J}$. We first observe by the Cauchy--Schwarz inequality
that
\begin{align*}
T_{11,J}\leq \biggl(\frac{1}{n}\sum_{i}
\bigl(Y_{i}- h(X_{i}) \bigr)^{2} \bigl\llVert A
b^{K}(W_{i}) \bigr\rrVert ^{2}
\biggr)^{1/2} \biggl(\frac{1}{n}\sum_{i}
\bigl\llVert A b^{K}(W_{i}) \bigl(\widehat
h^{\text{\textsc r}}_{J}-h \bigr) (X_{i}) \bigr\rrVert
^{2} \biggr)^{1/2}.
\end{align*}
Further, another application of the Cauchy--Schwarz inequality implies
\begin{eqnarray*}
\Evtex _{h}\max_{J\in \mathcal I_{n}} \bigl\llVert \bigl(Y-h(X)
\bigr)A b^{K}(W) \bigr\rrVert ^{2} &\leq& \max
_{J\in \mathcal I_{n}} \sqrt J \bigl\llVert A \Evtex _{h} \bigl[
\bigl(Y- h(X) \bigr)^{2} b^{K}(W)b^{K}(W)' \bigr]A' \bigr\rrVert
_{F}
\\
&=& \max_{J\in \mathcal I_{n}} \{\sqrt JV_{J} \},
\end{eqnarray*}
using the definition of the normalization term $V_{J}$. Consequently, we
evaluate
\begin{equation*}
\max_{J\in \mathcal I_{n}}\frac{T_{11,J}}{V_{J}\sqrt{\log
\log J}} \lesssim \max_{J\in \mathcal I_{n}} \frac{\zeta _{J} \bigl
\llVert \widehat h_{J}^{\text{\textsc r}}-h \bigr\rrVert _{L^{2}(X)}}{
\sqrt{\log \log J}} \times \max_{J\in \mathcal I_{n}} \frac{\sqrt{\Evtex
_{h} \bigl[ \bigl\llVert \bigl(Y- h(X)\bigr)A b^{K}(W)
\bigr\rrVert ^{2} \bigr]}}{\zeta _{J}s_{J}V_{J}}
\end{equation*}
wpa1 uniformly for $h\in \mathcal H_{0}$, where the right-hand side tends
to zero by the rate condition imposed in Assumption~\ref{h0:est}(i), that
is,
$\mathrm{P}_{h}(\max_{J\in \mathcal I_{n}}\|\widehat h^{\text{\textsc r}}_{J}-h
\|_{L^{2}(X)}\zeta _{J}/(\log \log J)^{1/4}>\varepsilon )\to 0$ uniformly
for $h\in \mathcal H_{0}$ for any $\varepsilon >0$. Similarly,
$\max_{J\in \mathcal I_{n}}T_{12,J}/(V_{J}\sqrt{\log \log J})$ vanishes
wpa1 uniformly for $h\in \mathcal H_{0}$, using that
\begin{align*}
&\mathrm{P} \Bigl(\max_{J\in \mathcal I_{n}} \bigl\{s_{J}^{2}
\zeta _{J}^{-1} \sqrt{n/(\log J)} \bigl\llVert (\widehat A-
A)G_{b}^{1/2} \bigr\rrVert \bigr\}>C \Bigr)&
\\
&\quad =\mathrm{P} \biggl(\max_{J\in \mathcal I_{n}} \biggl\{s_{J}^{2}
\zeta _{J}^{-1} \sqrt{\frac{n}{\log J}} \bigl\llVert \bigl(
\widehat G_{b}^{-1/2} \widehat S \widehat G^{-1/2}
\bigr)^{-}_{l}\widehat G_{b}^{-1/2}
G_{b}^{1/2}- \bigl(G_{b}^{-1/2}
SG^{-1/2}\bigr)^{-}_{l} \bigr\rrVert \biggr\}>C
\biggr)
\\
&\quad =o(1),
\end{align*}
by Lemma~E.5(i). Consider $T_{2,J}$. We have
\begin{eqnarray*}
T_{2,J}&\leq& \biggl\llVert \frac{1}{\sqrt n}\sum
_{i} \bigl(Y_{i}- h(X_{i})
\bigr)Ab^{K}(W_{i}) \biggr\rrVert + \biggl\llVert \frac{1}{
\sqrt n}\sum_{i} \bigl(Y_{i}-
h(X_{i}) \bigr) ( \widehat A-A)b^{K}(W_{i})
\biggr\rrVert
\\
 &:=& T_{21,J}+T_{22,J}.
\end{eqnarray*}
We have
$\Evtex _{h} \max_{J\in \mathcal I_{n}} T_{21,J}\leq \sqrt{\Evtex _{h}\max_{J
\in \mathcal I_{n}} \|(Y-h(X))Ab^{K(J)}(W)\|^{2}}\leq \max_{J\in
\mathcal I_{n}}\{J^{1/4}\sqrt{V_{J}}\}$ as derived above and conclude
\begin{align*}
\Evtex _{h}\max_{J\in \mathcal I_{n}} \frac{T_{21,J}}{
\bigl(V_{J}\sqrt{J(\log \log J}) \bigr)^{1/2}} \lesssim \max
_{J\in \mathcal I_{n}} \frac{{J^{1/4}\sqrt{V_{J}}}}{
\bigl(V_{J}\sqrt{J(\log \log J)} \bigr)^{1/2}} =o(1)
\end{align*}
uniformly for $h\in \mathcal H_{0}$. Concerning the second summand
$T_{22,J}$, by another application of Lemma~E.5,
$\max_{J\in \mathcal I_{n}}T_{22,J}/(V_{J}\sqrt{J(\log \log J)})$ vanishes
wpa1 uniformly for $h\in \mathcal H_{0}$.

 \textbf{Step~2:} We control the type II error of the test statistic
$\widetilde {\mathtt{T}}_{n}$. Let $J^{*}$ be as in the proof of Theorem~\ref{thm:adapt:test}. We evaluate for all
$h\in \mathcal H_{1}(\delta ^{\circ}\text{\textsf r}_{n})$ that
\begin{equation*}
\mathrm{P}_{h} (\widetilde {\mathtt{T}}_{n}=0 )=
\mathrm{P}_{h} \bigl(n \widehat{D}_{J}\leq \eta
''(\alpha ) V_{J} \text{ for all } J\in
\mathcal I_{n} \bigr)\leq \mathrm{P}_{h} (n
\widehat{D}_{J^{*}}\leq c_{1}\sqrt{\log \log n-\log \alpha}
V_{J^{*}} ),
\end{equation*}
with $c_{1}=4(1+c_{0}) $, by the definition of $\eta ''(\alpha )$. Let
$\widehat U_{i}^{J}:=(Y_{i}- \widehat h_{J}^{\text{\textsc r}}(X_{i}))Ab^{K}(W_{i})$;
then
\begin{eqnarray*}
\bigl\llVert \Evtex _{h}\bigl[\widehat U^{J^{*}}\bigr] \bigr
\rrVert ^{2}&=& \Evtex _{h}\bigl[\bigl(Y-\widehat
h_{J^{*}}^{
\text{\textsc r}}(X)\bigr)b^{K^{*}}(W)'\bigr]
A' A \Evtex _{h}\bigl[\bigl(Y-\widehat
h_{J^{*}}^{
\text{\textsc r}}(X)\bigr)b^{K^{*}}(W)\bigr]
\\
&=& \bigl\llVert
Q_{J^{*}}\bigl(h-\widehat h_{J^{*}}^{
\text{\textsc r}}\bigr) \bigr
\rrVert _{L^{2}(X)}^{2}.
\end{eqnarray*}
The triangular inequality implies
$ |\|Q_{J^{*}}(h-\widehat h_{J^{*}}^{\text{\textsc r}})\|_{L^{2}(X)}-\|h-
\widehat h_{J^{*}}^{\text{\textsc r}}\|_{L^{2}(X)} | \leq \sup_{\phi
\in \mathcal H}\|Q_{J^{*}}\phi - \phi \|_{L^{2}(X)}$ uniformly for
$h\in \mathcal H_{1}(\delta ^{\circ}\text{\textsf r}_{n})$. Consequently, Lemma~\ref{lemma:bias}(ii) together with the definition of $J^{*}$ implies
$\sup_{h\in \mathcal H_{1}(\delta ^{\circ}\text{\textsf r}_{n})}(\|\Evtex _{h}[
\widehat U^{J^{*}}]\|-\|h-\widehat h_{J^{*}}^{\text{\textsc r}}\|_{L^{2}(X)})^{2}
\leq C_{B} \text{\textsf r}_{n}^{2}$ for some constant $C_{B}>0$. Using this bound,
we derive
\begin{flalign*}
\mathrm{P}_{h} (n \widehat{D}_{J^{*}}\leq 2c_{1}
\sqrt{\log \log n} V_{J^{*}} )
 &= \mathrm{P}_{h} \biggl( \bigl
\llVert \Evtex _{h}\bigl[ \widehat U^{J^{*}}\bigr] \bigr\rrVert
^{2}-\widehat{D}_{J^{*}}> \bigl\llVert \Evtex _{h}
\bigl[\widehat U^{J^{*}}\bigr] \bigr\rrVert ^{2} -
\frac{2c_{1}\sqrt{\log \log n} V_{J^{*}}}{n} \biggr)&
\\
&\leq T_{1}+T_{2},
\\
T_{1}:=\mathrm{P}_{h} \Biggl( \Biggl\llvert
\frac{4}{n(n-1)}\sum_{j=1}^{J^{*}} \sum
_{i< i'} \bigl( &\widehat U_{ij}\widehat
U_{i'j}- \Evtex _{h}[ \widehat U_{1j}]^{2}
\bigr) \Biggr\rrvert >\rho _{h} \Biggr),
\\
T_{2}:=\mathrm{P}_{h} \biggl( \biggl\llvert
\frac{4}{n(n-1)}\sum_{i< i'} \bigl(Y_{i}-&
\widehat h_{J^{*}}^{\text{\textsc r}}(X_{i})\bigr)
\bigl(Y_{i'}-\widehat h_{J^{*}}^{
\text{\textsc r}}(X_{i'})
\bigr)b^{K^{*}}(W_{i})'
\bigl(A'A-
\widehat A'\widehat A \bigr)b^{K^{*}}(W_{i'})
\biggr\rrvert >\rho _{h} \biggr),
\end{flalign*}
where
$\rho _{h}=\|h-\mathcal H_{0}\|_{L^{2}(X)}^{2}/2-2c_{1}n^{-1}\sqrt{
\log \log n}V_{J^{*}} -C_{B}\text{\textsf r}_{n}^{2}$. To establish an upper bound
of $T_{1}$, we make use of Lemma~E.3 which yields
\begin{align}
\label{Markovs:bound:est_h} T_{1}\lesssim n^{-1}s_{J^{*}}^{-2}
\rho _{h}^{-2}\mathcal C_{h}^{2} \bigl(
\llVert h-\mathcal H_{0} \rrVert _{L^{2}(X)}^{2}+
\bigl(J^{*}\bigr)^{-2p/d_{x}} \bigr) + n^{-2}
s_{J^{*}}^{-4} J^{*}\rho _{h}^{-2}.
\end{align}
First, consider the case where
$n^{-2}s_{J^{*}}^{-4} J^{*}\rho _{h}^{-2}$ dominates the right-hand side.
For any $h\in \mathcal H_{1}(\delta ^{\circ}\text{\textsf r}_{n})$, we have
$\|h-\mathcal H_{0}\|_{L^{2}(X)}\geq \delta ^{\circ}\text{\textsf r}_{n}$ for
some sufficiently large $\delta ^{\circ}>0$ and hence, we obtain the lower
bound
$\rho _{h}\geq ((\delta ^{\circ})^{2}/2-C-C_{B}) \text{\textsf r}_{n}^{2}$ for
some constant $C>0$. Consequently, we have
$T_{1}\lesssim n^{-2}s_{J^{*}}^{-4} J^{*}(J^{*})^{4p/d_{x}}=o(1)$. Second,
consider the case where
$n^{-1}s_{J^{*}}^{-2} \rho _{h}^{-2} \mathcal C_{h}^{2} (\|h-
\mathcal H_{0}\|_{L^{2}(X)}^{2}+(J^{*})^{-2p/d_{x}} )$ dominates. For
any $h\in \mathcal H_{1}(\delta ^{\circ}\text{\textsf r}_{n})$, we have
$\|h-\mathcal H_{0}\|_{L^{2}(X)}^{2}\geq (\delta ^{\circ})^{2}
\text{\textsf r}_{n}^{2}\geq 5c_{1} n^{-1}V_{J^{*}}\sqrt{\log \log n}$ and we
obtain the lower bound
$\rho _{h}\geq (1/5 - C_{B}/(\delta ^{\circ})^{2})  \|h-\mathcal H_{0}
\|_{L^{2}(X)}^{2}$. Hence, \eqref{Markovs:bound:est_h} yields uniformly
for $h\in \mathcal H_{1}(\delta ^{\circ}\text{\textsf r}_{n})$ that
\begin{equation*}
T_{1}\lesssim n^{-1}s_{J^{*}}^{-2}
\mathcal C_{h}^{2} \bigl( \llVert h- \mathcal
H_{0} \rrVert _{L^{2}(X)}^{-2}+ \llVert h-\mathcal
H_{0} \rrVert ^{-4}_{L^{2}(X)} \bigl(J^{*}
\bigr)^{-2p/d} \bigr)\lesssim n^{-1}s_{J^{*}}^{-2}
\sqrt{J^{*}} \text{\textsf r}_{n}^{-2} =o(1)
\end{equation*}
using that
$\sup_{h\in \mathcal H_{1}(\delta ^{\circ}\text{\textsf r}_{n})}\mathcal C_{h}^{2}
\lesssim \sqrt{J^{*}}$ by Assumption~\ref{h0:est}(ii). Finally,
$T_{2}=o(1)$ uniformly for
$h\in \mathcal H_{1}(\delta ^{\circ}\text{\textsf r}_{n})$ by making use of
 Lemma~E.4.

 \textbf{Step~3:} Finally, we account for estimation of the normalization
factor $V_{J}$ and for estimation of the upper bound of the RES index set
$\widehat{\mathcal I}_{n}$. Lemma~\ref{lemma:J_hat}(i) implies
$\sup_{h\in \mathcal H_{0}}\mathrm{P}_{h} (\widehat J_{\max}>
\overline J )=o(1)$. We thus control the type I error of the test
$\widehat {\mathtt{T}}_{n}$ for testing composite hypotheses, as follows.
By the lower bound of Lemma~\ref{lemma:eta:bounds}, we have
\begin{equation*}
\mathrm{P}_{h} (\widehat {\mathtt{T}}_{n}=1 ) \leq
\mathrm{P}_{h} \biggl(\max_{J\in \mathcal I_{n}} \frac{n
\widehat{D}_{J}}{\eta _{J}'(\alpha
)V_{J}}> 1 \biggr)+ \mathrm{P}_{h} \Bigl(\max
_{J\in \mathcal I_{n}} \llvert \widehat{V}_{J}/V_{J}-1
\rrvert > c_{0} \Bigr)+o(1)\leq \alpha +o(1)
\end{equation*}
uniformly for $h\in \mathcal H_{0}$, where the last inequality is due to
Step~1 of this proof and Lemma~\ref{lem:est:var}(ii). To bound the type
II error of the test $\widehat {\mathtt{T}}_{n}$, recall the definition
of $J^{*}\in \mathcal I_{n}$ introduced in Step~2 and note that
$\sup_{h\in \mathcal H}\mathrm{P}_{h} (J^{*}>\widehat J_{\max}
 )=o(1)$ by Lemma~\ref{lemma:J_hat}(ii). Consequently, the upper bound
of Lemma~\ref{lemma:eta:bounds} and another application of Lemma~\ref{lemma:J_hat}(ii) give uniformly for
$h\in \mathcal H_{1}(\delta ^{\circ}\text{\textsf r}_{n})$:
$\mathrm{P}_{h}  (\widehat {\mathtt{T}}_{n}=0  ) \leq
\mathrm{P}_{h} (n\widehat{D}_{J^{*}}\leq \eta ''(\alpha )V_{J^{*}}
 ) +\mathrm{P}_{h}  ( \llvert \widehat{V}_{J^{*}}/V_{J^{*}}-1
 \rrvert > c_{0}   )+o(1)= o(1)$, where the last equation is due to
Step~2 and Lemma~\ref{lem:est:var}(i).
\end{proof}

\begin{lemma}%
\label{Lemma:adapt:est:step1}
Let Assumptions \ref{A:LB}(i)--(iii), \ref{A:basis}(i),
\ref{A:adapt:test}, and \ref{h0:est}(i) be satisfied. Recall the notation
$h_{J}^{*}=\argmin _{\phi \in \mathcal H_{0,J}} \|\sum_{i}(\phi -h)(X_{i})
\widehat Ab^{K}(W_{i})\|$. Then, for all $\varepsilon >0$, we have
\begin{equation*}
\sup_{h\in \mathcal H_{0}}\mathrm{P}_{h} \biggl(\max
_{J\in \mathcal I_{n}} \biggl\llVert \biggl(nV_{J}\sqrt{(\log \log
J)/J} \biggr)^{-1/2}\sum_{i}\bigl(h-
h_{J}^{*}\bigr) (X_{i}) \widehat
Ab^{K}(W_{i}) \biggr\rrVert >\varepsilon \biggr)=o(1).
\end{equation*}
\end{lemma}
\begin{proof}[{Proof of Lemma~\ref{Lemma:adapt:est:step1}.}]
The result is immediate under parametric null hypotheses. We now consider
the nonparametric case, where the semiparametric situation follows analogously.
Define
$\widetilde \Pi _{\mathcal B }h:=\argmin _{\phi \in \mathcal B}
 \|\sum_{i}(\phi -h)(X_{i})\widehat Ab^{K}(W_{i}) \|$ for any closed,
convex set $\mathcal B\subset \mathcal H$ and
$\Psi _{J,h}:=\{\phi : \phi =\kappa _{1} Q_{1} h+\cdots +\kappa _{J}
 Q_{J} h \text{ where } \sum_{j=1}^{J}|\kappa _{j}|\leq 1\}\subset
\Psi _{J}$ for any $h\in \mathcal H$. We have $0\in \Psi _{J,h}$; in particular,
the zero function belongs to the interior of $\Psi _{J,h}$. Thus,
$0\in \mathcal H_{0}$ implies that the zero function belongs to the interior
of $\Psi _{J,h}-\mathcal H_{0}$. Now, using that $\mathcal H_{0}$ and
$\Psi _{J,h}$ are closed and convex subsets of $\mathcal H$, we may apply
\citet [Corollary~4.5(i)]{bauschke1993}: there exist
$h_{J}\in \Psi _{J,h}\cap \mathcal H_{0}\neq \emptyset $ and $0<c<1$ such
that
\begin{equation}
\label{ineq:alternating_proj} \sup_{h\in \mathcal H_{0}}\mathrm{P}_{h} \biggl(
\max_{J\in \mathcal I_{n}} \biggl\{ \biggl\llVert n^{-1}\sum
_{i} \bigl(h_{J}- (\widetilde \Pi _{\Psi _{J,h}}
\widetilde \Pi _{\mathcal H_{0}})^{m} h \bigr) (X_{i})
\widehat Ab^{K}(W_{i}) \biggr\rrVert \lesssim
c^{m} \biggr\} \biggr)=1-o(1) 
\end{equation}
for all $m\geq 1$. Here, we used also that
$\Psi _{J,h}\subset \Psi _{J',h}$ whenever $J<J'$. The definition of
$h_{J}^{*}$ implies
\begin{align*}
\biggl\llVert\sum_{i}\bigl(h-
h_{J}^{*}\bigr) (X_{i})\widehat
Ab^{K}(W_{i}) \biggr\rrVert \leq{}& \biggl\llVert \sum
_{i}(h- h_{J}) (X_{i})\widehat
Ab^{K}(W_{i}) \biggr\rrVert
\\
\leq{}& \biggl\llVert \sum_{i} \bigl(h-(\widetilde
\Pi _{\Psi _{J,h}} \widetilde \Pi _{\mathcal H_{0}})^{m}h \bigr)
(X_{i})\widehat Ab^{K}(W_{i}) \biggr\rrVert
\\
&{}+
\biggl\llVert \sum_{i} \bigl((\widetilde \Pi
_{\Psi _{J,h}} \widetilde \Pi _{\mathcal H_{0}})^{m}h-
h_{J} \bigr) (X_{i})\widehat Ab^{K}(W_{i})
\biggr\rrVert .
\end{align*}%
We make use of the decomposition
$h-(\widetilde \Pi _{\Psi _{J,h}} \widetilde \Pi _{\mathcal H_{0}})^{m}h=
 (\text{id}+\widetilde \Pi _{\Psi _{J,h}} \widetilde \Pi _{
\mathcal H_{0}}+\cdots +(\widetilde \Pi _{\Psi _{J,h}}
\widetilde \Pi _{\mathcal H_{0}})^{m-1} ) (h-\widetilde \Pi _{\Psi _{J,h}}
\widetilde \Pi _{\mathcal H_{0}} h)$. We may assume that
$h\in \mathcal H_{0}$ does not belong to $\Psi _{J,h}$ and thus,
$\widetilde \Pi _{\Psi _{J,h}} \widetilde \Pi _{\mathcal H_{0}}$ forms
a contraction satisfying
\begin{align*}
\biggl\llVert \sum_{i} \bigl(h-(\widetilde \Pi
_{\Psi _{J,h}} \widetilde \Pi _{
\mathcal H_{0}})^{m}h \bigr)
(X_{i})\widehat Ab^{K}(W_{i}) \biggr\rrVert \leq
\biggl\llVert \sum_{i} (h-\widetilde \Pi
_{\Psi _{J,h}}h ) (X_{i}) \widehat Ab^{K}(W_{i})
\biggr\rrVert .
\end{align*}
Choosing $m=\lfloor \log _{c}(J^{-1/2}\sqrt{V_{J}/n})\rfloor $, we have
$m\geq 1$ for $n$ sufficiently large by the upper bound on $V_{J}$ established
in Lemma~\ref{upper:bound:v_n}, Assumption~\ref{A:adapt:test}(ii), and
using that $0<c<1$. Plugging this choice of $m$ in equation
\eqref{ineq:alternating_proj} thus implies
\begin{align*}
&\biggl\llVert \biggl(nV_{J}\sqrt{(\log \log J)/J}
\biggr)^{-1/2}\sum_{i}\bigl(h-
h_{J}^{*}\bigr) (X_{i}) \widehat
Ab^{K}(W_{i}) \biggr\rrVert
\\
&\quad \lesssim \biggl\llVert \biggl(nV_{J}\sqrt{(\log \log J)/J}
\biggr)^{-1/2}\sum_{i}(h-Q_{J} h)
(X_{i})\widehat Ab^{K}(W_{i}) \biggr\rrVert
+J^{-1/2}
\end{align*}
with probability approaching 1, uniformly for $h\in \mathcal H_{0}$, using
that $Q_{J}h\in \Psi _{J,h}$. It is sufficient to consider the first summand
on the right-hand side since
$\max_{J\in \mathcal I_{n}}J^{-1/2}=\underline J^{-1/2}=o(1)$. First,
we consider the off-diagonal summands:
\begin{align*}
&\frac{\sqrt J}{n}\sum_{i\neq i'}(h- Q_{J} h)
(X_{i}) (h- Q_{J} h) (X_{i'})b^{K}(W_{i})'A'A
b^{K}(W_{i'})
\\
& \qquad {}+\frac{\sqrt J}{n}\sum_{i\neq i'}(h- Q_{J}
h) (X_{i}) (h- Q_{J} h) (X_{i'})b^{K}(W_{i})'
\bigl(\widehat A'\widehat A - A'A \bigr)
b^{K}(W_{i'})
\\
&\quad  =: T_{31,J}+T_{32,J}.
\end{align*}%
Consider $T_{31,J}$. By the definition of
$Q_{J} h(\cdot )=\widetilde \psi ^{J}(\cdot )'A\Evtex [b^{K}(W)h(X)]$, we observe
\begin{equation*}
\Evtex \bigl[(h- Q_{J} h) (X)Ab^{K}(W) \bigr] =\Evtex
\bigl[Q_{J}(h-Q_{J} h) (X) \widetilde\psi ^{J}(X)
\bigr]=0.
\end{equation*}
Further, we infer for all $J\in \mathcal I_{n}$ that
$\sqrt{\Evtex [(Q_{J} h-h)^{2}(X) |W]}\lesssim \|Q_{J} h-h\|_{L^{2}(X)}
\lesssim J^{-p/d_{x}} $ wpa1 uniformly for $h\in \mathcal H_{0}$ by Lemma~\ref{lemma:bias}(ii) and thus,
$\Evtex |\sqrt J\Evtex [(Q_{J} h-h)^{2}(X) |W]|=o(1)$ by Assumption~\ref{A:adapt:test}(iii). Further, we obtain for all
$J\in \mathcal I_{n}$ and uniformly for $h\in \mathcal H_{0}$:
\begin{equation*}
\Evtex \bigl[\bigl(Q_{J} (h-\Pi _{J} h)\bigr)^{4}(X)
\bigr]\lesssim \zeta _{J}^{2} \bigl\llVert
\bigl(G_{b}^{-1/2} S G^{-1/2}\bigr)_{\ell}^{-}
\Evtex \bigl[(h- \Pi _{J} h) (X)\widetilde b^{K}(W)\bigr]
\bigr\rrVert ^{4}\lesssim \zeta _{J}^{2}
J^{-4p/d_{x}}
\end{equation*}
and $J\Evtex [(Q_{J} (h-\Pi _{J} h))^{4}(X)]=o(1)$ by Assumption~\ref{A:adapt:test}(iii). Using these moment bounds, we may follow Step~1 of the proof of Theorem~\ref{thm:adapt:test} by replacing
$Y_{i}-h(X_{i})$ with $J^{1/4}(Q_{J} h-h)(X_{i})$ for
$h\in \mathcal H_{0}$ and for any $\varepsilon >0$ obtain
$\mathrm{P}_{h}(\max_{J\in \mathcal I_{n}} T_{31,J}/(V_{J}\sqrt{
\log \log J})>\varepsilon )=o(1)$ uniformly for
$h\in \mathcal H_{0}$. Consider $T_{32,J}$. For any $\varepsilon >0$, we
have
$ \mathrm{P}_{h}(\max_{J\in \mathcal I_{n}}T_{32,J}/(V_{J}\sqrt{
\log \log J})>\varepsilon )=o(1)$ uniformly for
$h\in \mathcal H_{0}$, following Lemma~E.6 again
by replacing $Y_{i}- h(X_{i})$ with $J^{1/4}(Q_{J} h-h)(X_{i})$ for
$h\in \mathcal H_{0}$.

Finally, we control the diagonal elements of
$ J^{1/4}\|n^{-1/2}\sum_{i}(h-Q_{J} h)(X_{i})\widehat Ab^{K}(W_{i})
\|$. To do so, we make use of the decomposition
\begin{equation*}
\begin{split}&\frac{\sqrt J}{n}\sum_{i} \bigl\llVert (h-
Q_{J} h) (X_{i}) Ab^{K}(W_{i}) \bigr
\rrVert ^{2} +\frac{\sqrt J}{n}\sum_{i}
\bigl\llVert (h- Q_{J} h) (X_{i}) (\widehat A-A
)b^{K}(W_{i}) \bigr\rrVert ^{2}
\\
&\quad =:T_{41,J}+T_{42,J}.
\end{split}\end{equation*}
Using Lemma~E.5(i), for any
$\varepsilon >0$ we obtain
$ \mathrm{P}_{h} (\max_{J\in \mathcal I_{n}}T_{42,J}/(V_{J}\sqrt{
\log \log J})>\varepsilon )=o(1)$ uniformly for
$h\in \mathcal H_{0}$ and thus it is sufficient to consider
$T_{41,J}$. We have
\begin{align*}
\max_{J\in \mathcal I_{n}}\frac{T_{41,J}}{V_{J}\sqrt{\log
\log J}} & \lesssim \max_{J\in \mathcal I_{n}} \frac{\sqrt J \bigl( \llVert h-
Q_{J} h \rrVert _{L^{2}(X)}\zeta _{J}s_{J}^{-1}
\bigr)^{2}}{V_{J}\sqrt{\log \log J}}
\end{align*}%
wpa1 uniformly for $h\in \mathcal H_{0}$. The right-hand side tends to
zero using that $\|h- Q_{J} h\|_{L^{2}(X)}=O(J^{-p/d_{x}})$ and Assumption~\ref{A:adapt:test}(iii) together with
$s_{J}^{-2}\leq \underline\sigma ^{-2}V_{J}$ (by Lemma~\ref{lower:bound:v_n}).
\end{proof}
\end{appendix}
\bibliography{BiB}
\newpage
\clearpage
\,
\vskip .7cm
\setcounter{page}{1}

\begin{center}
{\LARGE Supplement to ``Adaptive, Rate-Optimal Hypothesis Testing in Nonparametric IV Models"\par\vspace{0.6\baselineskip}}
\end{center}
\begin{center}
{\textsc{ \large Christoph Breunig}}
  \qquad \quad{\textsc{\large Xiaohong Chen}}
\end{center}
\begin{center}
{\large First version: August 2018, Revised \today}
\end{center}
\vskip 1cm
This supplementary appendix contains materials to support our main paper.
Appendix~\ref{appendix:add_simulation} presents additional simulation results.
Appendix~\ref{appendix:proofs:adapt:st} provides proofs of our results
on confidence sets in Section~4.3. Appendix~\ref{sec:tech:appendix} presents additional technical lemmas and all the
proofs.



\begin{appendix}
\setcounter{section}{2}%
\section{Additional Simulations}
\label{appendix:add_simulation}

This section provides additional simulation results. All the simulation
results are based on $5000 $ Monte Carlo replications for every experiment
and are at the nominal level $\alpha =0.05$.

\subsection{Adaptive Testing for Monotonicity: Simulation Design II}
\label{sim:monotonII}

We generate the dependent variable $Y$ according to the NPIV model
(2.1), where
%
\begin{align}
\label{design:II:h} h(x)=c_{0}\bigl(x/5+x^{2}
\bigr)+c_{A}\sin (2\pi x) , 
\end{align}
$c_{0}\in \{0,1\}$, $c_{A}\in [0,0.6]$, and $ W=\Phi (W^{*})$,
$X=\Phi  (\xi W^{*}+\sqrt{1-\xi ^{2}}\epsilon  )$,
$U=(0.3\epsilon +\sqrt{1-(0.3)^{2}}\nu )/2$, where
$(W^{*}, \epsilon ,\nu )$ follows a multivariate standard normal distribution.
This design with $(c_{0},c_{A})=(1,0)$ and $\xi \in \{0.3,0.5\}$ is the
one in \citet{CW2017}. The null hypothesis is that the NPIV function
$h(\cdot )$ is weakly increasing on the support of $X$. The null is satisfied
when $c_{A}\in [0,0.184)$, and is violated when $c_{A}\geq 0.184$. We note
that $c_{0}=0$, $c_{A} = 0.0$ corresponds to the boundary of the null hypothesis.
Note that the degree of nonlinearity/complexity of $h$ given in
\eqref{design:II:h} becomes larger as $c_{A} >0$ increases.

We implement our adaptive test $\widehat {\mathtt{T}}_{n}$ given in
(2.12) in the main paper, and the \citet{fang2019} test
for monotonicity of a NPIV function, denoted as FS. The FS test is computed
using R language translation of their Matlab program code, with their deterministically
chosen $J=3$, $K\geq 3$ and other tuning parameter choices detailed in
their 2019 arXiv version (also see the description in our main paper).
%
\begin{table}[h!]
 \begin{center}
 \setlength{\tabcolsep}{0.3em}
 \renewcommand{\arraystretch}{1.1}
{\footnotesize  \begin{tabular}{c|ccc||cc|cc|cc|c|c|c}
\hline
\textit{$n$}&$c_0$& $c_A$&$\xi$&{\it $\widehat {\mathtt{T}}_n$}&{\it $\widehat J$}&{\it $\widehat {\mathtt{T}}_n$}&{\it $\widehat J$} &{\it $\widehat {\mathtt{T}}_n$}&{\it $\widehat J$} &FS&FS&FS\\
& &  & &  \multicolumn{2}{c|}{$K(J)=2J$}&\multicolumn{2}{c|}{$K(J)=4J$}&\multicolumn{2}{c|}{$K(J)=8J$}&$K=5$&$K=12$&$K=24$\\

    \hline
 $500$	&$0$ &$0.0$ &$0.3$ & 0.004  	&3.01					&0.012  	& 3.03 &0.011  	& 3.21 &0.005 &0.013&0.018\\
 &&	& $0.5$   													& 0.016 	& 3.32 				&0.018 		& 3.38	 &0.021  	& 3.40 &0.035 &0.035 &0.036\\
& &	 &   $0.7$	    										& 0.025		& 3.57 				&0.030 		& 3.58 &0.026  	& 3.49 &0.050&0.049&0.042\\
  \hline
&$1$ &$0.0$&$0.3$  						& 0.002 	& 3.01					&0.005  	& 3.03 &0.003  	& 3.12 &0.000 &0.000&0.000 \\
& &	& $0.5$   													& 0.004 	& 3.38 				&0.004		& 3.36 &0.004  	& 3.25 &0.000  &0.000&0.000 \\
& &	 &   $0.7$	    										& 0.004 	&  3.71 				&0.004 		& 3.65  &0.004  	& 3.38 & 0.000&0.000&0.000 \\
 \hline
 &$1$&$0.1$&$0.3$  						& 0.002 	& 3.01					&0.006  	& 3.03 &0.005  	& 3.12 &0.000  &0.000&0.000\\
& &	& $0.5$   													& 0.007 	& 3.37					&0.007 		& 3.35 &0.007  	& 3.25 &0.001  &0.001&0.001\\
& &	 &   $0.7$	    										& 0.009 	&  3.64				&0.008 		& 3.59 &0.008  	& 3.34 & 0.000&0.000&0.000\\
  \hline
$1000$&$0$ &$0.0$ &$0.3$ & 0.009  	&3.01 					&0.016 		& 3.07 &0.015  	& 3.27 &0.011 &0.021 &0.026\\
 &&	& $0.5$   													& 0.023 	& 3.50					&0.025 		& 3.47 &0.028  	& 3.45 &0.051 &0.046&0.044\\
& &	 &   $0.7$	    										& 0.034		& 3.87 				&0.034 		& 3.97 &0.034  	& 3.52 &0.059 &0.055 &0.047\\
  \hline
&$1$ &$0.0$&$0.3$  						& 0.003 	& 3.02 				& 0.005  	& 3.06 &0.004  	& 3.15 &0.000 &0.000 &0.000\\
& &	& $0.5$   													& 0.006 	&  3.63				&0.005		& 3.46 &0.006  	& 3.28 &0.000  &0.000&0.000 \\
& &	 &   $0.7$	    										& 0.003 		&  4.23 			&0.003 		& 4.22 &0.003  	& 3.46 & 0.000&0.000&0.000\\
 \hline
 &$1$&$0.1$&$0.3$  						& 0.004 	& 3.02					&0.008  	& 3.06 &0.005  	& 3.15 &0.000  &0.001&0.001\\
& &	& $0.5$   													& 0.009 	& 3.59					&0.009 		& 3.44 &0.010  	& 3.29 &0.001  &0.001&0.001\\
& &	 &   $0.7$	    										& 0.011 		& 4.09				&0.010 		& 4.10  &0.009  	& 3.38 & 0.000& 0.000&0.000\\
 \hline
$5000$&$0$ &$0.0$&$0.3$ & 0.020  	&3.38 					&0.019 		& 3.42 &0.026  	& 3.39 &0.040 &0.040 & 0.044 \\
 &&	& $0.5$   													& 0.038 	& 3.56 				&0.036 		& 3.62 &0.035  	& 3.49 &0.056 &0.057 &0.055\\
& &	 &   $0.7$	    										& 0.045		& 4.14 				&0.042  	& 4.12 &0.035  	& 3.75 &0.056 &0.059 & 0.058\\
  \hline
&$1$ &$0.0$&$0.3$  						& 0.005 	& 3.44 				&0.006  	& 3.35 &0.006  	& 3.23 &0.000  &0.001 &0.000\\
& &	& $0.5$   													& 0.004 	& 3.81					& 0.003 	& 3.80 &0.003  	& 3.47 &0.000  &0.000&0.000 \\
& &	 &   $0.7$	    										& 0.002 		&  4.74 			&0.002 		& 4.69  &0.002  	& 3.98 & 0.000&0.000&0.000\\
 \hline
 &$1$&$0.1$&$0.3$  						& 0.009 	& 3.42					& 0.008  	& 3.35 &0.009  	& 3.24 &0.001 &0.002&0.001 \\
& &	& $0.5$   													& 0.013 	&  3.70				&0.013 		& 3.69 &0.011  	& 3.40 &0.000  &0.000&0.000 \\
& &	 &   $0.7$	    										& 0.008 	&  4.52				&0.006 		& 4.46  &0.006  	& 3.75 & 0.000&0.000&0.000\\
  \hline
 \end{tabular}}
 \end{center}
 \vspace*{-6mm}
 \caption{{\small Testing Monotonicity - Empirical Size of our adaptive test $\widehat {\mathtt{T}}_n$ and of the FS test (with $J=3$). Monte Carlo average value $\widehat J$. Nominal level $\alpha=0.05$.  Design from Appendix \ref{sim:monotonII} with NPIV function \eqref{design:II:h}. Instrument strength increases in $\xi$.} }\label{table:size-add-ER}
 \end{table}

Table~\ref{table:size-add-ER} reports the empirical size of our adaptive
test $\widehat {\mathtt{T}}_{n}$, with $K(J)\in \{ 2J, 4J, 8J\}$, and using
quadratic B-spline basis functions with varying number of knots for the
unrestricted NPIV $h$. We also report the empirical size of the FS test,
using $J=3$ and $K\in \{5,  12, 24\}$ as comparison to our adaptive test's
$K(J)\in \{2J, 4J, 8J\}$. From Table~\ref{table:size-add-ER}, we observe
that our adaptive test $\widehat {\mathtt{T}}_{n}$ is slightly under-sized
across different sample sizes, different instrument strength, different
$K(J)$, and different design specifications. The FS test is mostly under-sized,
but is slightly over-sized at the boundary ($c_{0}=0$, $c_{A}=0.0$) for sample
sizes $n=1000, 5000$ and strong instrument strength $\xi =0.7$ even when
$J=3$, $K=5$ (the most powerful choice in the 2019 arXiv version of
\citet{fang2019}).

\begin{figure}
 \begin{center}
    \includegraphics[scale=0.8]{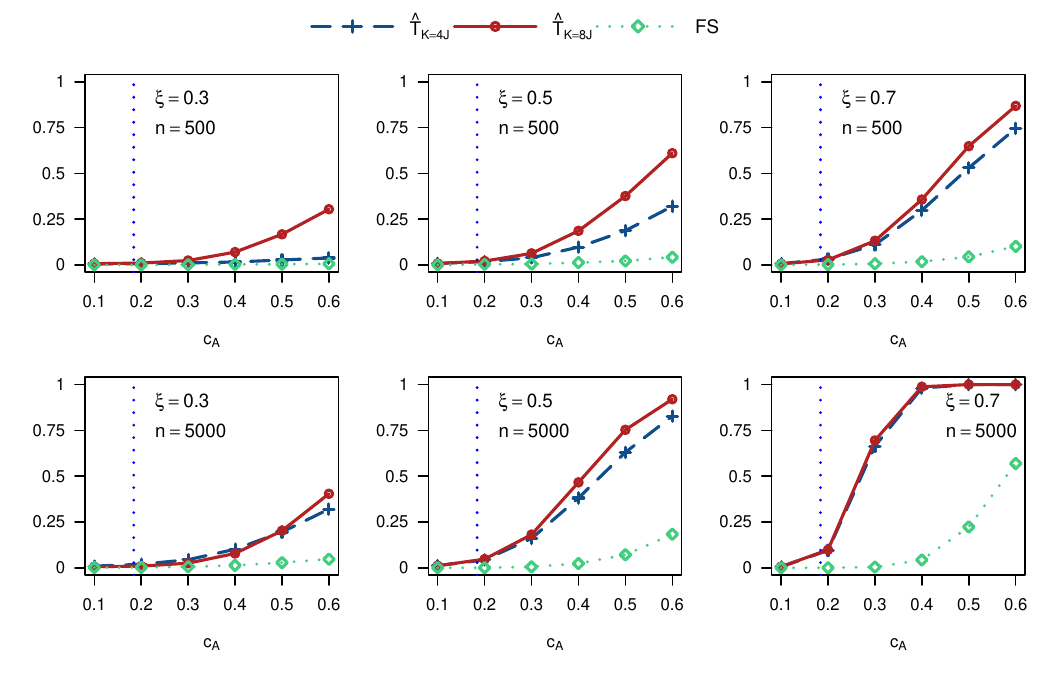}
    \vskip -.6cm
\caption{Testing monotonicity---empirical power of our adaptive
test $\widehat {\mathtt{T}}_{n}$ with $K(J)=4J$ (dashed plus lines) and
$K(J)=8J$ (solid circle lines) and the FS test (with $J=3$, $K=5$, dotted
square lines). Design from Appendix~\protect\ref{sim:monotonII} model
\protect\eqref{design:II:h} with $c_{0}=1$. The vertical dotted line indicates
when the null hypothesis is violated (when $c_{A}\geq 0.184$). Instrument
strength increases in $\xi $.}%
\label{fig:power-add-ER}
 \end{center}
\end{figure}

Figure~\ref{fig:power-add-ER} provides empirical rejection probabilities
of our adaptive test $\widehat {\mathtt{T}}_{n}$ (dashed plus and solid
circle lines) with $K(J)\in \{4J,  8J\}$ and of the FS test (with
$J=3$, $K=5$; dotted square lines). The power curves of all tests improve
as the instrument strength $\xi $ increases. Our adaptive test with
$K(J)=8J$ has better empirical power in finite samples when instrument
is weak, but the choice of $K(J)$ is less significant as the sample size
or the instrument strength increases. For instrument strength
$\xi =0.3$, the FS test has almost trivial power for
$c_{A} \in [0.2,0.5]$ even for large sample size $n=5000$, while our adaptive
test $\widehat {\mathtt{T}}_{n}$ has non-trivial power for all
$c_{A} \geq 0.3$. Moreover, the finite-sample power of our adaptive test
$\widehat {\mathtt{T}}_{n}$ increases much faster than the FS test as
$c_{A}> 0.2$ becomes larger. Figure~\ref{fig:power-add-ER} shows the substantial
finite-sample power gains through adaptation even in small sample size
$n=500$.

\begin{remark}
\label{remC.1}
When testing for inequality restrictions (IR)
$\mathcal H_{0}= \{h\in \mathcal H:  \partial ^{l} h\geq 0 \}$,
such as monotonicity and convexity, we could also compute our adaptive
test $\widehat {\mathtt{T}}_{n}$ using modified critical values in Step~2 as
follows: The estimator in (2.6) can be written as
$\widehat h_{J}^{\text{\textsc r}}(\cdot )=\psi ^{J}(\cdot )'\widehat \beta ^{
\text{\textsc r}}$. By construction of the estimator, we have
$ \partial ^{l} \widehat h_{J}^{\text{\textsc r}}(X_{i})\geq 0$, for all
$1\leq i\leq n$, or equivalently
$ \partial ^{l} \Psi \widehat \beta ^{\text{\textsc r}}\geq 0$, where the application
of the derivative operator is understood elementwise and
$\text{rank}( \partial ^{l} \Psi )\leq J$. Let $\Psi _{\mathrm{act}}$ be a submatrix
of $\Psi $ such that
$ \partial ^{l}\Psi _{\mathrm{act}}\widehat \beta ^{\text{\textsc r}}=0$. Set
$\widehat\gamma _{J}=\max   (1,\text{rank}( \partial ^{l} \Psi _{\mathrm{act}})
  )$ and compute for a given nominal level $\alpha \in (0,1)$:
%
\begin{align}
\label{def:eta_J:constraint} \widehat \eta _{J} (\alpha )= \frac{q \bigl(\alpha /\#(
\widehat{\mathcal I}_{n}), \widehat\gamma _{J} \bigr)-
\widehat\gamma _{J}}{\sqrt{\widehat\gamma _{J}}}, 
\end{align}
where $q (a, \gamma )$ denotes the $100(1-a)\%$-quantile of the chi-squared
distribution with $\gamma $ degrees of freedom. Assuming that
$J^{c}\leq \widehat\gamma _{J}$, $J\in \mathcal I_{n}$, for some constant
$0<c\leq 1$ with probability approaching 1 uniformly for
$h\in \mathcal H$, \citet{BC2021} established size control of the test
statistic using the modified critical values given in
\eqref{def:eta_J:constraint}. See \citet{BC2021} also for simulations and
real data application of testing for monotonicity and convexity using these
modified critical values. The simulations and empirical findings reported
in \citet{BC2021} are virtually the same, in terms of empirical size and
power, as the ones reported in this revised version for testing inequalities.
\end{remark}

\subsection{Simulations for Multivariate Instruments}
\label{sim:parII}

This section presents additional simulations for testing parametric hypotheses
in the presence of multivariate conditioning variable
$W = (W_{1},W_{2})$. We set $X_{i} =\Phi (X_{i}^{*})$,
$ W_{1i} =\Phi (W_{1i}^{*})$, and $ W_{2i} =\Phi (W_{2i}^{*})$, where
\begin{align}\label{design:II}
\begin{pmatrix}
X_i^*\\
W_{1i}^*\\
W_{2i}^*\\
U_i
\end{pmatrix}
\sim
\mathcal N\left(\begin{pmatrix}
0\\
0\\
0\\
0
\end{pmatrix},
\begin{pmatrix}
1 &\xi &0.4&0.3\\
\xi &1 &0 &0\\
0.4&0 &1 &0\\
0.3 &0&0&1
\end{pmatrix}\right)~.
\end{align}

 We generate the dependent variable $Y$ according to the NPIV model
(2.1) where $h(x)= -x/5+c_{A} x^{2}$. We test the null
hypothesis of linearity, that is, whether $c_{A}=0$.

\citet{Horowitz2006} assumed $d_{x}=d_{w} $ and hence we cannot compare
our adaptive test with his for Design \eqref{design:II}. Instead, we will
compare our adaptive test $\widehat {\mathtt{T}}_{n}$ against an adaptive
image-space test (IT), which is our proposed adaptive version of
\citet{bierens1990}'s type test for semi-nonparametric conditional moment
restrictions.\footnote{We refer readers to \citet{BC2020} for the theoretical
properties of the adaptive image-space test.} Specifically, our image-space
test (IT) is based on a leave-one-out sieve estimator of the quadratic
functional $\Evtex [\Evtex [Y-h^{\text{\textsc r}}(X)|W]^{2}]$, given by
\begin{align*}
\widehat{D}_{K}=\frac{2}{n(n-1)}\sum_{1\leq i< i'\leq n}
\bigl(Y_{i}- \widehat h^{\text{\textsc r}}(X_{i}) \bigr)
\bigl(Y_{i'}-\widehat h^{\text{\textsc r}}(X_{i'})
\bigr)b^{K}(W_{i})'\bigl(B'B/n
\bigr)^{-} b^{K}(W_{i'}),
\end{align*}
where $\widehat h^{\text{\textsc r}}$ is a null restricted parametric estimator
for the null parametric function $h^{\text{\textsc r}}$. The data-driven IT statistic
is
\begin{align*}
\widehat{\mathtt {IT}}_{n} = \1 \bigl\{\text{there exists } K\in
\widehat{\mathcal I}_{n}\text{ such that } n\widehat{D}_{K}/
\widehat{V}_{K}> \bigl(q \bigl(\alpha /\#(\widehat{\mathcal
I}_{n}), K \bigr)-K \bigr)/\sqrt{K} \bigr\},
\end{align*}
with the estimator
$\widehat{V}_{K}= \|(B'B)^{-1/2}\sum_{i=1}^{n} (Y_{i}-\widehat h^{
\text{\textsc r}}(X_{i}))^{2} b^{K}(W_{i})b^{K}(W_{i})'(B'B)^{-1/2} \|_{F}$,
and the adjusted index set
$\widehat{\mathcal I}_{n}=\{K\leq \widehat K_{\max}: K=\underline K2^{k}
\text{ where } k=0,1,\dots ,k_{\max}\}$, where
$\underline K:=\lfloor \sqrt{\log \log n}\rfloor $,
$k_{\max}:=\lceil \log _{2}(n^{1/3}/\underline K)\rceil $, and the empirical
upper bound
$\widehat K_{\max}=\min  \{K>\underline K:10 \zeta ^{2}(K)\sqrt{ (
\log K)/n}\geq s_{\min} ((B'B/n)^{-1/2} ) \}$. Finally,
$q (a, K)$ is the $100(1-a)\%$-quantile of the chi-squared distribution
with $K$ degrees of freedom. In this simulation, it is convenient to additionally
weight the basis functions by $(B'B/n)^{-1/2}$ to improve the finite-sample
performance of the IT statistic.
%
\begin{table}[h!]
 \begin{center}
 \setlength{\tabcolsep}{0.3em}
 \renewcommand{\arraystretch}{1.1}
{\footnotesize \begin{tabular}{c|cc||cc|cc}
\hline
\textit{$n$}&Design& $\xi$&{\it $\widehat {\mathtt{T}}_n$, $K(J)=4J$}&{\it $\widehat J$}\qquad  &\qquad{\it $\widehat {\mathtt{IT}}_n$}&\quad{\it $\widehat K$}\\
\hline
$500$& \eqref{design:I} 	&$0.3$  			& 	 0.023 	&3.03 \qquad &\qquad 0.046	&\quad 3.38	\\
									& $d_x=d_w$  &$0.5$  			& 	 0.028 	&3.40 \qquad &\qquad 0.046	&\quad 3.37	\\
									 												&	&$0.7$ 				&  0.035 	&3.56 \qquad &\qquad  0.046	&\quad 3.37\\
  \hline
   & \eqref{design:II} 			&$0.3$  			& 0.034		&3.45 \qquad & \qquad 0.034	&\quad 6.00\\
 								&$d_x<d_w$  &$0.5$  			& 0.035 	&3.49 \qquad & \qquad 0.035	&\quad 6.00\\
 																		&	&$0.7$	  			&  0.038 	&3.55 \qquad &\qquad 0.040	&\quad 6.00\\
  \hline
   $1000$& \eqref{design:I} &$0.3$  			& 	 0.022 	&3.07 \qquad &\qquad 0.053	&\quad 3.40\\
 																					&   &$0.5$  			& 	 0.027 	&3.48 \qquad &\qquad 0.051	&\quad 3.39	\\
																					 &	 &$0.7$ 			&  0.037 	&3.58 \qquad &\qquad  0.049	&\quad 3.39	\\
  \hline
   & \eqref{design:II} 			&$0.3$  			& 0.039		&3.47 \qquad & \qquad 0.032	&\quad 6.93\\
									 & 									&$0.5$  			& 0.040		&3.50 \qquad & \qquad 0.038	&\quad 6.92\\
 									&										 &   $0.7$	    		&0.043 	&3.58 \qquad &\qquad 0.037	&\quad 6.90\\
  \hline
  $5000$& \eqref{design:I} &$0.3$  			& 	 0.032 	&3.43 \qquad &\qquad 0.049	&\quad 3.38	\\
 																				&   &$0.5$  			& 	  0.043 	&3.55 \qquad &\qquad 0.045	&\quad 3.39	\\
																				 &	 &$0.7$ 			&  0.049	&3.63 \qquad &\qquad  0.042	&\quad 3.38	\\
  \hline
   & \eqref{design:II} 			&$0.3$  			& 0.049		&3.51 \qquad & \qquad 0.048	&\quad 10.28\\
									 & 									&$0.5$  			& 0.048	&3.57 \qquad & \qquad 0.046	&\quad 10.27\\
 								&										 &   $0.7$	    		&0.050 	&3.80 \qquad &\qquad 0.051	&\quad 10.25\\
  \hline
 \end{tabular}}
 \end{center}
  \vspace*{-6mm}
 \caption{{\small Testing Parametric Form - Empirical size of our adaptive tests $\widehat {\mathtt{T}}_n$ and of $\widehat {\mathtt{IT}}_{n}$. Nominal level $\alpha=0.05$. Monte Carlo average value $\widehat J$.  Design from Appendix \ref{sim:parII}. Instrument strength increases in $\xi$.} }\label{table:size-mult}
 \end{table}
Table~\ref{table:size-mult} compares the empirical size of the adaptive
image-space test $\widehat {\mathtt{IT}}_{n}$ with our adaptive structural-space
test $\widehat {\mathtt{T}}_{n}$, at the $5\%$ nominal level. We see that
both tests provide accurate size control. We also report the average choices
of sieve dimension parameters, as described in Section~5. The
multivariate design \eqref{design:II} leads to larger sieve dimension choices
$\widehat{K}$ in adaptive image-space tests
$\widehat {\mathtt{IT}}_{n}$, while the sieve dimension choices
$\widehat{J}$ of our adaptive structural-space test
$\widehat {\mathtt{T}}_{n}$ are not sensitive to the dimensionality ($d_{w}$)
of the conditional instruments.

\begin{figure}[h!]
 \begin{center}
    \includegraphics[scale=0.8]{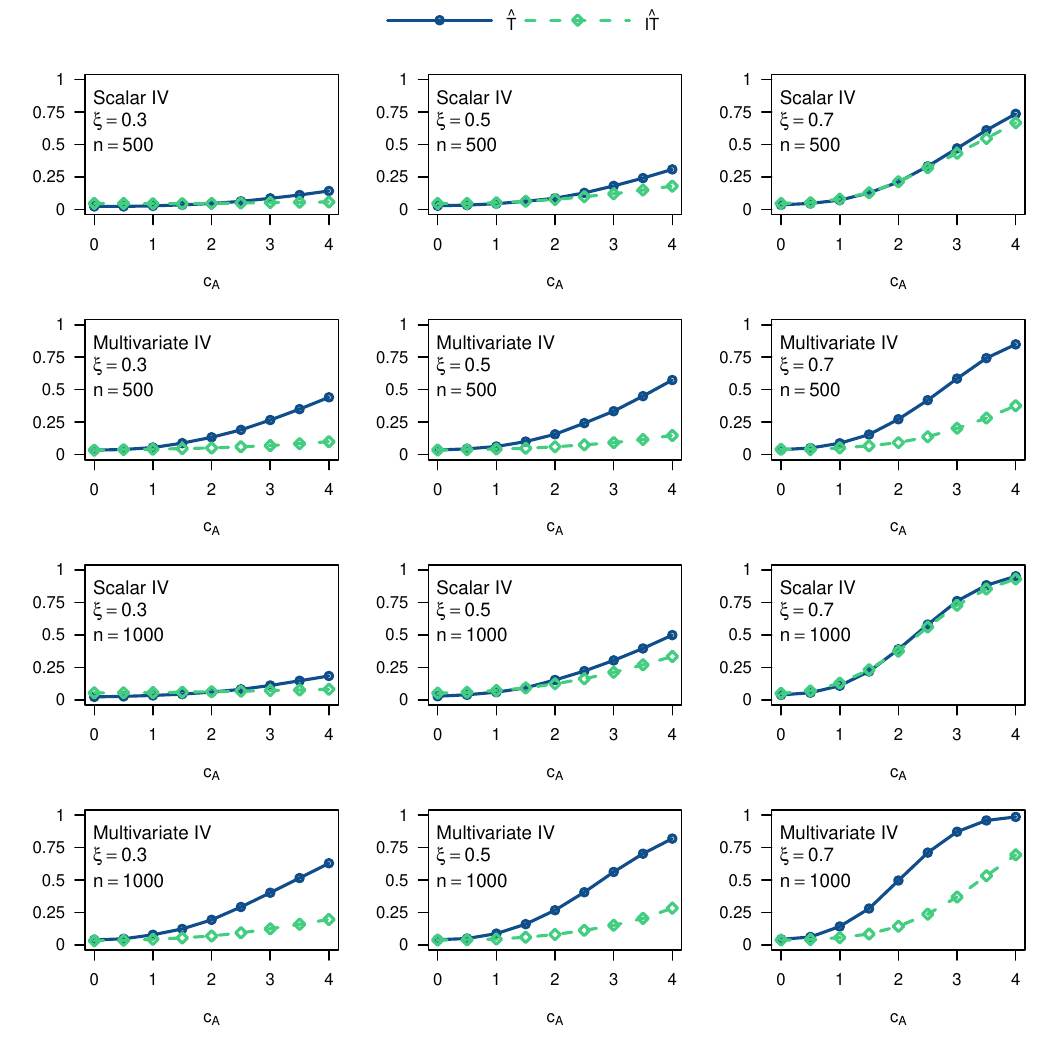}
    \vskip -.6cm
\caption{Testing parametric form---empirical power of our adaptive
tests $\widehat {\mathtt{T}}_{n}$ (solid circle lines) and of
$\widehat {\mathtt{IT}}_{n}$ (dashed square lines). First and third rows:
power comparisons in scalar IV case ($d_{w} =1$); second and fourth rows:
power comparisons in multivariate IV case ($d_{w}>1$). Design from Appendix~\protect\ref{sim:parII}. Instrument strength increases in $\xi $.}%
\label{fig:power:mult}
 \end{center}
\end{figure}

Figure~\ref{fig:power:mult} compares the empirical power of
$\widehat {\mathtt{IT}}_{n}$ and of $\widehat {\mathtt{T}}_{n}$, at the
$5\%$ nominal level, using the sample sizes $n=500$ (first and second rows)
and $n=1000$ (third and fourth rows). The finite-sample empirical power
curves of both tests increase with $\xi $ and sample size $n$. For the
scalar conditional instrument case, while our adaptive structural-space
test $\widehat {\mathtt{T}}_{n}$ is more powerful when
$\xi \in \{0.3, 0.5\}$ (weaker strength of instruments), the finite-sample
power curves of both tests are similar when $\xi =0.7$. For the multivariate
conditional instruments case, while the power of our adaptive structural-space
test $\widehat {\mathtt{T}}_{n}$ increases with larger dimension
$d_{w}$, the adaptive image-space test $\widehat {\mathtt{IT}}_{n}$ suffers
from larger $d_{w}$ and has lower power. The same patterns are also present
when we compare the two tests using size-adjusted empirical power curves
(see our arxiv:2006.09587v3 version, Appendix~C.3).

\section{Proofs of Inference Results in Section~4.3}
\label{appendix:proofs:adapt:st}

\begin{proof}[Proof of Corollary~4.1]
Proof of (4.8). We observe
\begin{align*}
\limsup_{n\to \infty}\sup_{h\in \mathcal H_{0}}\mathrm{P}_{h}
\bigl( h\notin \mathcal C_{n}(\alpha ) \bigr) =\limsup
_{n\to \infty}\sup_{h
\in \mathcal H_{0}}\mathrm{P}_{h}
\biggl( \max_{J\in
\widehat{\mathcal I}_{n}} \frac{n\widehat{D}_{J}(h)}{\widehat
\eta _{J}(\alpha ) \widehat{V}_{J}}> 1 \biggr)\leq \alpha ,
\end{align*}
where the last inequality is due to Step~1 and Step~3 of the proof of Theorem~4.1. Indeed, in that proof, we can replace $\mathrm{P}_{h_0}$ by $\sup_{h\in \mathcal H_{0}}\mathrm{P}_{h}$ by adopting the uniform moment conditions imposed in Assumption 2(i). 

Proof of (4.9). Let $J^{*}$ be as in Step~2 of the proof
of Theorem~4.1. We observe uniformly for
$h\in \mathcal H_{1}(\delta ^{\circ}\text{\textsf r}_{n})$ that
\begin{align*}
\mathrm{P}_{h} \bigl( h\notin \mathcal C_{n}(\alpha )
\bigr)&= \mathrm{P}_{h} \biggl( \max_{J\in \widehat{\mathcal I}_{n}} \frac{n
\widehat{D}_{J}(h)}{\widehat \eta _{J}(\alpha )
\widehat{V}_{J}}> 1 \biggr)=1-\mathrm{P}_{h} \biggl( \max
_{J\in \widehat{\mathcal I}_{n}} \frac{n\widehat{D}_{J}(h)}{\widehat \eta
_{J}(\alpha ) \widehat{V}_{J}} \leq 1 \biggr) = 1-o(1),
\end{align*}
where the last equation is due to Step~2 and Step~3 of the proof of Theorem~4.1.
\end{proof}

\begin{proof}[Proof of Corollary~4.2]
For any $h\in \mathcal H_{0}$, we analyze the diameter of the confidence
set $\mathcal C_{n}(\alpha )$ under $\mathrm{P}_{h}$. Lemma~B.8 implies
$\sup_{h\in \mathcal H_{0}}\mathrm{P}_{h} (\widehat J_{\max}>
\overline J )=o(1)$ and hence, it is sufficient to consider the deterministic
index set $\mathcal I_{n}$ given in (4.2). For all
$h_{1}\in \mathcal C_{n}(\alpha )\subset \mathcal H_{0}$, it holds for
all $J\in \mathcal I_{n}$ by using the definition of the projection
$Q_{J}$ given in (B.1):
%
\begin{align}
\llVert h-h_{1} \rrVert _{L^{2}(X)}&\leq \bigl\llVert
Q_{J}\Pi _{J} (h-h_{1}) \bigr\rrVert
_{L^{2}(X)}+ \llVert \Pi _{J} h - h \rrVert _{L^{2}(X)}+
\llVert \Pi _{J} h_{1} - h_{1} \rrVert
_{L^{2}(X)}
\nonumber
\\
&\leq \bigl\llVert Q_{J} (h-h_{1}) \bigr\rrVert
_{L^{2}(X)}+O\bigl(J^{-p/d_{x}}\bigr) \label{ineq:diam} , 
\end{align}
due to the triangular inequality and the sieve approximation bound from
the smoothness restrictions imposed on $\mathcal H$. By Theorem~B.1, we have
\begin{align*}
\bigl\llvert \bigl\llVert Q_{J} (h-h_{1}) \bigr\rrVert
_{L^{2}(X)}^{2}-\widehat D_{J}(h_{1}) \bigr
\rrvert & \lesssim n^{-1}s_{J}^{-2}
\sqrt{J}+n^{-1/2}s_{J}^{-1} \bigl( \llVert
h-h_{1} \rrVert _{L^{2}(X)}+J^{-p/d_{x}} \bigr)
\end{align*}
wpa1 uniformly for $h\in \mathcal H_{0}$. Consequently, the definition
of the confidence set $\mathcal C_{n}(\alpha )$ with
$h_{1}\in \mathcal C_{n}(\alpha )$ gives for all
$J\in \mathcal I_{n}$:
\begin{align*}
\bigl\llVert Q_{J} (h-h_{1}) \bigr\rrVert
_{L^{2}(X)}^{2} &\lesssim n^{-1}\widehat \eta
_{J}( \alpha ) \widehat{V}_{J}+n^{-1/2}s_{J}^{-1}
\bigl( \llVert h-h_{1} \rrVert _{L^{2}(X)}+J^{-p/d_{x}}
\bigr)+n^{-1}s_{J}^{-2} \sqrt{J}
\\
&\lesssim n^{-1}\sqrt{\log \log n} s_{J}^{-2}
\sqrt{J}+n^{-1/2}s_{J}^{-1} \bigl( \llVert
h-h_{1} \rrVert _{L^{2}(X)}+J^{-p/d_{x}} \bigr)
\end{align*}
wpa1 uniformly for $h\in \mathcal H_{0}$ by using Lemmas
B.2, B.5, and B.4(ii).
Consequently, inequality \eqref{ineq:diam} yields for all
$J\in \mathcal I_{n}$:
\begin{align*}
\llVert h-h_{1} \rrVert _{L^{2}(X)}^{2}&\lesssim
\frac{n^{-1}\sqrt{\log \log n} s_{J}^{-2}
\sqrt{J}+J^{-2p/d_{x}}}{1-C_{B}n^{-1/2}s_{J}^{-1}}
\end{align*}
wpa1 uniformly for $h\in \mathcal H_{0}$. Now using that
$n^{-1/2}s_{J}^{-1}=o(1)$ for all $J\in \mathcal I_{n}$, by Assumption~4(i) we obtain
$\|h-h_{1}\|_{L^{2}(X)}\lesssim n^{-1/2}(\log \log n)^{1/4} s_{J}^{-1}
J^{1/4}+J^{-p/d_{x}}$ with probability approaching 1 uniformly for
$h\in \mathcal H_{0}$. Also, by Assumption~3 we have $s_{J}^{-1}\lesssim \nu_{J}^{-1}$. We may choose $J=cJ^\circ \in \mathcal I_{n}$ for
some constant $c>0$ and $n$ sufficiently large and hence, the result follows.
\end{proof}

\section{Technical Results}
\label{sec:tech:appendix}

Below, $\lambda _{\max}(\cdot )$ denotes the maximal eigenvalue of a matrix.

\begin{lemma}%
\label{lemma:est:matrices}
Let Assumptions 1(ii)--(iii) and 2 hold. Then, wpa1
uniformly for $h\in \mathcal H$:
\begin{align*}
&\frac{1}{n(n-1)}\sum_{i\neq i'} \bigl(Y_{i}-
\Pi _{\mathcal H_0}h(X_{i}) \bigr) \bigl(Y_{i'}-\Pi
_{\mathcal H_0}h(X_{i'}) \bigr)b^{K}(W_{i})'
\bigl(A' A-\widehat A'\widehat A \bigr)b^{K}(W_{i'})
\\
&\quad \lesssim n^{-1} V_{J}+n^{-1/2}s_{J}^{-1}
\bigl( \llVert h-\Pi _{\mathcal H_{0}} h \rrVert _{L^{2}(X)}+J^{-p/d_{x}}
\bigr).
\end{align*}
\end{lemma}
\begin{proof}
Let $\Pi _{\mathcal H_0}^{\perp }:=\text{id}-\Pi _{\mathcal H_0}$. We establish
an upper bound of
\begin{align*}
&\frac{1}{n^{2}}\sum_{i,i'}
\bigl(Y_{i}-\Pi _{\mathcal H_0}h(X_{i}) \bigr)
\bigl(Y_{i'}-\Pi _{\mathcal H_0}h(X_{i'})
\bigr)b^{K}(W_{i})' \bigl(A' A-
\widehat A' \widehat A \bigr)b^{K}(W_{i'})
\\
&\quad =\Evtex \bigl[\Pi _{\mathcal H_0}^{\perp }h(X)b^{K}(W)
\bigr]' \bigl(A'A-\widehat A' \widehat A
\bigr)\Evtex \bigl[\Pi _{\mathcal H_0}^{\perp }h(X)b^{K}(W)\bigr]
\\
&\qquad {}+2 \biggl(\frac{1}{n}\sum_{i}
\bigl(Y_{i}-\Pi _{\mathcal H_0}h(X_{i})
\bigr)b^{K}(W_{i})-\Evtex \bigl[\Pi _{\mathcal H_0}^{\perp }h(X)b^{K}(W)
\bigr] \biggr)'
\\
&\qquad {}\times \bigl(A'A-\widehat A'
\widehat A \bigr)\Evtex \bigl[\Pi _{\mathcal H_0}^{
\perp }h(X)b^{K}(W)
\bigr]
\\
&\qquad {}+ \biggl(\frac{1}{n}\sum_{i}
\bigl(Y_{i}-\Pi _{\mathcal H_0}h(X_{i})
\bigr)b^{K}(W_{i})'-\Evtex \bigl[\Pi
_{\mathcal H_0}^{\perp }h(X)b^{K}(W)\bigr]' \biggr)
\bigl(A'A-\widehat A' \widehat A \bigr)
\\
&\qquad {} \times \biggl(\frac{1}{n}\sum_{i}
\bigl(Y_{i}- \Pi _{\mathcal H_0}h(X_{i})
\bigr)b^{K}(W_{i})'-\Evtex \bigl[\Pi
_{\mathcal H_0}^{
\perp }h(X)b^{K}(W)\bigr]' \biggr)
\end{align*}
uniformly for $h\in \mathcal H$. It is sufficient to bound the first summand
on the right-hand side. We make use of the decomposition
\begin{align*}
&\Evtex \bigl[\Pi _{\mathcal H_0}^{\perp }h(X)b^{K}(W)
\bigr]' \bigl(A'A-\widehat A' \widehat A
\bigr)\Evtex \bigl[\Pi _{\mathcal H_0}^{\perp }h(X)b^{K}(W)\bigr]
\\
&\quad = 2\Evtex \bigl[\Pi _{\mathcal H_0}^{\perp }h(X)b^{K}(W)
\bigr]'A'(A-\widehat A)\Evtex \bigl[ \Pi
_{\mathcal H_0}^{\perp }h(X)b^{K}(W)\bigr]
\\
&\qquad {} -\Evtex \bigl[\Pi _{\mathcal H_0}^{\perp }h(X)b^{K}(W)
\bigr]'(A-\widehat A)' (A- \widehat A)\Evtex \bigl[\Pi
_{\mathcal H_0}^{\perp }h(X)b^{K}(W)\bigr]=:2
T_{1}-T_{2}.
\end{align*}
We first consider the term $T_{1}$ as follows:
%
\begin{align}
\label{T1:dec} T_{1} ={}&\Evtex \bigl[\Pi _{\mathcal H_0}^{\perp }h(X)b^{K}(W)
\bigr]'A'(\widehat A-A) \Evtex \bigl[\Pi _{J}
\Pi _{\mathcal H_0}^{\perp }h(X) b^{K}(W)\bigr]
\nonumber
\\
&{} +\Evtex \bigl[\Pi _{\mathcal H_0}^{\perp }h(X)b^{K}(W)
\bigr]'A'(\widehat A-A) \Evtex \bigl[\bigl(\Pi
_{\mathcal H_0}^{\perp }h-\Pi _{J}\Pi _{\mathcal H_0}^{
\perp }h
\bigr) (X)b^{K}(W)\bigr]\nonumber
\\
:={}&A_{1}+A_{2}. 
\end{align}
We now consider the term $A_{1}$. Recall that
$Q_{J} \Pi _{J}h=\Pi _{J}h$ and
$\widehat S G^{-1}\langle h,\psi ^{J}\rangle _{L^{2}(X)}=n^{-1}\sum_{i}
\Pi _{J} h(X_{i})b^{K}(W_{i})$. We have
\begin{align*}
& \bigl(\bigl(G_{b}^{-1/2}S\bigr)_{l}^{-}
\Evtex \bigl[\Pi _{\mathcal H_0}^{\perp }h(X) \widetilde{b}^{K}(W)
\bigr] \bigr)'G \bigl(\bigl(G_{b}^{-1/2}S
\bigr)_{l}^{-}-\bigl(\widehat G_{b}^{-1/2}
\widehat S\bigr)_{l}^{-}\widehat G_{b}^{-1/2}G_{b}^{1/2}
\bigr)
\\
&\qquad {}\times \Evtex \bigl[\Pi _{J} \Pi _{\mathcal H_0}^{\perp }h(X)
\widetilde{b}^{K}(W)\bigr]
\\
&\quad = \bigl\langle Q_{J} \Pi _{\mathcal H_0}^{\perp }h, \Pi
_{J}\Pi _{
\mathcal H_0}^{\perp }h-\bigl(\psi ^{J}
\bigr)'\bigl(\widehat G_{b}^{-1/2}\widehat S
\bigr)_{l}^{-} \widehat G_{b}^{-1/2}\Evtex
\bigl[\Pi _{\mathcal H_0}^{\perp }h(X)b^{K}(W)\bigr] \bigr
\rangle _{L^{2}(X)}
\\
&\quad =\bigl\langle Q_{J} \Pi _{\mathcal H_0}^{\perp }h,\psi
^{J}\bigr\rangle _{L^{2}(X)}'\bigl( \widehat
G_{b}^{-1/2}\widehat S\bigr)_{l}^{-}
\widehat G_{b}^{-1/2}
\\
&\qquad {}\times \biggl( \frac{1}{n}\sum
_{i} \Pi _{J}\Pi _{\mathcal H_0}^{\perp }h(X_{i})b^{K}(W_{i})-
\Evtex \bigl[\Pi _{\mathcal H_0}^{\perp }h(X)b^{K}(W)\bigr]
\biggr)
\\
&\quad =\bigl\langle Q_{J} \Pi _{\mathcal H_0}^{\perp }h,\psi
^{J}\bigr\rangle _{L^{2}(X)}' \bigl(G_{b}^{-1/2}S
\bigr)^{-}_{l}
\\
&\qquad {}\times  \biggl(\frac{1}{n}\sum
_{i} \Pi _{J}\Pi _{
\mathcal H_0}^{\perp }h(X_{i})
\widetilde{b}^{K}(W_{i})-\Evtex \bigl[\Pi _{J} \Pi
_{\mathcal H_0}^{\perp }h(X)\widetilde{b}^{K}(W)\bigr] \biggr)
\\
&\qquad {} +\bigl\langle Q_{J} \Pi _{\mathcal H_0}^{\perp }h,\psi
^{J}\bigr\rangle _{L^{2}(X)}'\bigl(G_{b}^{-1/2}S
\bigr)^{-}_{l} G_{b}^{-1/2}
S' \bigl(\bigl(\widehat G_{b}^{-1/2}\widehat S
\bigr)_{l}^{-} \widehat G_{b}^{-1/2}G_{b}^{1/2}-
\bigl(G_{b}^{-1/2}S\bigr)^{-}_{l} \bigr)
\\
&\qquad {} \times \biggl(\frac{1}{n}\sum_{i} \Pi
_{J}\Pi _{\mathcal H_0}^{
\perp }h(X_{i})\widetilde{b}^{K}(W_{i})-\Evtex \bigl[\Pi _{J}\Pi
_{\mathcal H_0}^{
\perp }h(X)\widetilde{b}^{K}(W)\bigr]
\biggr)=:A_{11}+A_{12},
\end{align*}%
 where we used the notation
$\widetilde{b}^{K}(\cdot )=G_{b}^{-1/2} b^{K}(\cdot )$. Consider
$A_{11}$; we have
\begin{align*}
\Evtex \llvert A_{11} \rrvert ^{2}\leq{}& n^{-1}
\Evtex \bigl\llvert \bigl\langle Q_{J} \Pi _{\mathcal H_0}^{
\perp }h,
\psi ^{J}\bigr\rangle _{L^{2}(X)}'
\bigl(G_{b}^{-1/2}S\bigr)^{-}_{l} \Pi
_{J} \Pi _{\mathcal H_0}^{\perp }h(X)\widetilde{b}^{K}(W) \bigr\rrvert ^{2}
\\
\leq {}&2n^{-1} \bigl\llVert \bigl\langle Q_{J} \Pi
_{\mathcal H_0}^{\perp }h,\psi ^{J} \bigr\rangle
_{L^{2}(X)}' \bigl(G_{b}^{-1/2}S
\bigr)^{-}_{l} \bigr\rrVert ^{2} \bigl\llVert \Pi
_{K}T \Pi _{\mathcal H_0}^{\perp }h \bigr\rrVert
_{L^{2}(W)}^{2}
\\
& {}+2n^{-1} \bigl\llVert \bigl\langle Q_{J} \Pi
_{\mathcal H_0}^{\perp }h, \psi ^{J}\bigr\rangle
_{L^{2}(X)}' \bigl(G_{b}^{-1/2}S
\bigr)^{-}_{l} \bigr\rrVert ^{2} \bigl\llVert \Pi
_{K}T\bigl(\Pi _{\mathcal H_0}^{\perp }h-\Pi _{J}\Pi
_{\mathcal H_0}^{
\perp }h\bigr) \bigr\rrVert _{L^{2}(W)}^{2}
\\
\lesssim {}&n^{-1} \bigl\llVert \bigl\langle Q_{J} \Pi
_{\mathcal H_0}^{\perp }h, \psi ^{J}\bigr\rangle
_{L^{2}(X)}' \bigl(G_{b}^{-1/2}S
\bigr)^{-}_{l} \bigr\rrVert ^{2},
\end{align*}
where the second bound is due to the Cauchy--Schwarz inequality and the
third bound is due to Assumption~2(iv). Consider
$A_{12}$; we infer from \citet [Lemma F.10(c)]{ChenChristensen2017} and
Assumption~2(ii) that
\begin{align*}
\llvert A_{12} \rrvert ^{2}\leq {}& \bigl\llVert \bigl\langle
Q_{J} \Pi _{\mathcal H_0}^{\perp }h, \psi ^{J}\bigr
\rangle _{L^{2}(X)}'\bigl(G_{b}^{-1/2}S
\bigr)^{-}_{l} \bigr\rrVert ^{2} \bigl\llVert
G_{b}^{-1/2} S' \bigl(\bigl(\widehat
G_{b}^{-1/2}\widehat S\bigr)_{l}^{-}
\widehat G_{b}^{-1/2}G_{b}^{1/2}-
\bigl(G_{b}^{-1/2}S\bigr)^{-}_{l} \bigr)
\bigr\rrVert ^{2}
\\
& {}\times \biggl\llVert \frac{1}{n}\sum_{i} \Pi
_{J}\Pi _{\mathcal H_0}^{
\perp }h(X_{i})b^{K}(W_{i})
-\Evtex \bigl[\Pi _{J}\Pi _{\mathcal H_0}^{\perp }h(X)
b^{K}(W)\bigr] \biggr\rrVert ^{2}
\\
\lesssim{}& \bigl\llVert \bigl\langle Q_{J} \Pi _{\mathcal H_0}^{\perp }h,
\psi ^{J} \bigr\rangle _{L^{2}(X)}'
\bigl(G_{b}^{-1/2}S\bigr)^{-}_{l} \bigr
\rrVert ^{2}\times n^{-1}s_{J}^{-2} \zeta
_{J}^{2}(\log J)\times n^{-1}\zeta
_{J}^{2}
\\
\lesssim{}& n^{-1} \bigl\llVert \bigl\langle Q_{J} \Pi
_{\mathcal H_0}^{\perp }h, \psi ^{J}\bigr\rangle
_{L^{2}(X)}'\bigl(G_{b}^{-1/2}S
\bigr)^{-}_{l} \bigr\rrVert ^{2}
\end{align*}
wpa1 uniformly for $h\in \mathcal H$. Next, we consider the term
$A_{2}$ of \eqref{T1:dec}. Following the upper bound of $A_{12}$, we obtain
wpa1 uniformly for $h\in \mathcal H$:
\begin{align*}
& \bigl\llvert \Evtex \bigl[\Pi _{\mathcal H_0}^{\perp }h(X)b^{K}(W)
\bigr]'A' G (\widehat A-A) \Evtex \bigl[\bigl(h-\Pi
_{\mathcal H_0}h-\Pi _{J}\Pi _{\mathcal H_0}^{\perp }h\bigr) (X)
b^{K}(W)\bigr] \bigr\rrvert ^{2}
\\
&\quad \leq \bigl\llVert \bigl\langle Q_{J} \Pi _{\mathcal H_0}^{\perp }h,
\psi ^{J} \bigr\rangle _{L^{2}(X)}'
\bigl(G_{b}^{-1/2}S\bigr)^{-}_{l} \bigr
\rrVert ^{2} \bigl\llVert G_{b}^{-1/2} S \bigl(\bigl(
\widehat G_{b}^{-1/2}\widehat S\bigr)_{l}^{-}
\widehat G_{b}^{-1/2}G_{b}^{1/2}-
\bigl(G_{b}^{-1/2}S\bigr)^{-}_{l} \bigr)
\bigr\rrVert ^{2}
\\
&\qquad {} \times \bigl\llVert \bigl\langle T\bigl(\Pi _{\mathcal H_0}^{\perp }h-
\Pi _{J} \Pi _{\mathcal H_0}^{\perp }h\bigr), \widetilde{b}^{K}\bigr\rangle _{L^{2}(W)} \bigr\rrVert ^{2}
\\
&\quad \lesssim \bigl\llVert \bigl\langle Q_{J} \Pi _{\mathcal H_0}^{\perp }h,
\psi ^{J} \bigr\rangle _{L^{2}(X)}'
\bigl(G_{b}^{-1/2}S\bigr)^{-}_{l} \bigr
\rrVert ^{2} \bigl\llVert \Pi _{K}T\bigl(\Pi
_{
\mathcal H_0}^{\perp }h-\Pi _{J}\Pi _{\mathcal H_0}^{\perp }h
\bigr) \bigr\rrVert _{L^{2}(W)}^{2} \times n^{-1}s_{J}^{-2}
\zeta _{J}^{2}(\log J)
\\
&\quad \lesssim n^{-1} \bigl\llVert \bigl\langle Q_{J} \Pi
_{\mathcal H_0}^{\perp }h, \psi ^{J}\bigr\rangle
_{L^{2}(X)}'\bigl(G_{b}^{-1/2}S
\bigr)^{-}_{l} \bigr\rrVert ^{2},
\end{align*}
using that
$s_{J}^{-2}\|\Pi _{K} T(\Pi _{\mathcal H_0}^{\perp }h-\Pi _{J}\Pi _{
\mathcal H_0}^{\perp }h)\|_{L^{2}(W)}^{2}\lesssim \|\Pi _{\mathcal H_0}^{
\perp }h-\Pi _{J}\Pi _{\mathcal H_0}^{\perp }h\|_{L^{2}(X)}^{2}$ by Assumption~2(iv) and
$\zeta _{J}^{2}(\log J) \|h-\Pi _{J}h\|_{L^{2}(X)}^{2}=O(1)$ by Assumption~2(iii). Finally, we obtain
$|T_{1}|\leq |A_{1}|+|A_{2}|\lesssim n^{-1/2}\|\langle Q_{J} \Pi _{
\mathcal H_0}^{\perp }h,\psi ^{J}\rangle _{L^{2}(X)}' (G_{b}^{-1/2}S)^{-}_{l}
\|$ wpa1 uniformly for $h\in \mathcal H$.

We next consider the term $T_{2}$ using the decomposition
\begin{align*}
T_{2}\leq{}& 2\Evtex \bigl[\Pi _{J}\Pi _{\mathcal H_0}^{\perp }h(X)b^{K}(W)
\bigr]'( \widehat A-A)' G (\widehat A-A)\Evtex \bigl[\Pi
_{J}\Pi _{\mathcal H_0}^{
\perp }h(X) b^{K}(W)\bigr]
\\
& {}+2\Evtex \bigl[\Pi _{J}^{\perp}\Pi _{\mathcal H_0}^{\perp }h(X)b^{K}(W)
\bigr]'( \widehat A-A)' G (\widehat A-A)\Evtex \bigl[\Pi
_{J}^{\perp}\Pi _{\mathcal H_0}^{
\perp }h(X)
b^{K}(W)\bigr]
\\
=:{}&2T_{21}+2T_{22},
\end{align*}
where $\Pi _{J}^{\perp}=\text{id}-\Pi _{J}$ is the projection. We first
bound $T_{21}$ using Assumption~2(ii):
\begin{align*}
T_{21}\leq{}& \biggl\llvert \bigl\langle \Pi _{J} \Pi
_{\mathcal H_0}^{\perp }h , \psi ^{J}\bigr\rangle
_{L^{2}(X)}' \bigl(\bigl(\widehat G_{b}^{-1/2}
\widehat S\bigr)_{l}^{-} \widehat G_{b}^{-1/2}S-
I_{J} \bigr)'\bigl(\widehat G_{b}^{-1/2}
\widehat S\bigr)_{l}^{-} \widehat G_{b}^{-1/2}
\\
& {}\times \biggl(\frac{1}{n}\sum_{i} \Pi
_{J}\Pi _{
\mathcal H_0}^{\perp }h(X_{i})\widetilde{b}^{K}(W_{i})-\Evtex \bigl[\Pi _{J} \Pi
_{\mathcal H_0}^{\perp }h(X)\widetilde{b}^{K}(W)\bigr] \biggr)
\biggr\rrvert
\\
\leq{}& \bigl\llVert \bigl\langle \Pi _{J} \Pi _{\mathcal H_0}^{\perp }h,
\psi ^{J} \bigr\rangle _{L^{2}(X)} \bigr\rrVert \llVert S-\widehat S
\rrVert \bigl\llVert \bigl(\widehat G_{b}^{-1/2} \widehat S
\bigr)_{l}^{-}\widehat G_{b}^{-1/2} \bigr
\rrVert ^{2}
\\
&{} \times \biggl\llVert \frac{1}{n}\sum_{i} \Pi
_{J}\Pi _{
\mathcal H_0}^{\perp }h(X_{i})\widetilde{b}^{K}(W_{i})-\Evtex \bigl[\Pi _{J} \Pi
_{\mathcal H_0}^{\perp }h(X)\widetilde{b}^{K}(W)\bigr] \biggr
\rrVert
\\
\lesssim {}&\bigl\llVert \Pi _{J} (h-\Pi _{\mathcal H_0}h) \bigr\rrVert
_{L^{2}(X)} n^{-1/2}s_{J}^{-2} \zeta
_{J}\sqrt{\log J}\times n^{-1/2}s_{J}^{-1}
\zeta _{J}
\\
\lesssim{}& n^{-1/2}s_{J}^{-1} \bigl
\llVert \Pi _{J} (h-\Pi _{\mathcal H_0}h) \bigr\rrVert _{L^{2}(X)}
\end{align*}
wpa1 uniformly for $h\in \mathcal H$. For $T_{22}$, we note that uniformly
in $h\in \mathcal H$,
$\|\Evtex [\Pi _{J}^{\perp}\Pi _{\mathcal H_0}^{\perp }h(X)\*\widetilde{b}^{K}(W)]
\|=\|\Pi _{K} T(\Pi _{J} \Pi _{\mathcal H_0}^{\perp }h-\Pi _{
\mathcal H_0}^{\perp }h)\|_{L^{2}(W)}\lesssim s_{J} J^{-p/d_{x}}$ by Assumption~2(iv). Thus, following the upper bound derivations of
$T_{21}$, we obtain $T_{22}\lesssim n^{-1/2}s_{J}^{-1} J^{-p/d_{x}}$ wpa1
uniformly for $h\in \mathcal H$.
\end{proof}

\begin{lemma}%
\label{lemma:bound:var}
Under Assumption~2(i), it holds for
$\widetilde h\in \{h,\Pi _{\mathcal H_0}h\}$ that
\begin{align*}
\sup_{J\in \mathcal I_{n}}\sup_{h\in \mathcal H}\lambda _{\max}
\bigl(\Evtex _{h} \bigl[\bigl(Y-\widetilde h(X)\bigr)^{2}
\widetilde{b}^{K(J)}(W) \widetilde{b}^{K(J)}(W)' \bigr]
\bigr)\leq \overline\sigma ^{2}<\infty .
\end{align*}
\end{lemma}
\begin{proof}
We have for any $\gamma \in \mathbb R^{K}$ where $K=K(J)$ that
\begin{align*}
&\gamma '\Evtex _{h} \bigl[\bigl(Y-\widetilde h(X)
\bigr)^{2}\widetilde{b}^{K}(W) \widetilde{b}^{K}(W)'
\bigr]\gamma
\\
&\quad \leq \Evtex \bigl[\Evtex _{h}\bigl[\bigl(Y-\widetilde
h(X)\bigr)^{2}|W\bigr] \bigl(\gamma '\widetilde{b}^{K}(W) \bigr)^{2} \bigr]
\\
&\quad \leq \overline\sigma ^{2}\Evtex \bigl[ \bigl(\gamma '
\widetilde{b}^{K}(W) \bigr)^{2} \bigr]= \overline\sigma
^{2}\gamma 'G_{b}^{-1/2}\Evtex
\bigl[b^{K}(W)b^{K}(W)' \bigr]G_{b}^{-1/2}
\gamma = \overline\sigma ^{2} \llVert \gamma \rrVert ^{2}
\end{align*}
uniformly for $h\in \mathcal H$ and $J\in \mathcal I_{n}$, where the second
inequality is due to Assumption~2(i).
\end{proof}

\begin{proof}[Proof of Theorem~B.1]
From the definition of $Q_{J}$ given in (B.1), we infer
\begin{align*}
\bigl\llVert Q_{J}(h-\Pi _{\mathcal H_0}h) \bigr\rrVert
_{L^{2}(X)}^{2} = \bigl\llVert A\Evtex _{h}\bigl[
\bigl(Y- \Pi _{\mathcal H_0}h(X)\bigr) b^{K}(W)\bigr] \bigr\rrVert
^{2}= \bigl\llVert \Evtex _{h}\bigl[U^{J}\bigr]
\bigr\rrVert ^{2}
\end{align*}
using the notation
$U_{i}^{J}=(Y_{i}-\Pi _{\mathcal H_0}h(X_{i})) A b^{K}(W_{i})$. The definition
of $\widehat D_{J}$ implies
\begin{flalign}
&\widehat{D}_{J}(\Pi _{\mathcal H_0}h)- \bigl\llVert Q_{J}(h-
\Pi _{\mathcal H_0}h) \bigr\rrVert _{L^{2}(X)}^{2}&\nonumber
\\
&\,\, =
\frac{1}{n(n-1)}\sum_{j=1}^{J}\sum
_{i\neq i'} \bigl( U_{ij}U_{i'j}- \Evtex
_{h}[U_{1j}]^{2} \bigr) \label{eq1:proof:upper} 
\\
&\,\,+\frac{1}{n(n-1)}\sum_{i\neq i'} \bigl(Y_{i}-
\Pi _{\mathcal H_0}h(X_{i}) \bigr) \bigl(Y_{i'}-\Pi
_{\mathcal H_0}h(X_{i'}) \bigr) b^{K}(W_{i})'
\bigl(A' A-\widehat A'\widehat A \bigr)b^{K}(W_{i'}).
\label{eq2:proof:upper} 
\end{flalign}%
Consider the summand in \eqref{eq1:proof:upper}; we observe
\begin{equation*}
\Biggl\llvert \sum_{j=1}^{J}\sum
_{i\neq i'} \bigl( U_{ij}U_{i'j}- \Evtex
_{h}[U_{1j}]^{2} \bigr) \Biggr\rrvert
^{2} =\sum_{j,j'=1}^{J}\sum
_{i\neq i'}\sum_{i''\neq i'''} \bigl(
U_{ij}U_{i'j}- \Evtex _{h}[U_{1j}]^{2}
\bigr) \bigl( U_{i''j'}U_{i'''j'}- \Evtex _{h}[U_{1j'}]^{2}
\bigr).
\end{equation*}
We distinguish three different cases. First: $i$, $i'$, $i''$, $i'''$ are all different;
second: either $i=i''$ or $i'=i'''$; or third: $i=i'$ and $i'=i'''$. We
thus calculate for each $j, j'\geq 1$ that
\begin{align*}
&\sum_{i\neq i'}\sum_{i''\neq i'''}
\bigl( U_{ij}U_{i'j}- \Evtex _{h}[U_{1j}]^{2}
\bigr) \bigl( U_{i''j'}U_{i'''j'}- \Evtex _{h}[U_{1j'}]^{2}
\bigr)
\\
&\quad =\sum_{i, i',i'', i'''\text{all different}} \bigl( U_{ij}U_{i'j}-
\Evtex _{h}[U_{1j}]^{2} \bigr) \bigl(
U_{i''j'}U_{i'''j'}- \Evtex _{h}[U_{1j'}]^{2}
\bigr)
\\
& \qquad {}+2\sum_{i\neq i'\neq i''} \bigl( U_{ij}U_{i'j}-
\Evtex _{h}[U_{1j}]^{2} \bigr) \bigl(
U_{i''j'}U_{i'j'}- \Evtex _{h}[U_{1j'}]^{2}
\bigr)
\\
&\qquad {} +\sum_{i\neq i'} \bigl( U_{ij}U_{i'j}-
\Evtex _{h}[U_{1j}]^{2} \bigr) \bigl(
U_{ij'}U_{i'j'}- \Evtex _{h}[U_{1j'}]^{2}
\bigr).
\end{align*}
The expectation of the first term on the right-hand side vanishes due to
independent observations and thus, we have
\begin{align*}
&\Evtex _{h} \Biggl\llvert \sum_{j=1}^{J}
\sum_{i\neq i'} \bigl( U_{ij}U_{i'j}-
\Evtex _{h}[U_{1j}]^{2} \bigr) \Biggr\rrvert
^{2}
\\
&\quad =2n(n-1) (n-2) \underbrace{\sum_{j,j'=1}^{J}
\Evtex _{h} \bigl[ \bigl( U_{1j}U_{2j}- \Evtex
_{h}[U_{1j}]^{2} \bigr) \bigl(
U_{3j'}U_{2j'}- \Evtex _{h}[U_{1j'}]^{2}
\bigr) \bigr]}_{I}
\\
&\qquad {}+n(n-1) \underbrace{\sum_{j,j'=1}^{J}\Evtex
_{h} \bigl[ \bigl( U_{1j}U_{2j}- \Evtex
_{h}[U_{1j}]^{2} \bigr) \bigl(
U_{1j'}U_{2j'}- \Evtex _{h}[U_{1j'}]^{2}
\bigr) \bigr]}_{\mathit{II}}.
\end{align*}
Now using $\|(G_{b}^{-1/2}SG^{-1/2})_{l}^{-}\|=s_{J}^{-1}$ together with
the notation $\widetilde\psi ^{J}=G^{-1/2}\psi ^{J}$, we obtain
%
\begin{align}
&\bigl\llVert \bigl\langle Q_{J}(h-\Pi _{\mathcal H_0}h),\psi
^{J}\bigr\rangle _{L^{2}(X)}'\bigl(G_{b}^{-1/2}S
\bigr)_{l}^{-} \bigr\rrVert\nonumber
 = \bigl\llVert \bigl\langle
Q_{J}(h-\Pi _{\mathcal H_0}h),\widetilde \psi ^{J} \bigr
\rangle _{L^{2}(X)}'\bigl(G_{b}^{-1/2}SG^{-1/2}
\bigr)_{l}^{-} \bigr\rrVert
\nonumber
\\
&\quad \leq s_{J}^{-1} \bigl\llVert \bigl\langle Q_{J}(h-
\Pi _{\mathcal H_0}h), \widetilde\psi ^{J}\bigr\rangle _{L^{2}(X)}
\bigr\rrVert
 \lesssim s_{J}^{-1} \bigl( \llVert h-\Pi
_{\mathcal H_0}h \rrVert _{L^{2}(X)}+J^{-p/d_{x}} \bigr) \label{upper:bound:typeII} ,
\end{align}
where the last equation is due to Lemma~B.1(i). Consequently,
we bound the term $I$ by
\begin{align*}
I &=\sum_{j,j'=1}^{J}\Evtex
_{h}[U_{1j}]\Evtex _{h}[U_{1j'}]\Cov
_{h}(U_{1j},U_{1j'})
\\
&= \Evtex _{h}
\bigl[U_{1}^{J}\bigr]'\Cov _{h}
\bigl(U_{1}^{J},U_{1}^{J}\bigr)\Evtex
_{h}\bigl[U_{1}^{J}\bigr]
\\
&\leq \lambda _{\max} \bigl(\Var _{h}\bigl(\bigl(Y-\Pi
_{\mathcal H_0}h(X)\bigr) \widetilde{b}^{K}(W)\bigr) \bigr) \bigl
\llVert \bigl(G_{b}^{-1/2} SG ^{-1/2}
\bigr)_{l}^{-}\Evtex _{h}\bigl[U_{1}^{J}
\bigr] \bigr\rrVert ^{2}
\\
&\leq \overline\sigma ^{2} \bigl\llVert \bigl(\bigl(G_{b}^{-1/2}
S\bigr)_{l}^{-}\Evtex _{h}\bigl[\bigl(Y- \Pi
_{\mathcal H_0}h(X)\bigr) \widetilde{b}^{K}(W)\bigr]
\bigr)'G \bigl(G_{b}^{-1/2} S\bigr)_{l}^{-}
\bigr\rrVert ^{2}
\\
&=\overline\sigma ^{2} \bigl\llVert \bigl\langle Q_{J} (h-
\Pi _{\mathcal H_0}h), \psi ^{J}\bigr\rangle _{L^{2}(X)}'
\bigl(G_{b}^{-1/2}S\bigr)_{l}^{-} \bigr
\rrVert ^{2}
\\
&\lesssim s_{J}^{-2} \bigl( \llVert h-
\Pi _{\mathcal H_0}h \rrVert _{L^{2}(X)}^{2}+J^{-2p/d_{x}}
\bigr),
\end{align*}
using
$U_{i}^{J}=(Y_{i}-\Pi _{\mathcal H_0}h(X_{i}))(G_{b}^{-1/2} SG^{-1/2})_{l}^{-}
\widetilde{b}^{K}(W_{i})$ and Lemma~\ref{lemma:bound:var}. For term
$\mathit{II}$, we observe
\begin{align*}
\mathit{II}= \sum_{j,j'=1}^{J}\Evtex
_{h}[U_{1j}U_{1j'}]^{2}- \Biggl(\sum
_{j=1}^{J} \Evtex _{h}[U_{1j}]^{2}
\Biggr)^{2}\leq \sum_{j,j'=1}^{J}
\Evtex _{h}[U_{1j}U_{1j'}]^{2}=V_{J}^{2}.
\end{align*}
Thus, the upper bounds derived for the terms $I$ and $\mathit{II}$ imply for all
$n\geq 2$:
%
\begin{equation}
\label{bound:var:V} \Evtex _{h} \Biggl\llvert \frac{1}{n(n-1)}\sum
_{j=1}^{J}\sum_{i\neq i'}
\bigl( U_{ij}U_{i'j}- \Evtex _{h}[U_{1j}]^{2}
\bigr) \Biggr\rrvert ^{2}\lesssim \frac{ \llVert h-\Pi _{\mathcal H_0}h
\rrVert _{L^{2}(X)}^{2}+J^{-2p/d_{x}}}{ns_{J}^{2}}+
\frac{V_{J}^{2}}{n^{2}}. 
\end{equation}
Thus, equality \eqref{eq2:proof:upper} implies the result by employing
Lemma~B.2 and Lemma~\ref{lemma:est:matrices}.
\end{proof}

\begin{proof}[Proof of Lemma~A.1]
By Lemma~\ref{lemma:est:matrices} and the decomposition (\ref{eq1:proof:upper})--(\ref{eq2:proof:upper}),
we obtain
\begin{align*}
\mathrm{P}_{h_{0}} \biggl(\frac{n\widehat{D}_{J}(h_{0})}{V_{J}}>
\eta _{J}( \alpha ) \biggr)&= \mathrm{P}_{h_{0}} \Biggl(
\frac{1}{V_{J}(n-1)}\sum_{j=1}^{J}
\sum_{i\neq i'} U_{ij}U_{i'j}>\eta
_{J}(\alpha ) \Biggr) +o(1).
\end{align*}
Using the martingale central limit theorem (see, e.g.,
\citet [Lemma A.3]{breunig2020}), we obtain
\begin{align*}
\mathrm{P}_{h_{0}} \Biggl(\frac{1}{ \sqrt 2V_{J}(n-1)}\sum
_{j=1}^{J} \sum_{i\neq i'}
U_{ij}U_{i'j}>z_{1-\alpha} \Biggr)=\alpha +o(1),
\end{align*}
where $z_{1-\alpha}$ denotes the $(1-\alpha )$-quantile of the standard
normal distribution. Further, Lemma~B.4(i) implies
$V_{J}/\widehat{V}_{J}=1$ wpa1 uniformly for $h\in \mathcal H$, and since
$\eta _{J}(\alpha )/\sqrt 2=\frac{q(\alpha , J)-J}{\sqrt{2J}}$ converges
to $z_{1-\alpha}$ as $J$ tends to infinity, the result follows.
\end{proof}

\begin{proof}[Proof of Lemma~B.1]
Proof of (i): Using the notation
$\widetilde{b}^{K}(\cdot ):=G_{b}^{-1/2} b^{K}(\cdot )$, we observe for
all $h\in \mathcal H$ that
\begin{align*}
&\bigl\llVert Q_{J} (h-\Pi _{\mathcal H_0}h) \bigr\rrVert
_{L^{2}(X)}
\\
&\quad = \bigl\llVert \bigl(G_{b}^{-1/2}SG^{-1/2}
\bigr)^{-}_{l} \Evtex \bigl[\widetilde{b}^{K}(W)
(h- \Pi _{\mathcal H_0}h) (X)\bigr] \bigr\rrVert
\\
&\quad \leq \bigl\llVert \bigl(G_{b}^{-1/2}SG^{-1/2}
\bigr)^{-}_{l}\Evtex \bigl[\widetilde{b}^{K}(W) (
\Pi _{J} h- \Pi _{J} \Pi _{\mathcal H_0}h) (X)\bigr] \bigr
\rrVert
\\
& \qquad {}+ \bigl\llVert \bigl(G_{b}^{-1/2}SG^{-1/2}
\bigr)^{-}_{l}\Evtex \bigl[\widetilde{b}^{K}(W)
\bigl((h- \Pi _{\mathcal H_0}h) (X)-(\Pi _{J} h- \Pi _{J} \Pi
_{\mathcal H_0}h) (X)\bigr)\bigr] \bigr\rrVert
\\
&\quad \leq \llVert \Pi _{J} h- \Pi _{J} \Pi _{\mathcal H_0}h
\rrVert _{L^{2}(X)}+s_{J}^{-1} \bigl\llVert \Pi
_{K} T\bigl((h-\Pi _{\mathcal H_0}h)-(\Pi _{J} h- \Pi
_{J} \Pi _{
\mathcal H_0}h)\bigr) \bigr\rrVert _{L^{2}(W)}
\\
&\quad \leq \llVert \Pi _{J} h- \Pi _{J} \Pi _{\mathcal H_0}h
\rrVert _{L^{2}(X)}+O \bigl(J^{-p/d_{x}} \bigr)
\end{align*}
by Assumption~2(iv).

Proof of (ii): We observe
$\|Q_{J}h-h\|_{L^{2}(X)}\leq \|Q_{J}(h-\Pi _{J} h)\|_{L^{2}(X)}+\|
\Pi _{J} h-h\|_{L^{2}(X)}$. The result thus follows by replacing
$ \Pi _{\mathcal H_0}h$ with $ \Pi _{J} h$ in the derivation of (i).
\end{proof}

\begin{proof}[Proof of Lemma~B.2]
For any $J\times J$ matrix $M$, it holds
$\|M\|_{F}\leq \sqrt J \|M\|$ and hence
\begin{align*}
V_{J}^{2}&= \bigl\llVert \bigl(G_{b}^{-1/2}SG^{-1/2}
\bigr)^{-}_{l}\Evtex _{h} \bigl[ \bigl(Y-h(X)
\bigr)^{2}\widetilde{b}^{K}(W)\widetilde{b}^{K}(W)'
\bigr] \bigl(G_{b}^{-1/2}SG^{-1/2}
\bigr)^{-}_{l} \bigr\rrVert _{F}^{2}
\\
&\leq J \bigl\llVert \bigl(G_{b}^{-1/2}SG^{-1/2}
\bigr)^{-}_{l} \bigr\rrVert ^{4} \bigl\llVert
\Evtex _{h} \bigl[ \bigl(Y-h(X)\bigr)^{2}\widetilde{b}^{K}(W)\widetilde{b}^{K}(W)' \bigr] \bigr
\rrVert ^{2}.
\end{align*}
The result now follows from
$\|(G_{b}^{-1/2}SG^{-1/2})^{-}_{l}\|=s_{J}^{-1}$ and Lemma~\ref{lemma:bound:var}.
\end{proof}

\begin{proof}[Proof of Lemma~B.3]
In the following, let $e_{j}$ be the unit vector with $1$ at the $j$th
position. Introduce a unitary matrix $Q$ such that, by Schur decomposition,
$Q' A G_{b} A'Q=\operatorname{diag}(s_{1}^{-2},\dots ,s_{J}^{-2})$. We make use
of the notation
$\widetilde U_{i}^{J}=(Y_{i}-h(X_{i}))Q'Ab^{K}(W_{i})$. Now, since the
Frobenius norm is invariant under unitary matrix multiplication, we have
\begin{align*}
V_{J}^{2}&=\sum_{j,j'=1}^{J}
\Evtex _{h}[\widetilde U_{1j}\widetilde U_{1j'}]^{2}
\geq \sum_{j=1}^{J}\Evtex _{h}
\bigl[\widetilde U_{1j}^{2}\bigr]^{2} = \sum
_{j=1}^{J} \bigl(\Evtex _{h} \bigl
\llvert \bigl(Y-h(X)\bigr)e_{j}' Q'A
b^{K}(W) \bigr\rrvert ^{2} \bigr)^{2}.
\end{align*}
Consequently, using the lower bound
$\inf_{w\in \mathcal W}\inf_{h\in \mathcal H}\Evtex _{h}[(Y-h(X))^{2}|W=w]
\geq \underline\sigma ^{2}$ by Assumption~1(i), we obtain uniformly
for $h\in \mathcal H$:
\begin{align*}
V_{J}^{2}&\geq \underline\sigma ^{4} \sum
_{j=1}^{J} \bigl(\Evtex \bigl[e_{j}'
Q' A b^{K}(W)b^{K}(W)'A'Qe_{j}
\bigr] \bigr)^{2}
\\
&= \underline\sigma ^{4} \sum
_{j=1}^{J} \bigl(e_{j}'
Q'A G_{b} A'Qe_{j}
\bigr)^{2}
\\
&= \underline\sigma ^{4} \sum_{j=1}^{J}
\bigl(e_{j}' \operatorname{diag}\bigl(s_{1}^{-2},
\dots ,s_{J}^{-2}\bigr) e_{j}
\bigr)^{2}\geq \underline\sigma ^{4} \sum
_{j=1}^{J} s_{j}^{-4},
\end{align*}
which proves the result.
\end{proof}

Recall the definition
$\mathcal C_{h}=\max_{\mathsf e\in \mathcal S^{K^\circ }}\int _{0}^{1}
  (1+\log N_{[]}  (\epsilon \|F_{h,\mathsf e}\|_{L^{2}(Z)},
\mathcal F_{h,\mathsf e}, L^{2}(Z)  )  )^{1/2}d\epsilon $.
%
\begin{lemma}%
\label{lemma:quad:fctl:est_h}
Let Assumptions 1(ii)--(iii), 2(i),
4(i)(iii), and 5(ii) hold. Then, for
$J=J^\circ $, we have wpa1 uniformly for
$h\in \mathcal H_{1}(\delta ^{\circ}\text{\textsf r}_{n})$:
\begin{align*}
&\biggl\llvert \frac{1}{n(n-1)}\sum_{i\neq i'}U_{i}
\bigl(\widehat h_{J}^{\text{\textsc r}}\bigr) U_{i'}\bigl(\widehat
h_{J}^{\text{\textsc r}}\bigr) a_{J,i i'}-\Evtex _{h}
\bigl[U_{i}\bigl( \widehat h_{J}^{\text{\textsc r}}\bigr)
U_{i'}\bigl(\widehat h_{J}^{\text{\textsc r}}\bigr)a_{J,i i'}
\bigr] \biggr\rrvert
\\
&\quad \lesssim n^{-1/2}s_{J}^{-1} \mathcal
C_{h} \bigl( \llVert h-\mathcal H_{0} \rrVert
_{L^{2}(X)}+J^{-p/d_{x}} \bigr) +n^{-1} s_{J}^{-2}
\sqrt{J},
\end{align*}
where $U_{i}(\phi )=Y_{i}-\phi (X_{i})$ and
$a_{J,i i'}=b^{K}(W_{i})'A' Ab^{K}(W_{i'})$.
\end{lemma}
\begin{proof}
For simplicity of notation, we write $J$ instead of $J^\circ $ throughout
the proof. We observe for all
$h\in \mathcal H_{1}(\delta ^{\circ}\text{\textsf r}_{n})$ that
\begin{align*}
&\frac{1}{n(n-1)}\sum_{i\neq i'}U_{i}\bigl(
\widehat h_{J}^{\text{\textsc r}}\bigr) U_{i'}\bigl( \widehat
h_{J}^{\text{\textsc r}}\bigr) a_{J,i i'}-\Evtex _{h}
\bigl[U_{i}\bigl(\widehat h_{J}^{
\text{\textsc r}}\bigr)
U_{i'}\bigl(\widehat h_{J}^{\text{\textsc r}}\bigr)a_{J,i i'}
\bigr]
\\
&\quad =\frac{1}{n(n-1)}\sum_{i\neq i'}U_{i}(\Pi
_{\mathcal H_{0}} h) U_{i'}( \Pi _{\mathcal H_{0}} h) a_{J,i i'}-
\Evtex _{h} \bigl[U_{i}(\Pi _{
\mathcal H_{0}} h)
U_{i'}(\Pi _{\mathcal H_{0}} h)a_{J,i i'} \bigr]
\\
& \qquad {}+\frac{2}{n(n-1)}\sum_{i\neq i'}U_{i}(\Pi
_{\mathcal H_{0}} h) \bigl( \Pi _{\mathcal H_{0}} h-\widehat h_{J}^{\text{\textsc r}}
\bigr) (X_{i'}) a_{J,i i'}
\\
&{}\qquad - \Evtex _{h}
\bigl[U_{i}(\Pi _{\mathcal H_{0}} h) \bigl(\Pi _{\mathcal H_{0}} h-
\widehat h_{J}^{\text{\textsc r}}\bigr) (X_{i'}) a_{J,i i'}
\bigr]
\\
& \qquad {}+\frac{1}{n(n-1)}\sum_{i\neq i'} \bigl(\Pi
_{\mathcal H_{0}} h- \widehat h_{J}^{\text{\textsc r}}\bigr)
(X_{i}) \bigl(\Pi _{\mathcal H_{0}} h- \widehat h_{J}^{\text{\textsc r}}
\bigr) (X_{i'}) a_{J,i i'}
\\
&\qquad {} -\Evtex _{h} \bigl[ \bigl(
\Pi _{
\mathcal H_{0}} h-\widehat h_{J}^{\text{\textsc r}}\bigr)
(X_{i}) \bigl(\Pi _{
\mathcal H_{0}} h-\widehat h_{J}^{\text{\textsc r}}
\bigr) (X_{i'})a_{J,i i'} \bigr]
\\
&\quad =:T_{1}+2 T_{2}+T_{3}.
\end{align*}%
From the proof of Theorem~B.1, we conclude
$\sup_{h\in \mathcal H_{1}(\delta ^{\circ}\text{\textsf r}_{n})}\Evtex _{h}|T_{1}|
\lesssim n^{-1} s_{J}^{-2}\sqrt J$. Consider $T_{2}$. Below, we let
$a_{i}^{J}=Ab^{K}(W_{i})=(G_{b}^{-1/2} S G^{-1/2})^{-}_{\ell }
\widetilde{b}^{K}(W_{i})$. By Assumption~5(ii),
$\sup_{h\in \mathcal H_{1}(\delta ^{\circ}\text{\textsf r}_{n})}\mathrm{P}_{h}
 (\zeta _{J} \mathcal C_{h}\|\widehat h_{J}^{\text{\textsc r}}-\Pi _{
\mathcal H_0}h\|_{L^{2}(X)}> C )\to 0$ and consequently may assume that
$\widehat h_{J}^{\text{\textsc r}}\in \mathcal H_{0,J}(h):=\{\|\phi -\Pi _{
\mathcal H_0}h\|_{L^{2}(X)}\leq [\zeta _{J} \mathcal C_{h}]^{-1} :
\phi \in \mathcal H_{0,J}\}$. We have for all
$h\in \mathcal H_{1}(\delta ^{\circ}\text{\textsf r}_{n})$ that the absolute value
of $T_{2}$ is bounded by
\begin{align*}
&\sup_{\phi \in \mathcal H_{0,J}(h)} \biggl\llvert \frac{1}{n(n-1)}\sum
_{i
\neq i'} \bigl(U_{i}(\Pi _{\mathcal H_{0}}
h)a_{i}^{J} - \Evtex _{h}\bigl[U(\Pi
_{
\mathcal H_{0}} h)a^{J}\bigr] \bigr)'
\\
&\qquad {}\times \bigl((\Pi
_{\mathcal H_{0}} h-\phi ) (X_{i'}) a_{i'}^{J}-\Evtex
\bigl[ (\Pi _{\mathcal H_{0}} h-\phi ) (X) a^{J} \bigr] \bigr) \biggr
\rrvert &
\\
&\qquad  {}+\bigg| \frac{1}{n}\sum_{i} \bigl(U_{i}(\Pi
_{\mathcal H_{0}} h)a_{i}^{J} - \Evtex _{h}\bigl[U(
\Pi _{\mathcal H_{0}} h)a^{J}\bigr] \bigr)' \Evtex \bigl[
\bigl(\Pi _{
\mathcal H_{0}} h-\widehat h_{J}^{\text{\textsc r}}\bigr) (X)
a^{J} \bigr] \big) \biggl\llvert
\\
&\qquad  {}+\sup_{\phi \in \mathcal H_{0,J}(h)}\bigg| \frac{1}{n}\sum_{i}
\bigl(( \Pi _{\mathcal H_{0}} h-\phi ) (X_{i}) a_{i}^{J}-
\Evtex \bigl[ (\Pi _{
\mathcal H_{0}} h-\phi ) (X) a^{J} \bigr]
\bigr)'\Evtex _{h}\bigl[U(\Pi _{
\mathcal H_{0}}
h)a^{J}\bigr] \biggr\rrvert
\\
&\quad  =:T_{21}+T_{22}+T_{23}.
\end{align*}%
Below, we let $a_{k,i}=b^{K}(W_{i})'A' A G_{b}^{1/2}e_{k}$. Note that
$\Evtex \|a_{k,i}\|^{2}\leq \|(G_{b}^{-1/2} S G^{-1/2})_{\ell}^{-}\|^{4}=s_{J}^{-4}$
for all $k=1,\dots ,K$. We obtain uniformly for
$h\in \mathcal H_{1}(\delta ^{\circ}\text{\textsf r}_{n})$ by
\citet [Theorem~2.14.2]{Vaart2000} that
\begin{align*}
\Evtex _{h} T_{21}\leq{}& \sum_{k=1}^{K}
\Evtex _{h} \biggl\llvert \frac{1}{n}\sum_{i}U_{i}(
\Pi _{\mathcal H_{0}} h)a_{k,i} - \Evtex _{h}\bigl[U(\Pi
_{\mathcal H_{0}} h)a_{k}\bigr] \biggr\rrvert &
\\
& {}\times \Evtex _{h}\sup_{\phi \in \mathcal H_{0,J}(h)}\bigg| \frac{1}{n-1}\sum
_{i'}(\Pi _{\mathcal H_{0}} h-\phi ) (X_{i'})
\widetilde{b}_{k}(W_{i'})
\\
&\qquad {}-\Evtex _{h} \bigl[ (\Pi
_{\mathcal H_{0}} h- \phi ) (X) \widetilde{b}_{k}(W) \bigr]\bigg|
\\
\lesssim{}& \frac{\mathcal C_{h}}{n} \sqrt{\sum
_{k=1}^{K}\Evtex _{h} \bigl[ \bigl\llvert
U_{i}( \Pi _{\mathcal H_{0}} h) \bigr\rrvert ^{2} \bigl\llVert
G_{b}^{-1/2}A'Ab^{K}(W) \bigr\rrVert
^{2} \bigr]}
\\
&{}\times  \sqrt{\sum_{k=1}^{K}
\Evtex _{h}\sup_{\phi \in \mathcal H_{0,J}(h)} \bigl\llvert ( \Pi
_{\mathcal H_{0}} h-\phi ) (X) \widetilde{b}_{k}(W) \bigr\rrvert
^{2}}
\\
\lesssim {}&\frac{\mathcal C_{h}}{n} \overline \sigma
s_{J}^{-2} \sqrt J \zeta _{J} \llVert \Pi
_{\mathcal H_{0}} h-\Phi _{J} \rrVert _{L^{2}(X)}\lesssim
n^{-1} s_{J}^{-2}\sqrt J
\end{align*}%
for some $\Phi _{J}\in \mathcal H_{0,J}(h)$ and using that
$\Evtex _{h}[|U(\Pi _{\mathcal H_{0}} h)|^{2}|W]\leq \overline \sigma ^{2}$
by Assumption~2(i). Further, we evaluate uniformly for
$h\in \mathcal H_{1}(\delta ^{\circ}\text{\textsf r}_{n})$:
\begin{align*}
\Evtex _{h} T_{22} &=\overline \sigma n^{-1/2}\sqrt
{\Evtex |\bigl(a^{J}\bigr)'\Evtex _{h}
\bigl[ \bigl(\Pi _{\mathcal H_{0}} h-\widehat h_{J}^{\text{\textsc r}}\bigr)
(X) a^{J} \bigr]|^{2}}&
\\
&\leq \overline \sigma n^{-1/2} s_{J}^{-2}\sup
_{\phi \in \mathcal H_{0,J}(h)} \|\Pi _{K} T (\Pi _{\mathcal H_{0}} h-\phi )
\|_{L^{2}(W)}
\\
&\lesssim n^{-1/2} s_{J}^{-1} \bigl(
\llVert h-\mathcal H_{0} \rrVert _{L^{2}(X)}+J^{-p/d_{x}}
\bigr),
\end{align*}
where, in the last equation, we used Assumption~2(iv) and
$\|h-\Pi _{\mathcal H_{0}} h\|_{L^{2}(X)}=\|h-\mathcal H_{0}\|_{L^{2}(X)}$.
Consider $T_{23}$. Below, we make use of the relation
$\Evtex [U(\Pi _{\mathcal H_{0}} h)a^{J}]'a_{i}^{J}=\langle Q_{J}(h-\Pi _{
\mathcal H_{0}} h), \psi ^{J}\rangle _{L^{2}(X)}'(G_{b}^{-1/2} S )^{-}_{
\ell }\widetilde{b}^{K}(W_{i})$ and obtain uniformly for
$h\in \mathcal H_{1}(\delta ^{\circ}\text{\textsf r}_{n})$:
\begin{align*}
\Evtex _{h} T_{23}\leq{}& \bigl\llVert \bigl\langle
Q_{J}(h-\Pi _{\mathcal H_{0}} h), \psi ^{J}\bigr\rangle
_{L^{2}(X)}'\bigl(G_{b}^{-1/2} S
\bigr)^{-}_{\ell} \bigr\rrVert
\\
&{}\times \Evtex _{h} \sup_{\mathsf e\in \mathcal S^{K^\circ -1}}\sup
_{
\phi \in \mathcal H_{0,J}(h)} \biggl\llvert \frac{1}{n}\sum
_{i}(\Pi _{
\mathcal H_{0}} h-\phi ) (X_{i})
\widetilde{b}^{K}(W_{i})'\mathsf e
\\
&{}-\Evtex \bigl[
(\Pi _{\mathcal H_{0}} h-\phi ) (X) \widetilde{b}^{K}(W)'
\mathsf e \bigr] \biggr\rrvert
\\
\lesssim{}& \bigl\llVert \bigl\langle Q_{J}(h-\Pi _{\mathcal H_{0}} h),
\psi ^{J} \bigr\rangle _{L^{2}(X)}'
\bigl(G_{b}^{-1/2} S \bigr)^{-}_{\ell}
\bigr\rrVert \times \mathcal C_{h} n^{-1/2} \zeta
_{J} \llVert \Pi _{\mathcal H_{0}} h-\Phi _{J} \rrVert
_{L^{2}(X)}
\\
\lesssim{}& \mathcal C_{h} n^{-1/2}s_{J}^{-1}
\bigl( \llVert h-\mathcal H_{0} \rrVert _{L^{2}(X)}+J^{-p/d_{x}}
\bigr),
\end{align*}%
where we used that
$\sup_{w}|\widetilde{b}^{K}(w)'\mathsf e|\leq \zeta _{J}$ for all
$\mathsf e\in \mathcal S^{K^\circ }$. Consider $T_{3}$. We have
\begin{align*}
\llvert T_{3} \rrvert\leq{}& \sup_{\phi \in \mathcal H_{0,J}(h)} \biggl
\llvert \frac{1}{n(n-1)}\sum_{i\neq i'} \bigl((\Pi
_{\mathcal H_{0}} h-\phi ) (X_{i}) a_{i}^{J} -
\Evtex \bigl[ (\Pi _{\mathcal H_{0}} h-\phi ) (X)a^{J} \bigr]
\bigr)'
\\
& {}\times \bigl((\Pi _{
\mathcal H_{0}} h-\phi ) (X_{i'})
a_{i'}^{J} -\Evtex \bigl[ (\Pi _{
\mathcal H_{0}} h-\phi )
(X)a^{J} \bigr] \bigr) \biggr\rrvert &
\\
& {}+2\sup_{\phi \in \mathcal H_{0,J}(h)} \biggl\llvert \frac{1}{n}\sum
_{i} \bigl((\Pi _{\mathcal H_{0}} h-\phi ) (X_{i})
a_{i}^{J} -\Evtex \bigl[ (\Pi _{
\mathcal H_{0}} h-\phi )
(X)a^{J} \bigr] \bigr)'
\\
&{}\times \Evtex \bigl[ (\Pi _{\mathcal H_{0}}
h-\phi ) (X)a^{J} \bigr] \biggr\rrvert
\\
=:{}&T_{31}+T_{32}.
\end{align*}%
We evaluate for the first term on the right-hand side that uniformly for
$h\in \mathcal H_{1}(\delta ^{\circ}\text{\textsf r}_{n})$:
\begin{align*}
\Evtex T_{31}&\leq s_{J}^{-2}\sum
_{k=1}^{K} \biggl(\Evtex \sup_{\phi \in
\mathcal H_{0,J}(h)}
\biggl\llvert \frac{1}{n}\sum_{i} (\Pi
_{\mathcal H_{0}} h- \phi ) (X_{i}) \widetilde{b}_{k}(W_{i})
-\Evtex \bigl[ (\Pi _{\mathcal H_{0}} h-\phi ) (X) \widetilde{b}_{k}(W)
\bigr] \biggr\rrvert \biggr)^{2}
\\
&\lesssim \frac{\mathcal C_{h}^{2}}{n s_{J}^{2}}
\Evtex \sup_{\phi \in
\mathcal H_{0,J}(h)} \bigl\llVert (\Pi _{\mathcal H_{0}} h-\phi )
(X) \widetilde{b}^{K}(W) \bigr\rrVert ^{2}
\\
&\lesssim \frac{
\mathcal C_{h}^{2}}{n s_{J}^{2}}\zeta
_{J}^{2} \llVert \Pi _{\mathcal H_{0}} h-\Phi _{J}
\rrVert _{L^{2}(X)}^{2} \lesssim \frac{\sqrt J}{n
s_{J}^{2}},
\end{align*}%
{\spaceskip=0.2em plus 0.05em minus 0.03em for some $\Phi _{J}\in \mathcal H_{0,J}(h)$ and using that
$\mathcal C_{h}^{2}\lesssim \sqrt{J}$. Further, we have}
$\Evtex [(\Pi _{\mathcal H_{0}} h-\phi )(X)a^{J}]'a_{i}^{J}=\langle Q_{J}(
\Pi _{\mathcal H_{0}} h-\phi ), \psi ^{J}\rangle _{L^{2}(X)}'(G_{b}^{-1/2}
S )^{-}_{\ell }\widetilde{b}^{K}(W_{i})$ and thus, following the derivation
of the bound of $T_{23}$, we obtain
\begin{align*}
\Evtex T_{32}\leq{}& \sup_{\phi \in \mathcal H_{0,J}(h)} \bigl\llVert \bigl
\langle Q_{J}( \phi -\Pi _{\mathcal H_{0}} h), \psi ^{J}\bigr
\rangle _{L^{2}(X)}'\bigl(G_{b}^{-1/2} S
\bigr)^{-}_{\ell} \bigr\rrVert
\\
&{} \times \Evtex \sup_{\mathsf e\in \mathcal S^{K^\circ }}\sup_{
\phi \in \mathcal H_{0,J}(h)}
\biggl\llvert \frac{1}{n}\sum_{i}(\phi -\Pi
_{
\mathcal H_{0}} h) (X_{i}) \widetilde{b}^{K}(W_{i})'
\mathsf e-\Evtex \bigl[ ( \phi -\Pi _{\mathcal H_{0}} h) (X) \widetilde{b}^{K}(W)'\mathsf e \bigr] \biggr\rrvert
\\
\lesssim{}& n^{-1/2}s_{J}^{-1} \mathcal
C_{h} \bigl( \llVert h-\Pi _{\mathcal H_{0}} h \rrVert
_{L^{2}(X)}+J^{-p/d_{x}} \bigr)
\end{align*}%
uniformly for $h\in \mathcal H_{1}(\delta ^{\circ}\text{\textsf r}_{n})$, where
the last equation is due to Assumption~5(ii). Finally, the result
follows from an application of Markov's inequality.
\end{proof}

\begin{lemma}%
\label{lemma:est:matrices&h}
Let Assumptions 1(ii)--(iii), 2(i),
4(i)(iii), and 5(ii) hold. Then, for
$J=J^\circ $, we have wpa1 uniformly for
$h\in \mathcal H_{1}(\delta ^{\circ}\text{\textsf r}_{n})$:
\begin{align*}
&\frac{1}{n(n-1)}\sum_{i\neq i'} \bigl(Y_{i}-
\widehat h_{J}^{\text{\textsc r}}(X_{i}) \bigr)
\bigl(Y_{i'}-\widehat h_{J}^{\text{\textsc r}}(X_{i'})
\bigr)b^{K}(W_{i})' \bigl(A' A-
\widehat A'\widehat A \bigr)b^{K}(W_{i'})
\\
&\quad \lesssim n^{-1/2}s_{J}^{-1} \mathcal
C_{h} \bigl( \llVert h-\mathcal H_{0} \rrVert
_{L^{2}(X)}+J^{-p/d_{x}} \bigr) +n^{-1} s_{J}^{-2}
\sqrt{ J}.
\end{align*}
\end{lemma}
\begin{proof}
For simplicity of notation, we write $J$ instead of $J^\circ $ throughout
the proof. Following the proof of Lemma~\ref{lemma:est:matrices}, it is
sufficient to control
\begin{align*}
&\Evtex _{h}\bigl[\bigl(h-\widehat h_{J}^{\text{\textsc r}}
\bigr) (X)b^{K}(W)\bigr]' \bigl(A'A- \widehat
A'\widehat A \bigr)\Evtex _{h}\bigl[\bigl(h-\widehat
h_{J}^{\text{\textsc r}}\bigr) (X)b^{K}(W)\bigr]
\\
&\quad = 2\Evtex _{h}\bigl[\bigl(h-\widehat h_{J}^{\text{\textsc r}}
\bigr) (X)b^{K}(W)\bigr]'A'(A- \widehat A)
\Evtex _{h}\bigl[\bigl(h-\widehat h_{J}^{\text{\textsc r}}\bigr)
(X)b^{K}(W)\bigr]
\\
& \qquad {}-\Evtex _{h}\bigl[\bigl(h-\widehat h_{J}^{\text{\textsc r}}
\bigr) (X)b^{K}(W)\bigr]'(A- \widehat A)' (A-
\widehat A)\Evtex _{h}\bigl[\bigl(h-\widehat h_{J}^{\text{\textsc r}}
\bigr) (X)b^{K}(W)\bigr]=:2 T_{1}-T_{2},
\end{align*}
We first consider the term $T_{1}$ using the decomposition:
%
\begin{align}
\label{T1:dec:h_est} T_{1} ={}&\Evtex _{h}\bigl[\bigl(h-\widehat
h_{J}^{\text{\textsc r}}\bigr) (X)b^{K}(W)\bigr]'A'(
\widehat A-A)\Evtex _{h}\bigl[\Pi _{J}\bigl(h-\widehat
h_{J}^{\text{\textsc r}}\bigr) (X) b^{K}(W)\bigr]
\nonumber
\\
& {}+\Evtex _{h}\bigl[\bigl(h-\widehat h_{J}^{\text{\textsc r}}
\bigr) (X)b^{K}(W)\bigr]'A'(\widehat A-A)\nonumber
\\
&{}\times
\Evtex _{h}\bigl[\bigl(h-\widehat h_{J}^{\text{\textsc r}}-\Pi
_{J}\bigl(h-\widehat h_{J}^{
\text{\textsc r}}\bigr)\bigr)
(X)b^{K}(W)\bigr]. 
\end{align}
Consider the first summand on the right-hand side of equation (\ref{T1:dec:h_est}).
By Assumption~5(ii),
$\sup_{h\in \mathcal H_{1}(\delta ^{\circ}\text{\textsf r}_{n})}\mathrm{P}_{h}
 (\zeta _{J}\mathcal C_{h}\|\widehat h_{J}^{\text{\textsc r}}-\Pi _{
\mathcal H_0}h\|_{L^{2}(X)}> C )\to 0$ and consequently may assume that
$\widehat h_{J}^{\text{\textsc r}}\in \mathcal H_{0,J}(h):=\{\phi \in
\mathcal H_{0,J}: \|\phi -\Pi _{\mathcal H_0}h\|_{L^{2}(X)}\leq [
\zeta _{J}\mathcal C_{h}]^{-1} \}$. We calculate
\begin{align*}
&\sup_{\phi \in \mathcal H_{0,J}(h)} \bigl\llvert \bigl(\bigl(G_{b}^{-1/2}S
\bigr)_{l}^{-} \Evtex \bigl[(h-\phi ) (X)\widetilde{b}^{K}(W)\bigr] \bigr)'G \bigl(\bigl(G_{b}^{-1/2}S
\bigr)_{l}^{-}-\bigl( \widehat G_{b}^{-1/2}
\widehat S\bigr)_{l}^{-}\widehat G_{b}^{-1/2}G_{b}^{1/2}
\bigr)
\\
&\qquad {}\times \Evtex \bigl[(h-\phi ) (X)\widetilde{b}^{K}(W)\bigr] \bigr\rrvert
\\
&\quad =\sup_{\phi \in \mathcal H_{0,J}(h)} \biggl\llvert \bigl\langle Q_{J} (h-
\phi ), \psi ^{J}\bigr\rangle _{L^{2}(X)}'
\bigl(G_{b}^{-1/2}S\bigr)^{-}_{l}
\\
&\qquad {}\times  \biggl(
\frac{1}{n}\sum_{i} \Pi _{J}(h-\phi )
(X_{i})\widetilde{b}^{K}(W_{i})- \Evtex \bigl[\Pi
_{J}(h-\phi ) (X)\widetilde{b}^{K}(W)\bigr] \biggr) \biggr
\rrvert
\\
&\qquad  {}+\sup_{\phi \in \mathcal H_{0,J}(h)} \biggl\llvert \bigl\langle Q_{J} (h-
\phi ),\psi ^{J}\bigr\rangle _{L^{2}(X)}'
\bigl(G_{b}^{-1/2}S\bigr)^{-}_{l}
G_{b}^{-1/2} S \bigl(\bigl(\widehat G_{b}^{-1/2}
\widehat S\bigr)_{l}^{-}\widehat G_{b}^{-1/2}G_{b}^{1/2}-
\bigl(G_{b}^{-1/2}S\bigr)^{-}_{l} \bigr)
\\
&\qquad {} \times \biggl(\frac{1}{n}\sum_{i} \Pi
_{J}(h- \phi ) (X_{i})\widetilde{b}^{K}(W_{i})-
\Evtex \bigl[\Pi _{J}(h-\phi ) (X) \widetilde{b}^{K}(W)\bigr]
\biggr) \biggr\rrvert =:T_{11}+T_{12}.
\end{align*}%
Consider $T_{11}$, which coincides with the term $T_{32}$ in the proof
of Lemma~\ref{lemma:quad:fctl:est_h} and thus, we have
$\Evtex |T_{11}|\lesssim n^{-1/2}s_{J}^{-1} \mathcal C_{h} (\|h-
\mathcal H_{0}\|_{L^{2}(X)}+J^{-p/d_{x}} )$. To establish an upper bound
for $T_{12}$, we infer from
\citet [Lemma F.10(c)]{ChenChristensen2017} that
\begin{align*}
\llvert T_{12} \rrvert ^{2}\leq {}&\sup_{\phi \in \mathcal H_{0,J}(h)}
\bigl\llVert \bigl\langle Q_{J} (h-\phi ),\psi ^{J}\bigr
\rangle _{L^{2}(X)}'\bigl(G_{b}^{-1/2}S
\bigr)^{-}_{l} \bigr\rrVert ^{2}
\\
& {}\times \bigl\llVert G_{b}^{-1/2} S \bigl(\bigl(\widehat
G_{b}^{-1/2} \widehat S\bigr)_{l}^{-}
\widehat G_{b}^{-1/2}G_{b}^{1/2}-
\bigl(G_{b}^{-1/2}S\bigr)^{-}_{l} \bigr)
\bigr\rrVert ^{2}
\\
&{} \times \sup_{\phi \in \mathcal H_{0,J}(h)} \biggl\llVert \frac{1}{n} \sum
_{i} \Pi _{J}(h-\phi ) (X_{i})b^{K}(W_{i})-
\Evtex \bigl[\Pi _{J}(h-\phi ) (X) b^{K}(W)\bigr] \biggr
\rrVert ^{2}
\\
\lesssim{}& \sup_{\phi \in \mathcal H_{0,J}(h)} \bigl\llVert \bigl\langle
Q_{J} (h- \phi ),\psi ^{J}\bigr\rangle _{L^{2}(X)}'
\bigl(G_{b}^{-1/2}S\bigr)^{-}_{l} \bigr
\rrVert ^{2} \times n^{-1}s_{J}^{-2}\zeta
_{J}^{2}(\log J)\times n^{-1}\zeta
_{J}^{2} \mathcal C_{h}^{2}
\\
\lesssim{}& n^{-1}s_{J}^{-2} \mathcal
C_{h}^{2} \bigl( \llVert h-\Pi _{
\mathcal H_{0}} h \rrVert
_{L^{2}(X)}^{2}+J^{-2p/d_{x}} \bigr)
\end{align*}%
wpa1 uniformly for
$h\in \mathcal H_{1}(\delta ^{\circ}\text{\textsf r}_{n})$, where the last equation
is due to $s_{J}^{-1} \zeta _{J}^{2}\sqrt{(\log J)/ n}=O(1)$ from Assumption~4(i). Consider the second summand on the right-hand side
of equation \eqref{T1:dec:h_est}. Following the upper bound of
$T_{12}$, we obtain
\begin{align*}
&\sup_{\phi \in \mathcal H_{0,J}(h)} \bigl\llvert \Evtex \bigl[(h-\phi )
(X)b^{K}(W)\bigr]'A' G (\widehat A-A)\Evtex
\bigl[\bigl(h-\phi -\Pi _{J}(h-\phi )\bigr) (X) b^{K}(W)
\bigr] \bigr\rrvert ^{2}
\\
&\quad \leq \sup_{\phi \in \mathcal H_{0,J}(h)} \bigl\llVert \bigl\langle Q_{J}
(h- \phi ),\psi ^{J}\bigr\rangle _{L^{2}(X)}'
\bigl(G_{b}^{-1/2}S\bigr)^{-}_{l} \bigr
\rrVert ^{2}
\\
&\qquad {}\times  \bigl\llVert G_{b}^{-1/2}
S' \bigl(\bigl(\widehat G_{b}^{-1/2}\widehat S
\bigr)_{l}^{-} \widehat G_{b}^{-1/2}G_{b}^{1/2}-
\bigl(G_{b}^{-1/2}S\bigr)^{-}_{l} \bigr)
\bigr\rrVert ^{2}
\\
&\qquad {} \times \sup_{\phi \in \mathcal H_{0,J}(h)} \bigl\llVert \bigl\langle T\bigl(h- \phi
-\Pi _{J}(h-\phi )\bigr), \widetilde{b}^{K}\bigr\rangle
_{L^{2}(W)} \bigr\rrVert ^{2}
\\
&\quad \lesssim \sup_{\phi \in \mathcal H_{0,J}(h)} \bigl\llVert \bigl\langle
Q_{J} (h- \phi ),\psi ^{J}\bigr\rangle _{L^{2}(X)}'
\bigl(G_{b}^{-1/2}S\bigr)^{-}_{l} \bigr
\rrVert ^{2} \sup_{\phi \in \mathcal H_{0,J}(h)} \bigl\llVert \Pi
_{K}T\bigl(h-\phi -\Pi _{J}(h- \phi )\bigr) \bigr\rrVert
_{L^{2}(W)}^{2}
\\
&\qquad {} \times n^{-1}s_{J}^{-2}\zeta
_{J}^{2}(\log J)
\\
&\quad \lesssim n^{-1}s_{J}^{-2} \bigl( \llVert h-\Pi
_{\mathcal H_{0}} h \rrVert _{L^{2}(X)}^{2}+J^{-2p/d_{x}} \bigr)
\end{align*}%
wpa1 uniformly for
$h\in \mathcal H_{1}(\delta ^{\circ}\text{\textsf r}_{n})$, using that
$s_{J}^{-2}\|\Pi _{K} T(h-\Pi _{\mathcal H_0}h-\Pi _{J}(h-\Pi _{
\mathcal H_0}h))\|_{L^{2}(W)}^{2}\lesssim \|h-\Pi _{\mathcal H_0}h-
\Pi _{J}(h-\Pi _{\mathcal H_0}h)\|_{L^{2}(X)}^{2}$ by Assumption~4(i) and
$\zeta _{J}^{2}(\log J)\|h-\Pi _{J} h\|_{L^{2}(X)}^{2}=O(1)$ by Assumption~4(iii).

We now consider the term $T_{2}$ using the decomposition
\begin{align*}
T_{2}\leq{}& 2 \sup_{\phi \in \mathcal H_{0,J}(h)} \bigl\llvert \Evtex \bigl[
\Pi _{J}(h- \phi ) (X)b^{K}(W)\bigr]'(\widehat
A-A)' G (\widehat A-A)\Evtex \bigl[\Pi _{J}(h- \phi ) (X)
b^{K}(W)\bigr] \bigr\rrvert
\\
& {}+2 \sup_{\phi \in \mathcal H_{0,J}(h)} \bigl\llvert \Evtex \bigl[\Pi
_{J}^{
\perp}(h-\phi ) (X)b^{K}(W)
\bigr]'(\widehat A-A)' G (\widehat A-A)
\\
&{}\times \Evtex \bigl[\Pi
_{J}^{
\perp}(h-\phi ) (X) b^{K}(W)\bigr] \bigr\rrvert
\\
=:{}&2T_{21}+2T_{22},
\end{align*}
where $\Pi _{J}^{\perp}=\text{id}-\Pi _{J}$ is the projection. We bound
$T_{21}$ as follows:
\begin{align*}
T_{21}\leq{}& \sup_{\phi \in \mathcal H_{0,J}(h)} \biggl\llvert \bigl\langle
\Pi _{J} (h- \phi ),\psi ^{J}\bigr\rangle _{L^{2}(X)}'
\bigl(\bigl(\widehat G_{b}^{-1/2} \widehat S
\bigr)_{l}^{-}\widehat G_{b}^{-1/2}S-
I_{J} \bigr)'\bigl(\widehat G_{b}^{-1/2}
\widehat S\bigr)_{l}^{-}\widehat G_{b}^{-1/2}
\\
&{} \times \biggl(\frac{1}{n}\sum_{i} \Pi
_{J}(h- \phi ) (X_{i})\widetilde{b}^{K}(W_{i})-
\Evtex \bigl[\Pi _{J}(h-\phi ) (X) \widetilde{b}^{K}(W)\bigr]
\biggr) \biggr\rrvert
\\
\leq{}& \sup_{\phi \in \mathcal H_{0,J}(h)} \bigl\llVert \bigl\langle \Pi _{J}
(h- \phi ),\psi ^{J}\bigr\rangle _{L^{2}(X)} \bigr\rrVert \llVert
S-\widehat S \rrVert \bigl\llVert \bigl( \widehat G_{b}^{-1/2}
\widehat S\bigr)_{l}^{-}\widehat G_{b}^{-1/2}
\bigr\rrVert ^{2}
\\
& {}\times \biggl\llVert \frac{1}{n}\sum_{i} \Pi
_{J}(h- \phi ) (X_{i})\widetilde{b}^{K}(W_{i})-
\Evtex \bigl[\Pi _{J}(h-\phi ) (X) \widetilde{b}^{K}(W)\bigr]
\biggr\rrVert
\\
\lesssim{}& \sup_{\phi \in \mathcal H_{0,J}(h)} \bigl\llVert \Pi _{J} (h-\phi
) \bigr\rrVert _{L^{2}(X)} \times n^{-1/2}s_{J}^{-2}
\zeta _{J}\sqrt{\log J}\times n^{-1/2}\zeta _{J}
\mathcal C_{h}
\\
\lesssim{}& n^{-1/2}s_{J}^{-1} \mathcal
C_{h} \bigl( \llVert h-\Pi _{\mathcal H_{0}} h \rrVert
_{L^{2}(X)}+J^{-p/d_{x}} \bigr)
\end{align*}
wpa1 uniformly for
$h\in \mathcal H_{1}(\delta ^{\circ}\text{\textsf r}_{n})$. For $T_{22}$, we note
that uniformly in $h\in \mathcal H$ and
$\phi \in \mathcal H_{0,J}(h)$,
$\|\Evtex [\Pi _{J}^{\perp}( h-\phi )(X)\widetilde{b}^{K}(W)]\|=\|\Pi _{K} T
\Pi _{J}^{\perp }(h-\phi )\|_{L^{2}(W)}\lesssim s_{J} J^{-p/d_{x}}$ by
Assumption~2(iv). Thus, following the upper bound derivations
of $T_{21}$, we obtain
$T_{22}\lesssim n^{-1/2}s_{J}^{-1} J^{-p/d_{x}}$ wpa1 uniformly for
$h\in \mathcal H_{1}(\delta ^{\circ}\text{\textsf r}_{n})$.
\end{proof}

\begin{lemma}
\label{lemma:unif:matrix:bound}
Let Assumptions 1(i)--(iii), 2(i), and
4 be satisfied. Then, using the notation
$S^{o}:=G_{b}^{-1/2} SG^{-1/2}$, we have for some constant $C>0$:
\begin{align*}
&\mathrm{(i)}\qquad  \mathrm{P} \biggl(\max_{J\in \mathcal I_{n}} \biggl\{
\frac{s_{J}^{2}\sqrt n}{\zeta _{J}\sqrt{\log J}}
\bigl\llVert \bigl(\widehat G_{b}^{-1/2} \widehat S\widehat
G^{-1/2}\bigr)^{-}_{l}\widehat G_{b}^{-1/2}
G_{b}^{1/2}- \bigl(S^{o}\bigr)^{-}_{l}
\bigr\rrVert \biggr\}>C \biggr)=o(1),
\\
&\mathrm{(ii)}\qquad  \mathrm{P} \biggl(\max_{J\in \mathcal I_{n}} \biggl\{
\frac{s_{J}^{2}\sqrt n}{\zeta _{J}\sqrt{\log J}}
\bigl\llVert S^{o} \bigl(\bigl( \widehat G_{b}^{-1/2}
\widehat S\widehat G^{-1/2}\bigr)^{-}_{l}\widehat
G_{b}^{-1/2} G_{b}^{1/2}-
\bigl(S^{o}\bigr)^{-}_{l} \bigr) \bigr\rrVert
\biggr\}>C \biggr)=o(1).
\end{align*}
\end{lemma}
\begin{proof}
The results can be established by following the same proof from
\citet [Lemma C.4]{chen2021} with their $(\tau _{J}, \sqrt{J})$ replaced
by our $(s_{J}^{-1},\zeta _{J})$.
\end{proof}

\begin{lemma}%
\label{lemma:est:matrices:H0}
Let Assumptions 1(i)--(iii), 2(i), and
4(i) hold. Then, we have
\begin{equation*}
\begin{split}&\mathrm{P}_{h} \biggl(\max_{J\in \mathcal I_{n}} \biggl\llvert
\frac{(\log \log J)^{-1/2}}{(n-1)V_{J}}\sum
_{i\neq i'} U_{i}(\Pi _{
\mathcal H_0}h)U_{i'}(
\Pi _{\mathcal H_0}h)b^{K}(W_{i})'
\bigl(A'A- \widehat A'\widehat A \bigr)b^{K}(W_{i'})
\biggr\rrvert >\frac{1-c_{0}}{8} \biggr)
\\
&\quad =o(1)
\end{split}\end{equation*}
uniformly for $h\in \mathcal H_{0}$, where
$U_{i}(\phi )=Y_{i}-\phi (X_{i})$ and $c_{0}$ is as in the proof of Theorem~4.1.
\end{lemma}
\begin{proof}
Let $I_{s_{J}}$ denote the $J$-dimensional identity matrix multiplied by
the vector $C_{0}(s_{1},\dots , s_{J})'$ for some sufficiently large constant
$C_{0}$ and where $s_{j}^{-1}$, $1\leq j\leq J$, are the nondecreasing
singular values of $AG_{b}^{1/2}=(G_{b}^{-1/2}S G^{-1/2})^{-}_{l}$. There
exists a unitary matrix $Q$ such that
\begin{align*}
&\sum_{i\neq i'} U_{i}(\Pi
_{\mathcal H_0}h)U_{i'}(\Pi _{\mathcal H_0}h)b^{K}(W_{i})'
\bigl(A'A-\widehat A'\widehat A \bigr)b^{K}(W_{i'})
\\
&\quad \leq \biggl\llVert \sum_{i} U_{i}(\Pi
_{\mathcal H_0}h)\widetilde{b}^{K}(W_{i})'Q
I_{s_{J}}^{-1} \biggr\rrVert ^{2} \bigl\llVert
I_{s_{J}}Q'G_{b}^{1/2}
\bigl(A' A-\widehat A' \widehat A\bigr)G_{b}^{1/2}QI_{s_{J}}
\bigr\rrVert
\\
&\quad = \sum_{i\neq i'}U_{i}(\Pi
_{\mathcal H_0}h)U_{i'}(\Pi _{\mathcal H_0}h) \widetilde{b}^{K}(W_{i})' Q I_{s_{J}}^{-2}Q'
\widetilde{b}^{K}(W_{i'}) \bigl\llVert I_{s_{J}}Q'G_{b}^{1/2}
\bigl(A' A-\widehat A'\widehat A\bigr)G_{b}^{1/2}QI_{s_{J}}
\bigr\rrVert
\\
&\qquad  {}+\sum_{i} \bigl\llVert U_{i}(\Pi
_{\mathcal H_0}h)\widetilde{b}^{K}(W_{i})Q
I_{s_{J}}^{-1} \bigr\rrVert ^{2} \bigl\llVert
I_{s_{J}}Q'G_{b}^{1/2}
\bigl(A' A-\widehat A' \widehat A\bigr)G_{b}^{1/2}QI_{s_{J}}
\bigr\rrVert .
\end{align*}
The fourth moment condition imposed in Assumption~2(i) implies
uniformly for $h\in \mathcal H_{0}$:
\begin{align*}
&\Evtex _{h} \max_{J\in \mathcal I_{n}} \biggl\llvert
\frac{1}{nV_{J}}\sum_{i} \bigl( \bigl
\llVert U_{i}(\Pi _{\mathcal H_0}h)\widetilde{b}^{K}(W_{i})QI_{s_{J}}^{-1}
\bigr\rrVert ^{2}-\Evtex _{h} \bigl\llVert U(\Pi
_{\mathcal H_0}h)\widetilde{b}^{K}(W)QI_{s_{J}}^{-1}
\bigr\rrVert ^{2} \bigr) \biggr\rrvert ^{2}
\\
&\quad \lesssim n^{-1}\zeta _{\overline J}^{2}\sum
_{J\in \mathcal I_{n}} V_{J}^{-2}s_{J}^{-4}
\\
&\quad \lesssim n^{-1}\zeta _{\overline J}^{2}\sum
_{J\in \mathcal I_{n}} \Biggl(\sum_{j=1}^{J}
s_{J}^{4} s_{j}^{-4}
\Biggr)^{-1} \lesssim n^{-1} \zeta _{\overline J}^{2}
\sum_{J\in \mathcal I_{n}} J^{-1}=o(1),
\end{align*}%
due to Lemma~B.3 and the definition of the index set
$\mathcal I_{n}$. Consequently, from the second moment condition imposed
in Assumption~2(i), we obtain uniformly for
$J\in \mathcal I_{n}$:
\begin{align*}
n^{-1}\sum_{i} \bigl\llVert
\bigl(Y_{i}-\Pi _{\mathcal H_0}h(X_{i})\bigr)\widetilde{b}^{K}(W_{i})Q I_{s_{J}}^{-1} \bigr\rrVert
^{2}\leq \overline\sigma ^{2}c_{0}^{-1}
\zeta _{J} \Biggl(\sum_{j=1}^{J}s_{j}^{-4}
\Biggr)^{1/2}\leq \overline\sigma ^{2} \underline\sigma
^{-2}c_{0}^{-1}\zeta _{J}V_{J}
\end{align*}
with probability approaching 1 (under $h\in \mathcal H_{0}$), by making
use of Lemma~B.3. Further, we obtain uniformly for
$h\in \mathcal H_{0}$:
\begin{eqnarray*}
&&\mathrm{P}_{h} \biggl(\max_{J\in \mathcal I_{n}} \biggl\llvert
\frac{(\log \log J)^{-1/2}}{(n-1)V_{J}}\sum
_{i,i'} U_{i}(\Pi _{
\mathcal H_0}h)U_{i'}(
\Pi _{\mathcal H_0}h)b^{K}(W_{i})'
\bigl(A'A- \widehat A'\widehat A \bigr)b^{K}(W_{i'})
\biggr\rrvert >\frac{1-c_{0}}{8} \biggr)
\\
&&\quad \leq \mathrm{P}_{h} \biggl(\max_{J\in \mathcal I_{n}} \biggl
\llvert \frac{(\log \log J)^{-1/2}}{(n-1)V_{J}}\sum
_{i\neq i'} U_{i}(\Pi _{
\mathcal H_0}h)U_{i'}(
\Pi _{\mathcal H_0}h)\widetilde{b}^{K}(W_{i})'Q
I_{s_{J}}^{-2} Q'\widetilde{b}^{K}(W_{i'})
\biggr\rrvert >\frac{1-c_{0}}{8} \biggr)
\\
&& \qquad {}+ \mathrm{P}_{h} \biggl(\max_{J\in \mathcal I_{n}} \bigl( \bigl
\llVert I_{s_{J}}Q G_{b}^{1/2} \bigl(A'A-
\widehat A'\widehat A\bigr) G_{b}^{1/2}QI_{s_{J}}
\bigr\rrVert \bigr)>\frac{1-c_{0}}{16} \biggr)
\\
&&\qquad  {}+ \mathrm{P}_{h} \biggl(\max_{J\in \mathcal I_{n}} \bigl(
\overline\sigma ^{2}\underline\sigma ^{-2}c_{0}^{-1}
\zeta _{J}(\log \log J)^{-1/2} \bigl\llVert I_{s_{J}}Q
G_{b}^{1/2}\bigl(A'A-\widehat A'
\widehat A\bigr)G_{b}^{1/2}Q'I_{s_{J}}
\bigr\rrVert \bigr)>\frac{1-c_{0}}{16} \biggr)
\\
&&\qquad {} +o(1)
\\
&&\quad =:T_{1}+T_{2}+T_{3}+o(1).
\end{eqnarray*}%
Note that $T_{1}$ is arbitrarily small for $C_{0}$ sufficiently large by
following Step~1 in the proof of Theorem~4.1. Consider
$T_{2}$. We make use of the inequality
\begin{equation*}
\begin{split}&\bigl\llVert I_{s_{J}}Q G_{b}^{1/2}\bigl(\widehat
A' \widehat A -A' A\bigr)G_{b}^{1/2}Q'I_{s_{J}}
\bigr\rrVert
\\
&\quad \leq 2 \bigl\llVert I_{s_{J}}Q G_{b}^{1/2}
(\widehat A-A )' AG_{b}^{1/2}Q'I_{s_{J}}
\bigr\rrVert + \bigl\llVert (\widehat A- A)G_{b}^{1/2}QI_{s_{J}}
\bigr\rrVert ^{2}.
\end{split}\end{equation*}
It is sufficient to consider the first summand on the right-hand side.
Note that $\|AG_{b}^{1/2}\*Q'I_{s_{J}}\|\leq C_{0}^{-1}$. Consequently, from
Lemma~\ref{lemma:unif:matrix:bound}(ii) we infer
\begin{align*}
&\mathrm{P} \biggl(\max_{J\in \mathcal I_{n}} \biggl\{ \frac{s_{J}^{2}
\sqrt n}{\zeta _{J}\sqrt{\log J}} \bigl\llVert I_{s_{J}} Q
G_{b}^{1/2} (\widehat A-A )'
AG_{b}^{1/2}Q'I_{s_{J}} \bigr\rrVert
\biggr\}>C \biggr)=o(1).
\end{align*}
Assumption~4(i), that is,
$s_{J}^{-1}\zeta _{J}^{2}\sqrt{(\log J)/n}=O(1)$ uniformly for
$J\in \mathcal I_{n}$, thus implies $T_{3}=o(1)$.
\end{proof}

\begin{proof}[Proof of Lemma~B.4]
It is sufficient to prove (ii). Let
$\Sigma =\Evtex _{h}[(Y-h(X))^{2} b^{K(J)}(W)\* b^{K(J)}(W)']$ and
$\widehat \Sigma =n^{-1}\sum_{i}  (Y_{i}-\widehat h_{J}(X_{i})
 )^{2}b^{K(J)}(W_{i}) b^{K(J)}(W_{i})'$. Then
$V_{J}= \|A \Sigma A' \|_{F}$ and
$\widehat{V}_{J}= \|\widehat A \widehat \Sigma \widehat A' \|_{F}$.
For all $J\in \mathcal I_{n}$, the triangular inequality implies
\begin{equation*}
\begin{split}\llvert \widehat{V}_{J}-V_{J} \rrvert &\leq \bigl\llVert
\widehat A \widehat \Sigma \widehat A' -A\Sigma A' \bigr
\rrVert _{F}
\\
&\leq 2 \bigl\llVert (\widehat A-A ) \widehat \Sigma
A' \bigr\rrVert _{F}+ \bigl\llVert (\widehat A - A )
\widehat \Sigma ^{1/2} \bigr\rrVert _{F}^{2}+ \bigl
\llVert A (\widehat \Sigma - \Sigma )A' \bigr\rrVert _{F}.
\end{split}\end{equation*}
In the remainder of this proof, it is sufficient to consider
$\|(\widehat A-A)\Sigma   A'\|_{F}+\|A(\widehat \Sigma -\Sigma )  A'
\|_{F}=:T_{1}+T_{2}$. Consider $T_{1}$. By Lemma~\ref{lemma:bound:var}, we have the upper bound
$\|G_{b}^{-1/2}\Sigma\* G_{b}^{-1/2}\|\leq \overline\sigma $. Below, we make
use of the inequality
$\|m_{1}m_{2}\|_{F}\leq \|m_{1}\| \|m_{2}\|_{F}$ for matrices
$m_{1}$ and $m_{2}$. Since the Frobenius norm is invariant under rotation,
we calculate uniformly for $J\in \mathcal I_{n}$ that
\begin{align*}
T_{1}&= \bigl\llVert \bigl(G_{b}^{1/2}
SG^{1/2}\bigr) (\widehat A-A )\Sigma A'AG_{b}^{1/2}
\bigr\rrVert
\\
&\leq \bigl\llVert \bigl(G_{b}^{1/2} SG^{1/2}\bigr)
(\widehat A- A)G_{b}^{1/2} \bigr\rrVert \bigl\llVert
G_{b}^{-1/2} \Sigma G_{b}^{-1/2} \bigr
\rrVert \bigl\llVert \bigl(G_{b}^{1/2} SG^{1/2}
\bigr)^{-2}_{l} \bigr\rrVert _{F}
\\
&\lesssim \frac{
\zeta _{J}}{s_{J}} \Biggl(\frac{\log (J)}{n}\sum
_{j=1}^{J}s_{j}^{-4}
\Biggr)^{1/2}
\end{align*}%
wpa1 uniformly for $h\in \mathcal H$, by making use of Lemma~\ref{lemma:unif:matrix:bound}(i) and the Schur decomposition as in the
proof of Lemma~B.3. From Assumption~4(i), that is,
$s_{J}^{-1}\zeta _{J}^{2}\sqrt{(\log J)/n}=O(1)$, uniformly for
$J\in \mathcal I_{n}$, we infer
$T_{1}/V_{J}=J^{-1/2}(\sum_{j=1}^{J}s_{j}^{-4})^{1/2}/V_{J}\to 0$ wpa1
uniformly for $h\in \mathcal H$, where the last equation is due to Lemma~B.3. Consider $T_{2}$. Again using Lemma~B.3, we obtain
$T_{2}\leq \underline\sigma ^{-2}\|G_{b}^{-1/2}(\widehat \Sigma -
\Sigma ) G_{b}^{-1/2}\|$ by using the upper bound as derived for
$T_{1}$. Further, evaluate
\begin{align*}
\bigl\llVert G_{b}^{-1/2}(\widehat \Sigma-\Sigma )
G_{b}^{-1/2} \bigr\rrVert ={}& \biggl\llVert \frac{1}{n}\sum
_{i} \bigl(\bigl(Y_{i}-\widehat
h_{J}(X_{i})\bigr)^{2}-\bigl(Y_{i}-h(X_{i})
\bigr)^{2} \bigr)\widetilde{b}^{K}(W_{i})
\widetilde{b}^{K}(W_{i})' \biggr\rrVert
\\
\leq{}& \biggl\llVert \frac{1}{n}\sum_{i} \bigl(
\widehat h_{J}(X_{i})-h(X_{i})
\bigr)^{2}\widetilde{b}^{K}(W_{i})\widetilde{b}^{K}(W_{i})' \biggr\rrVert
\\
& {}+2 \biggl\llVert \frac{1}{n}\sum_{i} \bigl(
\widehat h_{J}(X_{i})-h(X_{i}) \bigr)
\bigl(Y_{i}-h(X_{i})\bigr)\widetilde{b}^{K}(W_{i})
\widetilde{b}^{K}(W_{i})' \biggr\rrVert
\\
=:{}&T_{21}+T_{22}.
\end{align*}
Consider $T_{21}$. The definition of the unrestricted sieve NPIV estimator
in (2.5) implies uniformly for $J\in \mathcal I_{n}$:
\begin{align*}
T_{21}\leq{}& \biggl\llVert \frac{1}{n}\sum
_{i} \bigl(\widehat h_{J}(X_{i})-Q_{J}h(X_{i})
\bigr)^{2}\widetilde{b}^{K}(W_{i})\widetilde{b}^{K}(W_{i})' \biggr\rrVert
\\
&{} + \biggl\llVert
\frac{1}{n}\sum_{i} \bigl(Q_{J}h(X_{i})-h(X_{i})
\bigr)^{2} \widetilde{b}^{K}(W_{i})\widetilde{b}^{K}(W_{i})' \biggr\rrVert
\\
\leq{}& \zeta _{J}^{2} \biggl\llVert \widehat A\frac{1}{n}\sum
_{i}Y_{i}b^{K}(W_{i})-A
\Evtex _{h}\bigl[Yb^{K}(W)\bigr] \biggr\rrVert
^{2}\times \biggl\llVert \frac{1}{n}\sum_{i}
\psi ^{J}(X_{i}) \psi ^{J}(X_{i})'
\biggr\rrVert
\\
& {}+ \zeta _{J}^{2} \biggl\llVert \frac{1}{n}\sum
_{i} \bigl(Q_{J}h(X_{i})-h(X_{i})
\bigr)^{2} \biggr\rrVert
\\
\lesssim{}& \zeta _{\overline J}^{4}
s_{\overline J}^{-2} n^{-1}+\max_{J\in \mathcal I_{n}}
\bigl\{\zeta _{J}^{2} \llVert Q_{J}h-h \rrVert
_{L^{2}(X)} \bigr\}
\end{align*}%
 wpa1 uniformly for $h\in \mathcal H$, where
the right-hand side tends to zero. This follows by the rate condition imposed
in Assumption~4(i) and that
$\|Q_{J}h-h\|_{L^{2}(X)}=O(J^{-p/d_{x}})$ uniformly for
$J\in \mathcal I_{n}$ and $h\in \mathcal H$ by Lemma~B.1(ii).
Analogously, we obtain that $\max_{J\in \mathcal I_{n}}T_{22}$ vanishes
wpa1 uniformly for $h\in \mathcal H$.
\end{proof}

\begin{proof}[Proof of Lemma~B.5]
We first prove the lower bound. By the definition of the RES index set
$\widehat{\mathcal I}_{n}$, we have that any element
$J\in \widehat{\mathcal I}_{n}$ tends slowly to infinity as
$n\to \infty $. Let $\widehat j_{\max}\leq j_{\max}$ be the largest integer
such that $\underline J2^{\widehat j_{\max}}\leq \widehat J_{\max}$. Consequently,
the definition of the RES index set implies for all
$J\in \widehat{\mathcal I}_{n}$ that
\begin{align*}
\log (J)\leq \log \bigl(\underline J2^{\widehat j_{\max}}\bigr)=\widehat
j_{
\max}\log (2)+\log (\underline J)\leq \widehat j_{\max}+1=\#(
\widehat{\mathcal I}_{n})
\end{align*}
for $n$ sufficiently large. From the lower bounds for quantiles of the
chi-squared distribution established in
\citet [Theorem~5.2]{inglot2010}, we deduce for all
$J\in \widehat{\mathcal I}_{n}$ and $n$ sufficiently large:
\begin{align*}
\widehat \eta _{J}(\alpha )&= \frac{q \bigl(\alpha /\#(\widehat{
\mathcal I}_{n}), J \bigr)-J}{\sqrt{J}}
\\
&\geq \frac{q \bigl(\alpha /(
\log J), J \bigr)-J}{\sqrt{J}}
\\
&\geq \frac{\sqrt{\log \bigl((\log J)/\alpha \bigr)}}{4}+ \frac{2\log \bigl((
\log J)/\alpha \bigr)}{\sqrt{J}}
\\
&\geq \frac{\sqrt{\log \log (J)-\log (\alpha
)}}{4}
\end{align*}
using the lower bounds for quantiles of the chi-squared distribution established
in \citet [Theorem~5.2]{inglot2010}. We now consider the upper bound. From
the definition of $\#(\widehat{\mathcal I}_{n})$, we infer
$\#(\widehat{\mathcal I}_{n})=\widehat j_{\max}+1\leq \lceil \log _{2}(n^{1/3}/
\underline J)\rceil +1 \leq \log (n^{1/3}/\underline J)+1$ and thus
$\#(\widehat{\mathcal I}_{n})\leq \log (n)$. Consequently, we calculate
for all $J\in \widehat{\mathcal I}_{n}$ and $n$ sufficiently large:
\begin{align*}
\widehat \eta _{J}(\alpha )&\leq \frac{q \bigl(\alpha /(\log n), J
\bigr)-J}{\sqrt{J}}
\\
&\leq 2 \sqrt{\log \bigl((\log n)/\alpha \bigr)}+ \frac{2\log
\bigl((\log n)/\alpha \bigr)}{\sqrt{J}}
\\
&\leq 2\sqrt{\log \bigl((\log n)/\alpha \bigr)}\bigl(1+o(1)\bigr) \leq 4 \sqrt{
\log \log (n)-\log (\alpha )},
\end{align*}
where the second inequality is due to \citet [Lemma~1]{laurent2000}.
\end{proof}

\begin{proof}[Proof of Lemma~B.6]
Result~B.6(i) directly follows from
\citet [Theorem~3.4]{houdre2003}; see also
\citet [Theorem~3.4.8]{GNbook}. We next prove the bounds on
$\Lambda _{1}$, $\Lambda _{2}$, $\Lambda _{3}$, $\Lambda _{4}$ for Result~B.6(ii).%

For the bound on $\Lambda _{1}$, we recall the notation
$U_{i}^{J}=U_{i}Ab^{K}(W_{i})$ with $U_{ij}$ as its $j$th entry for $1\leq j \leq J$, and $U_{i}=Y_{i}-h(X_{i})$ for
$h\in \mathcal H_{0}$. Then, under $\mathcal H_{0}$, we have
\begin{align*}
\Evtex _{h}\bigl[R_{1}^{2}(Z_{1},Z_{2})
\bigr]&\leq \Evtex _{h} \bigl\llvert U_{1}b^{K}(W_{1})'A'
Ab^{K}(W_{2})U_{2} \bigr\rrvert ^{2}
\\
& =
\Evtex _{h} \bigl[\bigl(U^{J}\bigr)' \Evtex
_{h} \bigl[U^{J}\bigl(U^{J}\bigr)'
\bigr] U^{J} \bigr]
\\
&=\sum_{j,j'=1}^{J}\Evtex _{h}[U_{1j}
U_{1j'}]^{2}=V_{J}^{2}.
\end{align*}
For the bound on $\Lambda _{2}$, for any function $\nu $ and
$\kappa $ with $\|\nu \|_{L^{2}(Z)}\leq 1$ and
$\|\kappa \|_{L^{2}(Z)}\leq 1$, respectively, we obtain
\begin{align*}
&\bigl\llvert \Evtex _{h}\bigl[R_{1}(Z_{1},Z_{2})
\nu (Z_{1})\kappa (Z_{2})\bigr] \bigr\rrvert
\\
&\quad  \leq \bigl
\llvert \Evtex _{h}\bigl[U \1_{M} b^{K}(W)'
\nu (Z)\bigr]A' A \Evtex _{h}\bigl[U\1_{M}b^{K}(W)
\kappa (Z)\bigr] \bigr\rrvert
\\
&\quad \leq \bigl\llVert A \Evtex _{h}\bigl[U\1_{M}b^{K}(W)
\kappa (Z)\bigr] \bigr\rrVert \bigl\llVert A \Evtex _{h}\bigl[U
\1_{M}b^{K}(W) \nu (Z)\bigr] \bigr\rrVert
\\
&\quad \leq \bigl\llVert AG_{b}^{1/2} \bigr\rrVert ^{2}
\sqrt{\Evtex \bigl[ \bigl\llvert \Evtex _{h}\bigl[U\1_{M}
\kappa (Z)|W\bigr] \bigr\rrvert ^{2} \bigr]}\times \sqrt{\Evtex \bigl[
\bigl\llvert \Evtex _{h}\bigl[U\1_{M}\nu (Z)|W\bigr] \bigr
\rrvert ^{2} \bigr]}.
\end{align*}
Now observe
$\Evtex  [|\Evtex _{h}[U\1_{M}\kappa (Z)|W]|^{2} ]\leq \Evtex  [\Evtex _{h}[U^{2}|W]
\kappa ^{2}(Z) ]\leq \overline\sigma ^{2}$ by Assumption~2(i) and using that $\|\kappa \|_{L^{2}(Z)}\leq 1$, which yields
the upper bound by using $\| AG_{b}^{1/2}\|=s_{J}^{-1}$.%

For the bound on $\Lambda _{3}$, observe that, for any $z=(u,w)$,
\begin{align*}
\bigl\llvert \Evtex _{h}\bigl[R_{1}^{2}(Z_{1},z)
\bigr] \bigr\rrvert &\leq \Evtex _{h} \bigl\llvert U\1\bigl\{ \llvert U
\rrvert \leq M_{n} \bigr\}b^{K}(W)'A'
Ab^{K}(w)u\1\bigl\{ \llvert u \rrvert \leq M_{n}\bigr\}
\bigr\rrvert ^{2}
\\
&\leq \bigl\llVert Ab^{K}(w)u\1\bigl\{ \llvert u \rrvert \leq
M_{n}\bigr\} \bigr\rrVert ^{2} \Evtex _{h} \bigl
\llVert Ab^{K}(W)U \bigr\rrVert ^{2}
\\
&\leq \overline\sigma
^{2} M_{n}^{2} \zeta _{b,K}^{2}
\bigl\llVert AG_{b}^{1/2} \bigr\rrVert ^{4},
\end{align*}
again by using Assumption~2(i) and hence the upper bound on
$\Lambda _{3}$ follows.%

For the bound on $\Lambda _{4}$, observe that for any
$z_{1}=(u_{1},w_{1})$ and $z_{2}=(u_{2},w_{2})$, we get
\begin{align*}
\bigl\llvert R_{1}(z_{1},z_{2}) \bigr\rrvert &
\leq \bigl\llvert u_{1}\1\bigl\{ \llvert u_{1} \rrvert \leq
M_{n}\bigr\}b^{K}(w_{1})'A'
Ab^{K}(w_{2})u_{2}\1\bigl\{ \llvert
u_{2} \rrvert \leq M_{n}\bigr\} \bigr\rrvert
\\
&\leq \sup_{u,w} \bigl\llVert Ab^{K}(w)u\1\bigl\{
\llvert u \rrvert \leq M_{n}\bigr\} \bigr\rrVert ^{2}\leq
M_{n}^{2} \zeta _{b,K}^{2} \bigl\llVert
AG_{b}^{1/2} \bigr\rrVert ^{2},
\end{align*}
which completes the proof.
\end{proof}

\begin{proof}[Proof of Lemma~B.7]
It suffices to prove (ii) for a simple null
$\mathcal H_{0} =\{h_{0}\}$. For any
$h\in \mathcal H_{1}(\delta ^{\circ}\text{\textsf r}_{n})$, we denote
$B_{J}=(\|\Evtex _{h}[U^{J}]\|-\|h-h_{0}\|_{L^{2}(X)})^{2}$. Recall $J^{\circ}\leq J^{*} < 2J^{\circ}$, applying
$\|\Evtex _{h}[U^{J^{*}}]\|^{2}=\|Q_{J^{*}}(h-h_{0})\|_{L^{2}(X)}^{2}$ and
Lemma~B.1(i), we obtain:
$B_{J^{*}}=  (\|Q_{J^{*}}(h-h_{0})\|_{L^{2}(X)}-\|h-h_{0}\|_{L^{2}(X)}
  )^{2} \leq C_{B}  \text{\textsf r}_{n}^{2}$ for some constant
$C_{B}$. By the inequality
$\|\Evtex _{h}[U^{J^{*}}]\|^{2}\geq \|h-h_{0}\|_{L^{2}(X)}^{2}/2-B_{J^{*}}$,
we have uniformly for $h\in \mathcal H_{1}(\delta ^{\circ }\text{\textsf r}_{n})$:
\begin{align*}
&\mathrm{P}_{h} \bigl(n \widehat{D}_{J^{*}}(h_{0})
\leq 2c_{1}\sqrt{ \log \log n} V_{J^{*}} \bigr)
\\
&\quad =
\mathrm{P}_{h} \biggl( \bigl\llVert \Evtex _{h}
\bigl[U^{J^{*}}\bigr] \bigr\rrVert ^{2}-\widehat{D}_{J^{*}}(h_{0})>
\bigl\llVert \Evtex _{h}\bigl[U^{J^{*}}\bigr] \bigr\rrVert
^{2} - \frac{2c_{1}\sqrt{\log \log n} V_{J^{*}}}{n}
\biggr)&
\\
&\quad \leq \mathrm{P}_{h} \Biggl( \Biggl\llvert \frac{4}{n(n-1)}\sum
_{j=1}^{J^{*}} \sum
_{i< i'} \bigl( U_{ij}U_{i'j}- \Evtex
_{h}[U_{1j}]^{2} \bigr) \Biggr\rrvert >\rho
_{h} \Biggr)
\\
& \qquad {}+ \mathrm{P}_{h} \bigg(\bigg|\frac{4}{n(n-1)}\sum
_{i< i'} \bigl(Y_{i}-h_{0}(X_{i})
\bigr) \bigl(Y_{i'}-h_{0}(X_{i'})
\\
&\qquad {}\times b^{K^{*}}(W_{i})'
\bigl(A'A-\widehat A'\widehat A \bigr)b^{K^{*}}(W_{i'})\bigg|>
\rho _{h} \biggr)
\\
&\quad =T_{1}+T_{2},
\end{align*}%
where
$\rho _{h}=\|h-h_{0}\|_{L^{2}(X)}^{2}/2-2c_{1}n^{-1}\sqrt{\log \log n}V_{J^{*}}
-B_{J^{*}}$. To bound term $T_{1}$, we apply inequality
\eqref{bound:var:V} and Markov's inequality:
%
\begin{align}
\label{Markovs:bound} T_{1}\lesssim n^{-1}s_{J^{*}}^{-2}
\rho _{h}^{-2} \bigl( \llVert h-h_{0} \rrVert
_{L^{2}(X)}^{2}+\bigl(J^{*}\bigr)^{-2p/d_{x}}
\bigr)+n^{-2}V_{J^{*}}^{2}\rho _{h}^{-2}.
\end{align}
In the following, we distinguish between two cases. First, consider the
case where $n^{-2}V_{J^{*}}^{2}\rho _{h}^{-2}$ dominates the right-hand
side. For any $h\in \mathcal H_{1}(\delta ^{\circ }\text{\textsf r}_{n})$, we
have $\|h-h_{0}\|_{L^{2}(X)}\geq \delta ^{\circ}\text{\textsf r}_{n} $ and hence,
we obtain the lower bound
%
\begin{equation}
\label{lower:bound:typeII} \rho _{h}= \llVert h-h_{0} \rrVert
_{L^{2}(X)}^{2}/2-2c_{1}n^{-1}\sqrt{\log \log
n}V_{J^{*}} -B_{J^{*}} \geq \kappa _{0} \text{\textsf
r}_{n}^{2}, 
\end{equation}
where $\kappa _{0}:=(\delta ^{\circ})^{2}/2-C-C_{B}$ for some constant
$C>0$ and $\kappa _{0}>0$ whenever
$\delta ^{\circ}>\sqrt{2(C+C_{B})}$. From inequality
\eqref{Markovs:bound}, we infer
$T_{1}\lesssim n^{-2}V_{J^{*}}^{2}(J^{*})^{4p/d_{x}}=o(1)$. Second, consider
the case where
$n^{-1}s_{J^{*}}^{-2}\rho _{h}^{-2} (\|h-h_{0}\|_{L^{2}(X)}^{2}+(J^{*})^{-2p/d_{x}}
 )$ dominates. For any
$h\in \mathcal H_{1}(\delta ^{\circ}\text{\textsf r}_{n})$, we have
$\|h-h_{0}\|_{L^{2}(X)}^{2}\geq (\delta ^{\circ})^{2}\text{\textsf r}_{n}^{2}
\geq 5c_{1}n^{-1}V_{J^{*}}\sqrt{\log \log n}$ for $\delta ^{\circ}$ sufficiently
large and hence, we obtain
$\rho _{h}\geq \kappa _{1}  \|h-h_{0}\|_{L^{2}(X)}^{2}$ for some constant
$\kappa _{1}:=1/5-C_{B}/(\delta ^{\circ})^{2}$, which is positive for any
$\delta ^{\circ}>\sqrt{5 C_{B}}$. Under Assumption 3, inequality
\eqref{Markovs:bound} yields uniformly for
$h\in \mathcal H_{1}(\delta ^{\circ}\text{\textsf r}_{n})$ that
\begin{align*}
T_{1}\lesssim n^{-1}s_{J^{*}}^{-2} \bigl(
\llVert h-h_{0} \rrVert _{L^{2}(X)}^{-2}+ \llVert
h-h_{0} \rrVert ^{-4}_{L^{2}(X)} \bigl(J^{*}
\bigr)^{-2p/d_{x}} \bigr) \lesssim n^{-1}s_{J^{*}}^{-2}
\text{\textsf r}_{n}^{-2} =o(1).
\end{align*}
Finally, $T_{2}=o(1)$ uniformly for
$h\in \mathcal H_{1}(\delta ^{\circ}\text{\textsf r}_{n})$ by making use of Lemma~\ref{lemma:est:matrices}.
\end{proof}

\begin{proof}[Proof of Lemma~B.8]
Recall the definition of
$\overline J=\sup \{J:  \zeta ^{2}(J)\sqrt{(\log J)/n}\leq
\overline c  s_{J}\}$. Following the proof of
\citet [Lemma C.6]{chen2021}, using Weyl's inequality (see, e.g.,
\citet [Lemma F.1]{ChenChristensen2017}) together with
\citet [Lemma F.7]{ChenChristensen2017}, we obtain that
$|\widehat s_{ J}-s_{J}|\leq c_{0} s_{J}$ uniformly in
$J\in \mathcal I_{n}$ for some $0<c_{0}<1$ with probability approaching
1 uniformly for $h\in \mathcal H$.%

Proof of (i). By making use of the definition of $\widehat J_{\max}$ given
in (2.11), we obtain uniformly for $h\in \mathcal H$:
\begin{equation*}
\begin{split}\mathrm{P}_{h} (\widehat J_{\max}> \overline J ) &\leq
\mathrm{P}_{h} \biggl(\zeta ^{2}(\overline J)\sqrt{\log (
\overline J)/n}< \frac{3}{2}\widehat s_{\overline J} \biggr)
\\
&\leq
\mathrm{P}_{h} \biggl( \zeta ^{2}(\overline J)\sqrt{\log (
\overline J)/n}< \frac{3}{2}(1+c_{0})s_{
\overline J} \biggr)+o(1).
\end{split}\end{equation*}
The upper bound imposed on the growth of $\overline J$ is determined by
a sufficiently large constant $\overline c>0$ and hence, there exists a
constant $\underline c\geq 3(1+c_{0})/2$ such that
$s_{\overline J}^{-1}\zeta ^{2}(\overline J)\sqrt{\log (\overline J)/n}
\geq \underline c$. Consequently, we obtain
\begin{align*}
\mathrm{P}_{h} (\widehat J_{\max}> \overline J ) &\leq
\mathrm{P}_{h} \biggl(s_{\overline J}^{-1}\zeta
^{2}(\overline J) \sqrt{\log (\overline J)/n}< \frac{3}{2}(1+c_{0})
\biggr)+o(1)=o(1).
\end{align*}
Proof of (ii). From the definition of $J^\circ $ given in
(4.3), we have uniformly for $h\in \mathcal H$:
\begin{align*}
\mathrm{P}_{h} \bigl(J^\circ >\widehat J_{\max}
\bigr) &\leq \mathrm{P}_{h} \bigl(n^{-1}
\sqrt{\log \log n} \widehat J_{
\max}^{2p/d_{x}+1/2}\leq \nu_{\widehat J_{\max}}^{2} \bigr).
\end{align*}
By Assumption 3 there is a constant $c>0$ such that $\nu_J^{2} \leq s_J^{2} /c$ for all $J$. We infer as above for some constant
$0<c_{0}<1$ and uniformly in $J\in \mathcal I_{n}$, that $\nu_J^{2} \leq (1-c_0)^{-1}\widehat s_{J}^{2} $ with probability approaching 1, and hence uniformly for $h\in \mathcal H$:
\begin{align*}
\mathrm{P}_{h} \bigl(J^\circ >\widehat J_{\max}
\bigr) &\leq \mathrm{P}_{h} \bigl((1-c_{0})n^{-1}
\sqrt{\log \log n} \widehat J_{
\max}^{2p/d_{x}+1/2}\leq \widehat
s_{\widehat J_{\max}}^{2} \bigr)+o(1).
\end{align*}
Consider the case $\zeta (J)=\sqrt{J}$. The definition of
$\widehat J_{\max}$ in (2.11) yields uniformly for
$h\in \mathcal H$:
\begin{align*}
\mathrm{P}_{h} \bigl(J^\circ >\widehat J_{\max}
\bigr) &\leq \mathrm{P}_{h} \bigl((1-c_{0})\sqrt{\log \log
n} \widehat J_{\max}^{2p/d_{x}-3/2} \leq (\log \overline J) \bigr)+o(1)
\\
&\leq \mathrm{P}_{h} \biggl((1-c_{0})\widehat
s_{\widehat J_{\max}} \sqrt n\leq \frac{2}{3} \sqrt{\log \overline J} \biggl( \frac{
\log \overline J}{\sqrt{\log \log n}} \biggr)^{1/(2p/d_{x}-3/2)} \biggr)+o(1)
\\
&\leq \mathrm{P}_{h} \biggl( (1-c_{0})^{2}s_{\overline J}
\sqrt{n}\leq \frac{2}{3} \sqrt{\log \overline J} \biggl( \frac{\log \overline J}{
\sqrt{\log \log n}} \biggr)^{1/(2p/d_{x}-3/2)} \biggr)+o(1)
\\
&\leq \mathrm{P}_{h} \biggl( \frac{(1-c_{0})^{2}}{
\overline c} \overline J\leq \frac{2}{3} \biggl( \frac{\log \overline J}{\sqrt{
\log \log n}} \biggr)^{1/(2p/d_{x}-3/2)} \biggr)+o(1),
\end{align*}%
where the last inequality follows from the definition of
$\overline J$, that is,
$s_{\overline J}\geq \overline c^{-1}\overline J \sqrt{\log (
\overline J)/n}$. From Assumption~4(iii), that is,
$p\geq 3d_{x}/4$, we infer
$\mathrm{P}_{h}(J^\circ >\widehat J_{\max})=o(1)$ and, in particular,
$\mathrm{P}_{h}(2J^\circ >\widehat J_{\max})=o(1)$ uniformly for
$h\in \mathcal H$. The proof of $\zeta (J)=J$ follows analogously using
the condition $p\geq 7d_{x}/4$.
\end{proof}
\end{appendix}

\end{document}